\documentclass[12pt]{article}

\usepackage[utf8]{inputenc}   
\usepackage[T1]{fontenc}

\RequirePackage[tt=false, type1=true]{libertine} 
\RequirePackage[varqu]{zi4}

\usepackage{geometry}         
\usepackage{microtype}        

\usepackage{tikz}
\usepackage{standalone}
\usepackage{subfig}

\usepackage{graphicx}%
\usepackage{multirow}%
\usepackage{amsmath,amsfonts}%
\usepackage{amssymb}
\usepackage{xcolor}%
\usepackage{textcomp}%
\usepackage{manyfoot}%
\usepackage{booktabs}%
\usepackage{listings}
\usepackage{tabularx} 
\usepackage{caption}
\usepackage{array}

\usepackage[ruled, vlined,linesnumbered]{algorithm2e}
\usepackage{xspace}
\usepackage{multirow}
\usepackage{pifont}
\usepackage{balance}
\usepackage{comment}
\usepackage{graphics}
\usepackage[normalem]{ulem}
\usepackage{float}

\usepackage{amsthm}

\RequirePackage[libertine]{newtxmath}

\usepackage[bookmarks]{hyperref}
\hypersetup{
	colorlinks=true,
	linkcolor=myblue,
	citecolor=colorcite,
	urlcolor=colorcite
}

\usepackage[nameinlink,capitalize]{cleveref}
\usepackage[sort&compress]{natbib}

\newtheorem{problem}{Problem}
\newtheorem{lemma}{Lemma}
\newtheorem{theorem}{Theorem}

\newtheorem{example}{Example}
\newtheorem{proposition}{Proposition}

\usepackage{xparse}
\usepackage{pifont}
\usepackage{bm}
\usepackage{paralist}
\usepackage{nccmath}
\usepackage{mathtools}

\definecolor{myblue}{RGB}{80,80,160}
\definecolor{mygreen}{RGB}{80,160,80}
\definecolor{mygrey}{HTML}{706968}
\definecolor{mydarkgrey}{HTML}{43110a}
\definecolor{mydarkred}{HTML}{e12800}

\newcommand{\vmark}{\ding{51}\xspace}
\newcommand{\xmark}{\ding{55}\xspace}

\newcommand{\bigO}{\ensuremath{\mathcal{O}}\xspace}
\newcommand{\bigT}{\ensuremath{\Theta}\xspace}

\newcommand{\Exp}{\ensuremath{\mathbb{E}}\xspace}
\newcommand{\Prob}{\ensuremath{\mathbb{P}}\xspace}

\newcommand{\clust}{\ensuremath{\mathcal{C}}\xspace}

\newcommand{\weightEl}[1]{\ensuremath{W_{C_{#1}}(e)}\xspace}

\newcommand{\weightClustEl}[2]{\ensuremath{W_{C_{#1}}(#2)}\xspace}
\newcommand{\weightElEst}[1]{\ensuremath{\widehat{W}_{C_{#1}}(e)}\xspace}
\newcommand{\weightElEstUni}[1]{\ensuremath{\widehat{W}_{C_{#1}}^{\mathtt{U}}(e)}\xspace}
\newcommand{\C}{\ensuremath{C}\xspace}
\newcommand{\F}{\ensuremath{F}\xspace}

\newcommand{\PPSsample}{\ensuremath{\F_{\C_j}}}
\newcommand{\PPSinitSample}{\ensuremath{\F^0_{\C_j}}}

\newcommand{\sEl}[1]{\ensuremath{s(#1)}\xspace}
\newcommand{\sElEst}[1]{\ensuremath{\widehat{s}\,(#1)}\xspace}
\newcommand{\sElEstUni}[1]{\ensuremath{\widehat{s}_{\mathtt{U}}(#1)}\xspace}
\newcommand{\aEl}[1]{\ensuremath{a(#1)}\xspace}
\newcommand{\aElEst}[1]{\ensuremath{\widehat{a}(#1)}\xspace}
\newcommand{\bEl}[1]{\ensuremath{b(#1)}\xspace}
\newcommand{\bElEst}[1]{\ensuremath{\widehat{b}(#1)}\xspace}
\newcommand{\dist}[2]{\ensuremath{d(#1,#2)}\xspace}

\newcommand{\kmeans}{\ensuremath{k}-means\xspace}
\newcommand{\kmed}{\ensuremath{k}-medoids\xspace}

\newcommand{\deucl}{\ensuremath{d_{\mathtt{eucl}}}\xspace}
\newcommand{\dcos}{\ensuremath{d_{\mathtt{cos}}}\xspace}
\newcommand{\dmanh}{\ensuremath{d_{\mathtt{manh}}}\xspace}
\newcommand{\dcanb}{\ensuremath{d_{\mathtt{can}}}\xspace}
\newcommand{\fsbl}{{\texttt{FS}}\xspace}
\newcommand{\simplbl}{{\texttt{SIMPL}}\xspace}
\newcommand{\kvecEx}{\ensuremath{\bm{k}^{\mathsf{exact}}}\xspace}
\newcommand{\kvecPPS}{\ensuremath{\bm{k}^{\PPS}}\xspace}
\newcommand{\kvecUNI}{\ensuremath{\bm{k}^{\UNIbuck}}\xspace}

\newcommand{\globOne}[1]{\ensuremath{\widehat{s}_1(#1)}\xspace}
\newcommand{\globTwo}[1]{\ensuremath{\widehat{s}_2(#1)}\xspace}
\newcommand{\globHyb}[1]{\ensuremath{\bar{s}(#1)}\xspace}

\newcommand{\brdata}{\texttt{BR}\xspace}
\newcommand{\mtdata}{\texttt{MT}\xspace}

\newcommand{\gowdata}{\texttt{GOW}\xspace}
\newcommand{\powdata}{\texttt{PH}\xspace}
\newcommand{\shudata}{\texttt{SHU}\xspace}
\newcommand{\crdata}{\texttt{CR}\xspace}
\newcommand{\iotdata}{\texttt{IOT}\xspace}
\newcommand{\winedata}{\texttt{WI}\xspace}

\newcommand{\Qone}{\textbf{I1}\xspace}
\newcommand{\Qtwo}{\textbf{I2}\xspace}
\newcommand{\Qthr}{\textbf{I3}\xspace}
\newcommand{\Qfou}{\textbf{I4}\xspace}
\newcommand{\BO}[1]{\ensuremath{\mathcal{O}\!\left(#1\right)}\xspace}
\newcommand{\BT}[1]{\ensuremath{\Theta\!\left(#1\right)}}
\newcommand{\BOM}[1]{\ensuremath{\Omega\!\left(#1\right)}}
\newcommand{\clustC}{\ensuremath{\mathcal{C}}\xspace}
\newcommand{\algLocal}{\textsc{silh-pps-all}\xspace} 
\newcommand{\algGlobOne}{\textsc{gl-s}\xspace} 
\newcommand{\algGlobPPS}{\textsc{gl-pps-f}\xspace} 
\newcommand{\algGlobSub}{\textsc{gl-pps-s}\xspace} 
\newcommand{\PPS}{\textsc{pps}\xspace}
\newcommand{\Binomial}{\ensuremath{\text{Bin}}\xspace}

\newcommand{\algGlobUNI}{\textsc{gl-uni-f}\xspace} 
\newcommand{\algGlobUNISub}{\textsc{gl-uni-s}\xspace} 

\newcommand{\PPSbuck}{\textsc{pps}\xspace}
\newcommand{\UNIbuck}{\textsc{uni}\xspace}

\newcommand{\codeRepo}{\url{https://github.com/iliesarpe/ScalableSilhouetteComputation}}

\crefname{problem}{problem}{problems}
\Crefname{problem}{Problem}{Problems}

\definecolor{vir1}{HTML}{440154}
\definecolor{vir2}{HTML}{30678d}
\definecolor{vir3}{HTML}{35b778}
\definecolor{vir4}{HTML}{fde724}

\definecolor{myblue}{HTML}{4278f5}
\definecolor{colorcite}{HTML}{3b4670}
\usetikzlibrary{arrows,decorations.pathmorphing,backgrounds,positioning,fit,petri,calc,positioning,shapes.geometric}
\usetikzlibrary{shapes.geometric}
\usetikzlibrary{shapes, shapes.misc}
\usepackage{comment}
\usetikzlibrary{positioning,chains,fit,shapes,calc}
\usepackage{tkz-berge}
\usepackage{pgfplots}
\pgfplotsset{compat=1.17}

\let\oldnl\nl
\newcommand{\nonl}{\renewcommand{\nl}{\let\nl\oldnl}}

\DeclarePairedDelimiter{\ceil}{\lceil}{\rceil}

\title{Scalable and Distributed Silhouette Approximation}

\author{
	\centering
	\begin{tabular}{cc}
		\parbox[t]{0.45\linewidth}{\centering
			Ilie Sarpe\\[-2pt]
			\url{ilsarpe@kth.se}\\[-2pt]
			KTH Royal Institute of Technology
		} &
		\parbox[t]{0.45\linewidth}{\centering
			Federico Altieri\\[-2pt]
			\url{altierif87@gmail.com}\\[-2pt]
			University of Padova\vspace{7pt}
		} \\ 
		\parbox[t]{0.45\linewidth}{\centering
			Andrea Pietracaprina\\[-2pt]
			\url{capri@dei.unipd.it}\\[-2pt]
			University of Padova
		} &
		\parbox[t]{0.45\linewidth}{\centering
			Geppino Pucci\\[-2pt]
			\url{geppo@dei.unipd.it}\\[-2pt]
			University of Padova\vspace{7pt}
		} \\[2em] 
		\multicolumn{2}{c}{%
			\parbox[t]{0.6\linewidth}{\centering
				Fabio Vandin\\[-2pt]
				\url{fabio.vandin@unipd.it}\\[-2pt]
				University of Padova
			}%
		}
	\end{tabular}
}

\date{}

\begin{document}

\maketitle

\begin{abstract}
The \emph{silhouette} is one of the most widely used measures to assess the quality of a $k$-clustering of a dataset of $n$ elements. 
Its popularity stems from the fact that its evaluation requires no information beyond the clustering assignment. 
In addition, the silhouette is extremely easy to interpret, with a plethora of applications in various domains. 
The silhouette provides a score to measure the quality of a clustering as a whole or for each individual element.
However, the exact computation of the: 
\begin{inparaenum}[($i$)]
	\item silhouette of each element of a dataset; and 
	\item the global silhouette of the entire clustering;
\end{inparaenum}
require $\Theta(n^2)$ distance calculations, under general metrics. 
The quadratic complexity $\Theta(n^2)$ is extremely prohibitive, especially on massive modern datasets. 
Surprisingly, existing approximate methods using $\bigO{(n^2)}$ distance calculations are heuristics and do not offer \emph{provable} and \emph{controllable} guarantees on the quality of their results.

In this work, we introduce the first rigorous and efficient algorithms to estimate: 
\begin{inparaenum}[($i$)]
	\item the (local) silhouette of each element of a dataset; and
	\item the (global) silhouette;
\end{inparaenum}
of any metric $k$-clustering.
Our methods are based on sampling, performing $\bigO(nk\varepsilon^{-2}\ln (nk/\delta))$
distance computations, and providing estimates with additive error $\BO{\varepsilon}$ with probability at least $1-\delta$. 
That is, the user-defined parameters $\varepsilon$ and $\delta$ in $(0,1)$ control the trade-off between accuracy and efficiency. 
Furthermore, we introduce a scalable and \emph{distributed} design of our methods for the MapReduce and Massively Parallel Computing (MPC) frameworks.
Our distributed algorithms require a constant number of rounds and sublinear local memory, under practical parameter setting. Finally, we perform extensive experiments to compare our methods with state-of-the-art approaches. 
The results show that our new techniques yield the best trade-off between accuracy and efficiency for both local and global silhouette estimation. 
In addition, our methods scale efficiently to massive datasets for which an exact computation of the silhouette is not practical. 
\end{abstract}

\section{Introduction}
\label{sec:intro}

Clustering is of fundamental importance for data analysis, with ubiquitous applications in various areas including pattern recognition~\citep{Xu2015Survey}, bioinformatics and biomedicine~\citep{Ikotun2023Surv}, and data management~\citep{AggarwalR13,Karypis1998metis}. 
Broadly, clustering  a dataset requires grouping elements with high similarity, and separating dissimilar elements~\citep{kleinberg2002impossibility}.
The clustering of a dataset is often viewed as an optimization problem, where the objective varies according to the desired properties of the optimal clustering. 
A large number of algorithms have been designed to cluster a dataset according to various objectives~\citep{HennigMMR15}. 
In particular, recent research for clustering methods has its focus on the design of rigorous and scalable approaches to cluster massive datasets~\citep{AwasthiB15,MalkomesKCWM15,BargerF16,CeccarelloFPPV17,CeccarelloPP19,MazzettoPP19}.

The vast number of clustering objectives and methods, introduce the need for rigorous methods to identify ``good quality'' clusterings of a dataset~\citep{Schubert2023StopElbow, Liu2010Internal}.
\emph{Clustering evaluation} (or simply \emph{validation}) assesses the quality of a given clustering relying on an evaluation measure. 
An evaluation measure is either \emph{external} (i.e., \emph{supervised}) or \emph{internal} (i.e., \emph{unsupervised}). 
Supervised measures rely on external knowledge about the data, e.g., ground truth labels---rarely available \citep{TanSK06}. 
In contrast, internal measures evaluate the quality of a clustering relying uniquely on the dataset. 
Due to their broad applicability, internal measures are widely used in practice to identify high-quality clusterings~\citep{Hassan2024Metrics}.

The \emph{silhouette} is one of the most commonly used internal measures for clustering evaluation~\citep{Rousseeuw1987,Hassan2024Metrics}.
The silhouette $\sEl{e}$ of an element $e$ assigned to a cluster
$C$ is defined as the ratio $(\bEl{e}-\aEl{e})/\max\{\aEl{e},\allowbreak\bEl{e}\}$ where $\aEl{e}$ is the average distance of $e$ from the other elements of $C$,
and $\bEl{e}$ is the minimum average distance of $e$ from the elements of a cluster $C'$ different from $C$. 
In other words, $\sEl{e}$ provides a simple to interpret score between -1 and 1, which captures the quality of the assignment of $e$ to its group $C$.
That is, if $e$ is closer (on average) to the elements in its cluster
$C$ than to the elements in the ``closest'' cluster $C'\neq C$, then $\sEl{e}$ will be close to 1. In contrast, $\sEl{e}$ becomes negative when the element $e$ is closer to the elements in $C'$. 
As an example, consider elements $e_{13}$ and $e_{14}$ in~\Cref{subfig:dataset} (where the distance is the Euclidean distance): $i$) the silhouette $\sEl{e_{13}}$ is positive since $e_{13}$ is closer to the elements of its cluster ($e_{11}$ and $e_{12}$) compared to its average distance from the elements of other clusters; $ii$) the silhouette $\sEl{e_{14}}$ is close to -1 since $e_{14}$ is very close to all elements of cluster $C_3$ and distant from the other elements of its cluster $C_1$.
Hence, the silhouette $\sEl{e}$ of an element $e$ (or \emph{local silhouette}) provides an interpretable and reliable score 
to assess the quality of the assignment of each element $e$ to its cluster~\citep{Lovmar2005SilhLoc,Dinh2019SilhPlots}.

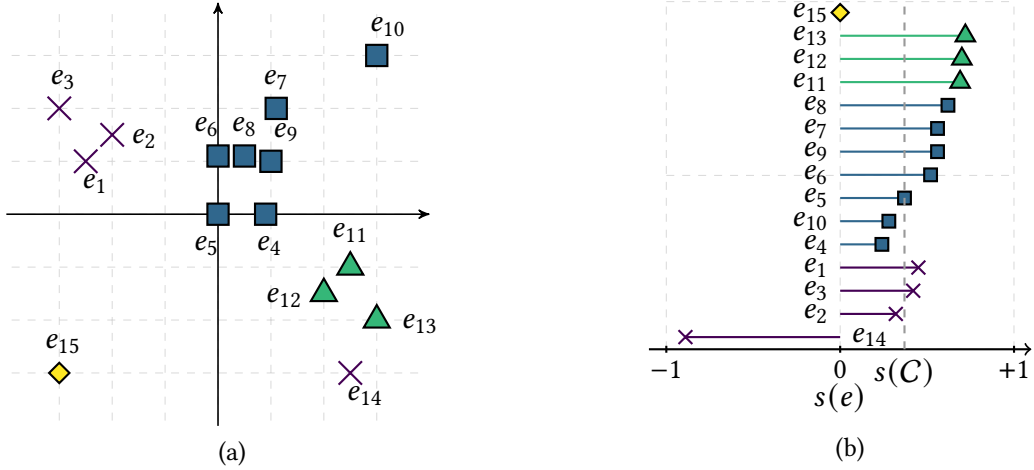
\begin{figure}[t]
	\centering
	\subfloat[\label{subfig:dataset}]{
		\begin{minipage}{0.49\textwidth}
			\centering
			\begin{tikzpicture}[scale=0.7, thick,
		green place/.style={circle,draw=mydarkred!80,fill=white!60,thick, inner sep=0pt,minimum size=6.5mm},
		square/.style={regular polygon,regular polygon sides=4}
		every node/.style={draw,circle},
		fsnode/.style={fill=mygrey},
		pre/.style={->,shorten <=0.pt,>=stealth',semithick},
		post/.style={<-,shorten >=0.pt,>=stealth',semithick},
		ssnode/.style={fill=mygreen},
		every fit/.style={ellipse,draw,inner sep=-2pt,text width=2cm},
		]
		\draw[help lines, color=gray!30, dashed] (-3.9,-3.9) grid (3.9,3.9);
		\draw[pre] (-4,0)--(4,0) node[right]{};
		\draw[pre] (0,-4)--(0,4) node[above]{};
		
		\node[draw, cross out, color=vir1, minimum size=4pt, inner sep=4pt, thick, label={[xshift=-3.5mm, yshift=-3mm]right:{$e_1$}}] at (-2.5,1) {};
		\node[draw, cross out,color=vir1, minimum size=4pt, inner sep=4pt, thick, label={[xshift=-0.5mm, yshift=-0.5mm]right:{$e_2$}}] at (-2,1.5) {};
		\node[draw, cross out,color=vir1, minimum size=4pt, inner sep=4pt, thick,  label={[xshift=0.5mm, yshift=-0.5mm]above:{$e_3$}}] at (-3,2) {};
		\node[draw, cross out,color=vir1, minimum size=4pt, inner sep=4pt, thick, label={[xshift=-3.5mm, yshift=-3mm]right:{$e_{14}$}}] at (2.5,-3) {};
		
		\node[draw, rectangle, minimum size=3pt, inner sep=4pt, thick, fill=vir2,  label={[xshift=1.5mm, yshift=-0.5mm]above:{$e_{10}$}}] at (3,3) {};
		\node[draw, rectangle, minimum size=3pt, inner sep=4pt, thick, fill=vir2, label={[xshift=-1.5mm, yshift=0.2mm]below:{$e_{5}$}}] at (0,0) {};
		\node[draw, rectangle, minimum size=3pt, inner sep=4pt, thick, fill=vir2, label={[xshift=2mm, yshift=-0.5mm]above:{$e_9$}}] at (1,1) {};
		\node[draw, rectangle, minimum size=3pt, inner sep=4pt, thick, fill=vir2, label={[xshift=0mm, yshift=-0.5mm]above:{$e_8$}}] at (0.5,1.1) {};
		\node[draw, rectangle, minimum size=3pt, inner sep=4pt, thick, fill=vir2, label={[xshift=-0.mm, yshift=-0.5mm]above:{$e_7$}}] at (1.1,2) {};
		\node[draw, rectangle, minimum size=3pt, inner sep=4pt, thick, fill=vir2, label={[xshift=-1.5mm, yshift=-0.5mm]above:{$e_6$}}] at (0,1.1) {};
		\node[draw, rectangle, minimum size=3pt, inner sep=4pt, thick, fill=vir2, label={[xshift=0.5mm, yshift=0.2mm]below:{$e_4$}}] at (0.9,0) {};
		
		\node[draw, regular polygon, fill=vir3, regular polygon sides=3, minimum size=2pt, inner sep=2pt, thick, label={[xshift=0mm, yshift=-0.5mm]above:{$e_{11}$}}] at (2.5,-1) {};
		\node[draw, regular polygon, fill=vir3,regular polygon sides=3, minimum size=3pt, inner sep=2pt, thick, label={[xshift=0mm, yshift=-0.5mm]left:{$e_{12}$}}] at (2,-1.5) {};
		\node[draw, regular polygon, fill=vir3, regular polygon sides=3, minimum size=3pt, inner sep=2pt, thick,  label={[xshift=0.5mm, yshift=-0.5mm]right:{$e_{13}$}}] at (3,-2) {};
		
		\node[draw,  diamond, minimum size=2.8pt, fill=vir4, inner sep=2pt, thick,  label={[xshift=0.5mm, yshift=-0.5mm]above:{$e_{15}$}}] at (-3,-3) {};
		
	\end{tikzpicture}
		\end{minipage}
	}
	\subfloat[\label{subfig:silhPlot}]{
		\begin{minipage}{0.49\textwidth}
			\centering
			\begin{tikzpicture}[scale=2.3, thick,
		green place/.style={circle,draw=mydarkred!80,fill=white!60,thick, inner sep=0pt,minimum size=6.5mm},
		square/.style={regular polygon,regular polygon sides=4}
		every node/.style={draw,circle},
		fsnode/.style={fill=mygrey},
		pre/.style={->,shorten <=0.pt,>=stealth',semithick},
		post/.style={<-,shorten >=0.pt,>=stealth',semithick},
		ssnode/.style={fill=mygreen},
		every fit/.style={ellipse,draw,inner sep=-2pt,text width=2cm},
		]
  \draw[help lines, color=gray!30, dashed] (-1,-1) grid (1,1);

	\draw[->, thick] (-1.1,-1) -- (1.1,-1) node[right] {};
	
	\draw (-1,-1.02) -- (-1,-0.98) node[below=1pt] {\small $-1$};
	\draw (0,-1.02) -- (0,-0.98) node[below=1pt] {\small $0$};
	\draw ( 1,-1.02) -- ( 1,-0.98) node[below=1pt] {\small $+1$};
	
	\pgfmathdeclarefunction{getxval}{1}{%
	\pgfmathparse{{-0.89, 0.32, 0.42, 0.45, 0.24, 0.28, 0.37, 0.52, 0.56, 0.56, 0.62, 0.69, 0.7, 0.72, 0.0}[#1]}%
	}
	
	\def\colors{{red, blue, green, orange, purple, cyan, magenta, brown, olive, teal, violet, lime, pink, gray, black}}
	\pgfmathdeclarefunction{getxlab}{1}{%
		\pgfmathparse{{"$e_{14}$","$e_2$","$e_3$","$e_1$","$e_4$","$e_{10}$",
				"$e_5$","$e_6$","$e_9$","$e_7$","$e_8$","$e_{11}$",
				"$e_{12}$","$e_{13}$","$e_{15}$"}[#1]}%
	}
	
\foreach \i in {0,...,14} {
	\pgfmathsetmacro{\x}{getxval(\i)}
	\pgfmathsetmacro{\y}{-1 + 0.07 + 2*\i/15}
	\pgfmathsetmacro{\xlab}{getxlab(\i)}
	\ifnum\i<4
	\def\col{vir1};
	\else\ifnum\i<11
	\def\col{vir2};
	\else\ifnum\i<14
	\def\col{vir3}
	\else
	\def\col{vir4};
	\fi\fi\fi
	\draw[thick, color=\col] (0,\y) -- (\x,\y);

	\ifnum\i<1
	 \node[anchor=west] at (0, \y) {\xlab};
	 \else
	 \node[anchor=east] at (0, \y) {\xlab};
	 \fi
	 
	\ifnum\i<4
	\def\col{vir1};
	\node[draw, cross out, thick, inner sep=2.2pt, color=\col] at (\x, \y) {};
	\else\ifnum\i<11
	\def\col{vir2};
	\node[draw, rectangle, thick, fill=\col, inner sep=2.4pt] at (\x, \y) {};
	\else\ifnum\i<14
	\def\col{vir3}
	\node[draw, regular polygon, regular polygon sides=3, fill=\col, inner sep=1.5pt] at (\x, \y) {};
	\else
	\def\col{vir4};
	\node[diamond, draw, fill=\col, inner sep=1.7pt] at (\x, \y) {};
	\fi\fi\fi
	}
	\draw[dashed, gray!80] (0.36993,-1) -- (0.36993,1.02);
	
	\node[anchor=north] at (0.36993,-0.985) {$s(\mathcal{C})$};
	\node[anchor=north] at (0, -1.13) {$s(e)$};
	\useasboundingbox (-1.1,-1.1) rectangle (1.1,1.1);
		
	\end{tikzpicture}
		\end{minipage}
	}
	\caption{(a): dataset $V=\{e_1\dots,e_n\} $ for $n=15$ with elements in the Euclidean plane $\mathbb{R}^2$. Different shapes represent $k=4$ different clusters, that is $\clust = \{\C_1=\{e_1,\dots,e_3,e_{14}\}, \C_2=\{e_4,\dots,e_{10}\},\C_3=\{e_{11},e_{12},e_{13}\}, \C_4=\{e_{15}\}\}$. (b): silhouette $s(e)$ of the elements $e\in V$ (where the distance is the Euclidean distance): values are grouped by clusters and sorted. The dashed line represents the value of the \emph{silhouette of the clustering} $\clust$, marked with $\sEl{\clust}$. Such representation is known as the \emph{silhouette plot}~\citep{Rousseeuw1987}.
	}\label{fig:workingex}
\end{figure}

Furthermore, the \emph{average} silhouette of \emph{all} elements, known as the silhouette of a
clustering (or \emph{global silhouette}), has a plethora of applications in many important areas~\citep{hossain2007gdclust, ngefficient, sellam2016blaeu, wiwie2015comparing}, and is described in widely used data mining textbooks~\citep{TanSK06,han2011data}. 
As an example, the silhouette of a clustering can be used to identify a good value for the number $k$ of clusters computed by popular clustering algorithms, e.g., \kmeans, or \kmed~\citep{Dinh2019SilhPlots,Lenssen2024MedSilh}. 

For general distance metrics, the exact
computation of the silhouette for a $k$-clustering of $n$ elements
requires $\Theta(n^2)$ distance calculations, irrespective of $k$.
The quadratic barrier constitutes a key challenge that prevents silhouette-based analyses on modern massive datasets. 
Surprisingly, for such a general setting, there are no methods to efficiently approximate the
silhouette with \emph{provably high accuracy} guarantees and breaking the $\Theta(n^2)$ barrier. 
In contrast, various scalable techniques have been proposed to efficiently cluster large datasets with rigorous guarantees, especially in distributed environments~\citep{ene2011fast,feldman2011unified,bahmani2012scalable,balcan2013distributed,MalkomesKCWM15}.

In this work, we introduce the \emph{first randomized approximation algorithms} with rigorous (probabilistic) guarantees on the quality of their estimates for two key problems: $i$) estimating the silhouette of each element $e$ of a given clustering; $ii$) estimating the global silhouette of a clustering. 
Our algorithms feature a memory-efficient distributed implementation in popular frameworks such as \emph{MapReduce}~\citep{DeanG08,LeskovecRU14,PietracaprinaPRSU12},
and the \emph{Massively Parallel Computing} (MPC) model~\citep{Im2023MPC}, in analogy with state-of-the-art scalable clustering algorithms.
Similar to other randomized algorithms, the guarantees of our methods are controlled by two easy to interpret parameters $\varepsilon$ and $\delta$, controlling the trade-off between efficiency and accuracy. 
In particular, $\varepsilon$ controls the (absolute) error between the output of our methods and the actual values being estimated, and $\delta$ controls the error probability (all our results hold with probability at least $1-\delta$).

For the estimation of the silhouette of each element $e$, we rely on the Probability Proportional to Sampling (\PPSbuck) 
technique by~\citet{Chechik2015}, which can be used to obtain highly accurate estimates of the average distance between $e$ and the elements of a given cluster. 
For a $k$-clustering, our method performs $\bigO(nk\varepsilon^{-2} \log(nk/\delta))$  
distance computations, producing estimates $\sElEst{e}$ such that $|\sEl{e}-\sElEst{e}| \le f(\varepsilon)$ with $f(\varepsilon) = \BO{\varepsilon}$ with probability at least $1-\delta$, simultaneously for \emph{all} elements $e$. 
Note the substantial improvement over the $\Theta(n^2)$ distance computations of an exact method, at the expense of a small estimation error.

For the global silhouette estimation, we design three estimators, including a simple algorithm that computes the average silhouette using a small random subset of elements, and two slightly more complex algorithms relying on the \PPSbuck approach. 
We show in our analysis that all our methods perform $o(n^2)$ distance computations, and report an estimate with bounded error from the actual global silhouette, with controlled error probability. 
Finally, our experiments show that all our methods are extremely efficient and scalable in practice, especially on large datasets, where exact approaches are infeasible and existing techniques perform poorly.

\smallskip
\textbf{Our contributions.} To summarize our contributions are as follows.

\begin{itemize}
	\item We present the first rigorous sampling-based approximation algorithms for two key problems involving the silhouette: 
	$i$) estimating the silhouette of all elements of a clustered dataset; 
	$ii$) estimating the global silhouette of a clustering. 
	Our methods require the distance between elements to be defined through a metric.
	\item We prove that all our methods require a significantly smaller number of distance computations than $\Theta(n^2)$, reporting high-quality estimates with (absolute) error controlled by parameter $\varepsilon$,  with probability controlled by parameter $\delta$. 
	\item We provide a distributed MapReduce design of our methods and discuss how our methods can be ported to the popular MPC framework. 
	Our distributed methods use sublinear local memory (at each worker) and a constant number of rounds, two features that are highly desirable when processing massive datasets.
	\item We perform an extensive experimental evaluation on medium- and large-size datasets with multiple goals: 
	$i$) assess the trade-off between the quality of the estimates and the number of distance computations of our methods for estimating the global silhouette of the entire clustering; 
	$ii$) evaluate the accuracy of our method for the estimation of the local silhouette for all elements; 
	$iii$) assess the performance of our distributed methods; 
	$iv$) illustrate applications of our new methods, leveraging both the global and local silhouette. 
\end{itemize}

\textbf{Contributions over the conference version.}
A preliminary version of this work was presented at the SIAM International Conference on Data Mining, 2021, authored by a subset of the authors~\citep{Altieri2021Silh}. This work feartures substantial modifications and additions to that preliminary work, as summarized below.
\begin{itemize}
	\item We propose two novel estimators for the global silhouette coefficient, in addition to the estimator originally devised in~\citep{Altieri2021Silh}. 
	Our new methods require considerably fewer samples compared with the method appeared in~\citep{Altieri2021Silh}, yielding a reduction in the time complexity for global estimation.
	\item We rigorously study the problem of local silhouette estimation, providing:
	\begin{inparaenum}[$i$)]
		\item a sound algorithm to solve the local silhouette estimation problem, based on techniques appeared in our previous version;
		\item rigorous bounds on the sample complexity, fixing existing issues with previous proofs in the literature~\citep{Chechik2015};
		\item we prove that the guarantees of our new method based on \PPSbuck sampling cannot be matched by a simpler and more intuitive uniform sampling-based approach---supporting the design of our method.
	\end{inparaenum}
	\item We include all proofs omitted from \citep{Altieri2021Silh} for space constraints.
	\item We show how the distributed MapReduce implementation of our methods can be ported to the popular MPC framework, which is widely adopted for large-scale data analysis.
	\item We perform a new and extensive experimental evaluation, over medium- and large-size real datasets available online, showing the effectiveness of our  approaches in various settings.
\end{itemize}

The paper is structured as
follows. In~\Cref{sec:prelims}, we introduce the necessary preliminary notions. In \Cref{sec:related}, we discuss relevant previous work, with a focus on the silhouette approximation. 
In \Cref{sec:methods}, we introduce our methods and their analysis.  
In \Cref{sec:exp}, we present the results of our extensive experimental evaluation.
Lastly, \Cref{sec:conclusions} summarizes our concluding remarks.

The code and necessary material to reproduce our results is available at the following repository~\codeRepo. 

\section{Preliminaries} 
\label{sec:prelims}
Let $U$ be a metric space equipped with a distance function $d(\cdot,\cdot)$,
and let $V = \{e_1, \dots e_n\} \subseteq U$ be a dataset of $n$
distinct elements from $U$. 
Let $\clust = \{\C_1, \dots \C_k\}$ be a
\emph{$k$-clustering}\footnote{We will drop $k$ when clear from the context.} of $V$, that is $\clust$ consist of a \emph{partition} of $V$ into $k$ disjoint and non-empty subsets called \emph{clusters}. 
Formally, the \textit{silhouette $\sEl{e}$ of an element} $e \in V$ belonging to a cluster $C\in \clust$, as introduced by~\citet{Rousseeuw1987} is defined as,
\begin{equation}\label{eq:localSilh}
	\sEl{e} \doteq \frac{\bEl{e}-\aEl{e}}{\max\{\aEl{e}, \bEl{e}\}} 
\end{equation}
where 
\begin{equation}\label{eq:AandBterms}
	\aEl{e} = \frac{\sum_{e' \in C} \dist{e}{e'}}{|C|-1} \enspace,
	\qquad
	\bEl{e} = \min_{\substack{C_j \in \clust, \\ C_j \neq C}}\frac{\sum_{e' \in C_j}\dist{e}{e'}}{|C_j|}\enspace.
\end{equation}
That is, $\aEl{e}$ is the average distance of element $e\in V$ from all other
elements of its cluster $C$; 
while $\bEl{e}$ is the minimum average distance of $e\in C$ from
the elements in some cluster $C_j\in \clust$ different from $C$.
Note that by \Cref{eq:localSilh} it holds that $\sEl{e} \in [-1,1]\subseteq \mathbb{R}$. If element $e$ is the only element in its cluster $C$ such that $C=\{e\}$, we define  $\sEl{e}=0$, as proposed by~\citet{Rousseeuw1987}. As an example, consider $e_{15}$ in~\Cref{fig:workingex}, then $s(e_{15})=0$.

The silhouette $\sEl{e}$ of an element $e\in E$ is a \emph{local measure} that evaluates the quality of a clustering $\clust$ with respect to the assignment of $e$ to its cluster $C$. 
Under a good assignment, the average distance of $e$ to the points in $C$ is much smaller than the average distance to the points of a different cluster, thus from \Cref{eq:localSilh} it holds $\sEl{e} \approx 1$. 
Hence, high-values of $\sEl{e}$ denote a ``good'' assignment of element $e$ to its cluster $C$. 
Conversely, if  $e$ is, on average, much closer to the elements of a different cluster $C_j\neq  C$ then by~\Cref{eq:localSilh} it holds $\sEl{e}\approx -1$, i.e., a ``bad'' cluster assignment for $e$.
The individual $\sEl{e}$'s are often represented through
the \emph{silhouette plot}~\citep{Rousseeuw1987}, where the abscissae range in $[-1,1]$ (possible silhouette values) and the ordinate values  correspond to the distinct elements, ordered by clusters. 

\begin{example}\label{ex:silhplots}
	Consider the dataset from~\Cref{subfig:dataset} and $d(\cdot,\cdot)$ being the Euclidean distance.\footnote{Recall for $\bm{x},\bm{y}\in \mathbb{R}^z$ the Euclidean distance $d(\bm{x},\bm{y})$ is $\sqrt{\sum_i (\bm{x}_i-\bm{y}_i)^2}$ where $\bm{x}_i, i\in[z]$ denotes the $i$-th coordinate of $\bm{x}$.} \Cref{subfig:silhPlot} shows its silhouette plot, with points grouped by cluster. We observe that, for $e_{11}\in \C_3$ it holds $s(e_{11}) > 0 $ since $e_{11}$ is much closer to $e_{12}$ and $e_{13}$ than to the points of any other cluster different from $\C_3$. 
	In contrast, consider $e_{14}$, then $b(e_{14}) = (d(e_{11},e_{14}) +d(e_{12},e_{14}) + d(e_{13},e_{14}))/3 < (d(e_{1},e_{14}) + d(e_{2},e_{14}) + d(e_{3},e_{14}))/3 = a(e_{14})$. Therefore $s(e_{14}) < 0$, capturing the intuition that assigning $e_{14}$ to $\C_3$ instead of $\C_1$ leads to a better clustering.
\end{example}

The \emph{silhouette of a clustering} $\clust$ evaluates the quality of a clustering $\clust$ as the \emph{average} silhouette of all the elements $e\in V$. Formally, the silhouette $\sEl{\clust}$ of a clustering $\clust$ is
\begin{equation}
	\label{eq:avgShil}
	\sEl{\clust} \doteq \frac{1}{n} \sum_{i = 1}^{n}\sEl{e_i}\enspace.
\end{equation}

Therefore,
$\sEl{\clust}$ provides a \emph{global measure} of the quality of the whole clustering $\clust$, where a value close to 1 denotes a high quality of the clustering $\clustC$.

\begin{example}\label{ex:globalSilh}
	Consider~\Cref{subfig:silhPlot} under the setting of~\Cref{ex:silhplots}: we observe that for the global silhouette of the clustering $\clust$ it holds $s(\clust) \approx 0.36$. That is, on average, most points are well assigned over $\clust$, as we can observe from~\Cref{subfig:dataset}. 
	Clearly, a positive value for $s(\clust)$, does not imply that all points have $s(e)>0$, e.g., $s(e_{14})<0$.
\end{example}

In practice, the score $s(\clust)$ is widely used to identify a good value for $k$ (e.g.,~\citep{Lleti2004Silh}). 
That is, analysts often compare various clusterings $\{\clust_{k_i}\}_{i\ge 1}$ obtained with different values of $k_i$, selecting the best value of $k$ as the one maximizing $s(\clust_{k_i})$, for $i\ge 1$. 
In fact, using the silhouette to select a good value of $k$ for a $k$-clustering is highly effective and reliable, in contrast to other popular methods, such as the elbow method, which often lead to poor results~\citep{Schubert2023StopElbow}.

It is clear from~\Cref{eq:localSilh} that the straightforward computation of 
$\sEl{e}$ for each element $e\in V$ and the subsequent computation of $\sEl{\clust}$ require $\Theta(n^2)$ distance calculations.
We remark that such a quadratic complexity is \emph{extremely resource demanding} even when dealing with medium-sized
datasets, i.e., impractical for modern massive data volumes.

Finally, we wish to recall that for an algorithm employing random sampling, its \emph{sample complexity} refers to the total number of samples generated throughout its execution \citep{Mitzenmacher2017Prob}. 

\section{Related work}
\label{sec:related}

Since the introduction of the silhouette metric  by~\citet{Rousseeuw1987}, its use has been ubiquitous for clustering applications \citep{MoulaviJCZS14,XiongL15,TomasiniEBM16, Arbelaitz2013ComparingIndices},
including the use of the silhouette to obtain new clustering objectives for both \kmed~\citep{Kaufmann1987Clust,Lenssen2024MedSilh} and \kmeans~\citep{Lai2024Silhkmeans}. 
However, as observed before, obtaining the exact silhouette is impractical
for large datasets, due to the $\Theta(n^2)$ distance computations required. 
Therefore, various methods
have been proposed to reduce the time complexity, 
or to
simplify the silhouette calculation under certain assumptions.
In this section, we review to the most relevant methods for silhouette approximation, 
which mostly apply to center-based clusterings. 
That is, clusterings where each cluster is associated with a distinguished center from the metric space (e.g., the centroids in the case of $k$-means clusterings).

\citet{VanderLaan2003Simpl} introduced a variant of the silhouette $\sEl{e}$, where 
the term $\bEl{e}$ is computed as the distance from the closest center, other than the cluster $e$ belongs to (while $\aEl{e}$ is computed according to~\Cref{eq:AandBterms}). 
Following such idea, \citet{Hruschka2004}  introduced the \emph{simplified silhouette}. 
The simplified silhouette for an element $e$ assigned to cluster $C$ evaluates $\aEl{e}$ as the distance between $e$ and the center of $C$, 
and $\bEl{e}$ as the distance between $e$ and the closest center of a cluster different from $C$. 
The complexity of computing the simplified silhouette reduces to $\bigO(nk)$. 
\citet{Hruschka2004}
and \citet{Wang2017}
show empirically that the simplified silhouette can be an effective evaluation measure for clusterings returned by Lloyd's
algorithm \citep{Lloyd82}. 
However, there is no evidence of the effectiveness of the simplified silhouette
for other clusterings (e.g, based on arbitrary distance functions). 
In addition, as shown in our experimental evaluation in~\Cref{sec:exp}, the difference between the values of the silhouette (\Cref{eq:localSilh}) and the simplified silhouette, can be arbitrary large, making the simplified silhouette unreliable for rigorous silhouette-based analyses.

A different heuristic for improving the computation of the exact silhouette
of center-based clusterings under Euclidean distances was introduced by \citet{FrahlingS08}.
For each element $e$ of a cluster $C$, the term $\aEl{e}$ is computed according to its definition (\Cref{eq:AandBterms}), while for $\bEl{e}$ 
the heuristic first determines the average distance $\bar{d}_{e,C'}$ between
$e$ and the elements of the cluster $C' \neq C$, whose centroid is
closest to $e$. If, for  any other cluster $C'' \not\in
\{C,C'\}$, it holds that the distance between $e$ and the centroid of $C''$ is greater than or equal to $\bar{d}_{e,C'}$, then
it sets $\bEl{e}=\bar{d}_{e,C'}$, otherwise it computes
$\bEl{e}$ according to~\Cref{eq:AandBterms}. 
For Euclidean distances, it can be shown that in the former case $\bar{d}_{e,C'}$ coincides exactly with the
value of $\bEl{e}$ defined in~\Cref{eq:AandBterms}.
The worst case complexity remains quadratic,
but the heuristic likely performs less than $\BT{n^2}$ distance computations in practice. 

The Apache Spark programming framework provides
optimized methods for computing the silhouette
of a clustering under $d$-dimensional \emph{squared Euclidean} distances and under
one formulation of cosine distance.\footnote{\url{https://spark.apache.org/}}  
For such specific distances,
simple algebra suffices to show that the pre-computation of a limited number of values dependent on the coordinates of the points in each of the $k$
clusters yields a parallelizable
algorithm performing $\bigO(nkd)$ distance computations. However, the above optimization does not apply to \emph{arbitrary} distance metrics, as of interest in our work. 

The methods presented in this paper apply to any metric distance and provide provably accurate silhouette estimations. Our methods rely on the use a \emph{Probability Proportional to Size} (\PPSbuck)
sampling scheme, where each element is sampled with a probability proportional to a certain ``size'' measure.
The use of \PPSbuck sampling has been pioneered in the context of center-based clustering algorithms, distance query processing and centrality estimation in graphs~\citep{feldman2011unified,balcan2013distributed,Chechik2015,cohen2018clustering}. 
To the best of our knowledge, prior to the conference version of the current work~\citep{Altieri2021Silh} the use of \PPSbuck for efficient clustering evaluation had not been explored.

\section{Methods} 
\label{sec:methods}

In this section we present our methods. In \Cref{subsec:localPPS} we focus on the estimation of the silhouette of all individual elements, while in \Cref{subsec:globalSilh} we focus on the estimation of the silhouette of the whole clustering. Finally, \Cref{sec:MR} describes scalable distributed implementations of our methods on the MapReduce and  MPC frameworks.

\subsection{Estimating the silhouette of each element}\label{subsec:localPPS}

In this section we address the following problem,
\begin{problem}\label{probl:LocalEst}
	Given a dataset $V$ from an arbitrary metric space, a $k$-clustering $\clust$, and two parameters $\varepsilon, \delta \in (0,1)$, obtain $\sElEst{V} = \{(e, \sElEst{e} : e \in V)\}$ such that, with probability at least $1-\delta$, 
	$|\sElEst{e} - \sEl{e}| =\BO{\varepsilon}$ simultaneously for each element $e\in V$, minimizing the number of distance computations performed.
\end{problem}

Our main result, addressing~\Cref{probl:LocalEst} is the  following.

\begin{theorem}\label{theo:localAlg}
	Consider a dataset $V$, a $k$-clustering $\clust$, and two parameters $\varepsilon, \delta \in (0,1)$. 
	There exists a randomized algorithm based on \PPSbuck sampling with sample complexity $\bigO(k\varepsilon^{-2}\log(nk/\delta))$ that, with probability $1-\delta$:
	\begin{compactenum}
		\item returns a set $\sElEst{V} = \{(e, \sElEst{e} : e \in V)\}$ such that, for all $e\in V$, $|\sElEst{e} -\sEl{e}|\le 4\varepsilon/(1-\varepsilon)$;
		\item performs $\bigO(nk\varepsilon^{-2} \log(nk/\delta))$ distance computations.
	\end{compactenum}
\end{theorem}

\Cref{theo:localAlg} shows that we can obtain a very accurate probabilistic approximation of each silhouette value $\sEl{e}, e\in V$. 
For practical values of $k\ll n$, and constant $\varepsilon$ and $\delta$, the result of~\Cref{theo:localAlg} breaks
the quadratic barrier of $\Theta(n^2)$ distance computations, enabling the efficient and practical computation of silhouette values over large and massive datasets. 
Note that the $\BO{\varepsilon^{-2}}$ factor in the sample complexity and in the number of distance computations from~\Cref{theo:localAlg}, might become large for small $\varepsilon$.
While a quadratic dependence on $\varepsilon$ cannot be avoided in a worst-case analysis~\citep{Mitzenmacher2017Prob}, in practice, 
even large values of $\varepsilon$ (e.g., close to 1) yield highly accurate estimates for \emph{all} values in $\sElEst{V}$.

In this section, we also provide evidence that \PPSbuck is crucial to achieve our result. Namely, we argue that, under the same sample complexity of~\Cref{theo:localAlg}, a more intuitive \emph{uniform sampling} approach performs arbitrarily poorly.

\begin{theorem}\label{theo:unifails}
	{
		\sloppy
		There exists a large dataset $V$, and a $2$-clustering ${\clust}$ of\hspace{0.1cm}$V$, such that the same algorithm as the one of~\hspace{0.1cm}\Cref{theo:localAlg}, modified replacing \PPSbuck sampling with \emph{uniform sampling}, returns a set of estimates $\sElEst{V} = \{(e, \sElEst{e} : e \in V)\}$, with arbitrarily large probability, such that:
		$|\sElEst{e'}-\sEl{e'}| \ge 1$ for at least $n/2-1$ elements $e'\in V$.
	}
\end{theorem}

\Cref{theo:unifails} indicates that uniform sampling cannot guarantee an (arbitrarily small) estimation error $\BO{\varepsilon}$ using a sublinear sample size, in contrast with \Cref{theo:localAlg}.

In the rest of this section, we describe our main algorithm (\Cref{sec:alg_desc}) and prove the above two key results
(in \Cref{sec:analysis,subsubsec:uniFails}, respectively).

\subsubsection{Algorithm}
\label{sec:alg_desc}
Consider the setting of~\Cref{probl:LocalEst}. For each element $e \in V$ and cluster $C_j \in
\mathcal{C}$, let
\begin{equation}\label{eq:weights}
	\weightEl{j} = \sum_{e' \in \C_j}d(e, e')\enspace.
\end{equation}
Given the definition of $\weightEl{j}$ in \Cref{eq:weights}, for an element $e$ of a cluster $C$
we can express the quantities
$a(e)$ and $b(e)$ in the definition of the silhouette $s(e)$ (from Equation~\eqref{eq:AandBterms}) as
\[
a(e) = \frac{\weightEl{}}{|C|-1} \qquad \text{and} \qquad b(e) = \min_{C_j \neq C}\frac{\weightEl{j}}{|C_j|} \enspace. 
\]

Building on the above formulation, we rely on the accurate estimation of the terms
$\weightEl{j}$ to approximate 
the silhouette $s(e)$ for each element $e\in V$. 

A first attempt to obtain an estimate of $\weightEl{j}$ would be to sample \emph{uniformly} at random a set of $t\ge 1$ points from each cluster $C_j\in \clust$ to be used to estimate $\weightEl{j}$. 
Unfortunately, in~\Cref{subsubsec:uniFails} we show an instance of $V$ and $\clust$ where such an approach fails with large probability for the sample complexity guarantees of~\Cref{theo:localAlg} over a large fraction of points (see~\Cref{theo:unifails}). 
That is, with a uniform sampling approach we \emph{cannot} obtain the accurate approximation of the silhouette values as captured in~\Cref{probl:LocalEst} with the guarantees provided by~\Cref{theo:localAlg}. 

Therefore, to achieve an accurate estimation of the weights $\weightEl{j}$, we exploit the more complex but effective  \emph{Probability Proportional to Size} (\PPS) sampling strategy proposed by
\citet{Chechik2015}.  
The \PPS sampling strategy allows us to collect a suitable \emph{small} sample $\F_{\C_j}$ from each cluster $\C_j$. 
The sample $\F_{\C_j}$ can be then used to approximate the value of $\weightEl{j}$  of every element $e \in V$, within user-defined error bounds and with high
probability, using only \emph{weighted} distances between $e$ and elements in
$\F_{\C_j}$. 
The selection of the elements of $\F_{\C_j}$ is \emph{not} uniformly at random. Instead, it involves
carefully designed sampling probabilities that favor the selection of distant elements from a suitable ``central'' element of each cluster $C_j$ for $j\in[k]$.

\begin{algorithm}[t]
	\KwIn{Clustering $\mathcal{C}=\{C_1,\dots,C_k\}$ of $V=\{e_1,\dots,e_n\}$ and $\varepsilon, \delta \in (0,1), t\ge 2$.}
	\KwOut{$\hat{s}(V) =\{(e_1, \sElEst{e_1},\dots, (e_n,\sElEst{e_n})\}$.}
	\nonl $\triangleright$ Phase 1: \emph{PPS sampling}\\
	$n \leftarrow |V|$; $t\gets \ceil{(c\varepsilon^{-2})\ln{(5nk/\delta)}}$\;
	\For{\upshape\textbf{each} cluster $\C_j\in \clust$} {
		\lIf{$|C_j| \leq t$}{$\F_{\C_j} \leftarrow \C_j$}
		\Else {
			$\PPSinitSample \leftarrow$ Poisson sampling of $\C_j$
			with probability $(2/|\C_j|)\ln (5k/\delta)$\;
			\lFor{\upshape\textbf{each} $\bar{e} \in \PPSinitSample$} {
				$\weightClustEl{j}{\bar{e}} \leftarrow \sum_{e' \in \C_j} d(\bar{e},e')$\label{line:initSampleW}}
			\For{\upshape\textbf{each} $e \in \C_j$\label{line:forLoopPPS}} {
				$\gamma_e \leftarrow \max \{d(e,\bar{e})/\weightClustEl{j}{\bar{e}} : \bar{e} \in \PPSinitSample\}$\;
				\label{line:num8}
				$\gamma_e \leftarrow \max \{1/|C_j|,\gamma_e\}$\;
				$p_e \leftarrow \min\{1, t\gamma_e\}$\;
				\label{line:num10}
			}
			$\F_{\C_j} \leftarrow$ Poisson sampling of $\C_j$ with probabilities 
			$\{p_e : e \in \C_j\}$\;
		}
	}
	\nonl $\triangleright$ Phase 2: \emph{Silhouette estimation} \\
	\For{\upshape\textbf{each} $e \in V$\label{line:forLoopAllElements}} {
		Let $e$ belong to cluster $\C$\;
		\For{\upshape\textbf{each} cluster $\C_j$} {
			$\weightElEst{j} = \sum_{e' \in \F_{\C_j}} d(e, e')/p_{e'}$\;
		}
		$\aElEst{e} \leftarrow \weightElEst{}/(|\C|-1)$\;
		$\bElEst{e} \leftarrow \min \{\weightElEst{j}/|\C_j| : \C_j \neq \C\}$\; 
		$\sElEst{e} \leftarrow 
		(\bElEst{e}-\aElEst{e})/(\max\{\aElEst{e}, \bElEst{e}\})$\;
	}
	\KwRet{$\sElEst{V}=\{(e,\sElEst{e}) : e \in V\}$;}
	\caption{\algLocal}\label{code:algorithm}
\end{algorithm}

At high level, our algorithm
\algLocal (\Cref{code:algorithm}), consists of two \emph{phases}:  
\emph{1)} For each cluster
$\C_j \in\clust$, the algorithm computes a sample $\F_{\C_j}$ of
\emph{expected} size $t= \ceil{(c\varepsilon^{-2})\ln{(5nk/\delta)}}$
for a suitably small constant $c$; 
\emph{2)} each sample $\F_{\C_j}$ is then used to estimate the weights $\weightEl{j}$ for $j\in [k]$, and such weights are combined to obtain $\sElEst{e}$, i.e., the estimates of $s(e)$ for each $e\in V$. More in detail.
\smallskip

\textbf{Phase 1.} 
Each cluster $\C_j, j\in [k]$ is
processed independently.  
If $|\C_j| \leq t$, then $\F_{\C_j}$ is set
to $\C_j$. 
Otherwise, we perform \emph{Poisson sampling} over $\C_j$ where each element $e \in \C_j$ is included in $\F_{\C_j}$ independently with
a suitable probability $p_e$. 
For $e \in \C_j$, probability $p_e$ is determined as follows:

\begin{compactitem}
	\item First, we select an initial sample $\PPSinitSample$, by 
	Poisson sampling. That is, each $e \in \C_j$ is included
	in $\PPSinitSample$ independently with fixed probability 
	$(2/|C_j|)\ln (5k/\delta)$. 
	Each sample $\PPSinitSample$ will contain,
	with sufficiently high probability, an aforementioned ``central''
	element of $C_j$.\footnote{A \emph{central} element, is an element close to a majority of the elements of $\C_j$, see the proof of~\Cref{lem:mainerror} for a formal definition.}
	\item We then compute the value $
	\weightClustEl{j}{\bar{e}}$ for each sampled element $\bar{e} \in \PPSinitSample$. Then, for each $e \in C_j$ we set 
	$p_e=\min\{1, t\gamma_e \}$,
	where $\gamma_e$ is computed as 
	\begin{equation*}
		\gamma_e = \max \left\{\underbrace{\frac{1}{|C_j|}}_{T_1}, \underbrace{\max_{\bar{e} \in \PPSinitSample}\frac{d(e,\bar{e})}{\weightClustEl{j}{\bar{e}}}}_{T_2}\right\}\enspace.
	\end{equation*}
	For $e\in \C_j$,  $\gamma_e$ is the maximum between two terms $T_1$ and $T_2$. Where, $T_1$ corresponds to the probability of selecting $e$ uniformly among all elements of $\C_j$, ensuring that all points have a sufficiently large probability of being sampled; while $T_2$ corresponds to the maximum relative contribution of $e$ to the weights $\weightClustEl{j}{\bar{e}}$ of elements $\bar{e}\in \PPSinitSample$. 
	The intuition is that a higher relative contribution of element $e\in \C_j$ to $\weightClustEl{j}{\bar{e}}$ implies that element $e\in\C_j$ is extremely important for the accurate estimation of $\weightClustEl{j}{\bar{e}}$ with $\bar{e}\in \PPSinitSample$---captured by the value $\gamma_e$ through $T_2$, which yields a high sampling probability $p_e$ for $e\in \C_j$. 
\end{compactitem}

\smallskip

\textbf{Phase 2.} 
We now have access to a sample $\F_{\C_j}$ obtained through Poisson sampling with individual probabilities $p_e$ for each $e\in V$ (computed over \textbf{Phase 1}), over each cluster $C_j$.
Hence, we use each $\F_{\C_j}$ to obtain the estimates
\[
\weightElEst{j} = \sum_{e' \in \F_{\C_j}}\frac{d(e, e')}{p_{e'}} \enspace,
\]
which are accurate estimators of $W_{C_j}(e)$,
as we prove in our analysis (Section~\ref{sec:analysis}). We then use values $\weightElEst{j}, j\in[k], e\in V$ 
to compute $\sElEst{e}, e\in V$ as follows.
Let $e \in V$ with $e\in \C, \C\in \clust$, we compute estimates
$\aElEst{e} = \tfrac{\weightElEst{}}{(|\C|-1)}$ and 
$\bElEst{e} = \min_{\C_j \neq \C} \left\{ \tfrac{\weightElEst{j}}{|\C_j|}\right\}$. 
Finally we obtain the estimates $\sElEst{e}$ of the silhouette of each element $e\in V$ as
\begin{equation}\label{eq:localSilhEst}
	\sElEst{e} = 
	\frac{\bElEst{e}-\aElEst{e}}{\max\left\{\aElEst{e}, \bElEst{e}\right\}} \enspace.
\end{equation}

\Cref{code:algorithm} (\algLocal) reports all details of the resulting procedure that combines the phases described above.

\subsubsection{Proof of \protect\Cref{theo:localAlg}}
\label{sec:analysis}
We prove~\Cref{theo:localAlg} in two steps. 
First, we show that, with probability at least $1 - 3\delta/5$,
for each element $e\in V$ the value
$\sElEst{e}$ computed by \algLocal
approximates the actual value $s(e)$ within a small absolute error expressed as a function $f(\varepsilon)=\Theta(\varepsilon)$, with $\varepsilon\in (0,1)$. Second, we upper bound the sample complexity and the number of distance
computations with probability at least $1 - 2\delta/5$.

Our key ingredient is the next lemma that provides a probabilistic
upper bound on the \emph{relative} approximation error of the estimates $\weightElEst{j}$
with respect to the actual values $\weightEl{j}$ for an arbitrary element $e\in V$, and $j\in[k]$.

\begin{lemma} \label{lem:mainerror}
	Let $c=18$.
	If $t = \ceil{(3c/\varepsilon^2)\ln{(5nk/\delta)}}$, then with probability at least $1 -3\delta/5$, for every element
	$e\in V$ and every cluster $\C_j$, the estimate $\weightElEst{j}$ computed by~\algLocal is such that
	\[
	\left|\frac{\weightElEst{j} - \weightEl{j}}{\weightEl{j}} \right| 
	\leq \varepsilon\enspace.
	\]
\end{lemma}
\begin{proof}
	Our proof expands the poorly sketched argument in the proof of~\cite[Lemma 12]{Chechik2015}. Consider
	an arbitrary cluster $C_j$. If $|C_j| \leq t$ the statement follows trivially, since $\weightElEst{j}= \weightEl{j}$. Thus, we focus on the case $|C_j| > t$.
	For an element
	$e \in C_j$, let $m(e)$ denote the median of the distances from $e$ to all other
	elements of $C_j$, that is, the distance between $e$ and the $\ceil{|C_j|/2+1}$-th closest element to $e$ in $\C_j$. Let $e_{\rm min} = \arg\min_{e' \in C_j} m(e')$.  
	The element $e$ is called \emph{well positioned} if
	$m(e) \leq 2m(e_{\rm min})$. 
	Using the triangle inequality, it is easy to see that the $\lceil |C_j|/2 \rceil$ 
	elements of $C_j$ closest to $e_{\rm min}$ are well positioned. Hence, the initial
	random sample $\PPSinitSample$ will contain a well positioned 
	element with probability
	at least 
	\[
	1-\left(1-{2 \over |C_j|} \ln \left(\frac{5k}{\delta}\right) \right)^{|C_j| \over 2} \geq 1-{\delta \over 5k}\enspace,
	\]
	since $(1+q/x)^x \le \exp(q)$, for $x\ge1, |q|\le x$, and 
	$|C_j|/2> t/2 > \ln (5k/\delta)$.
	As proved in \citet[Lemma 9]{Chechik2015}, if $\PPSinitSample$ contains a well positioned element, then for $c=18$ we have that for all $e \in C_j$, the values $\gamma_e$ computed by \Cref{code:algorithm} satisfy
	\begin{equation}
		\label{eq:gamma}
		\gamma_e \geq {1 \over c} \max_{e' \in C_j} {d(e,e') \over \weightClustEl{j}{e'}}\enspace.
	\end{equation}
	To obtain a worst-case bound, 
	we assume that for each $e \in C_j$ it holds that $p_e = t \gamma_e$ where $p_e \le1$.
	In fact, when it holds $p_e=1$, 
	the contribution of element $e$ to the weights $\widehat{W}_{\C_j}(e')$ of all elements $e'\in V$ is always computed \emph{exactly}, implying no error over the estimates $\widehat{W}_{\C_j}(e')$. 
	From \Cref{eq:gamma}, we have that $p_e \geq (t/c) \max_{e' \in C_j} (d(e,e')/\weightClustEl{j}{e'})$.
	Fix an arbitrary element $e \in C_j$, and for every $e' \in C_j$ let $X_e({e'})$ be the random variable taking value $d(e,e')/p_{e'}$, with probability $p_{e'}$, and 0 otherwise. We have that
	\[
	\weightElEst{j}= \sum_{e' \in C_j} X_e({e'})\enspace.
	\]
	It is immediate to see that $\mathbb{E}[\weightElEst{j}] = \weightEl{j}$ by the linearity of expectation and the definition of variables $X_e({e'})$. 
	Let $\tau_e = \weightEl{j}/(t/c)$ and, for every $e' \in C_j$, let $Y_e({e'}) = X_e({e'})/\tau_e$ be a normalized variable associated to each $X_e({e'})$. Note that $Y_e({e'}) \le (d(e,e')/\weightEl{j})/\allowbreak(\max_{e''\in \C_j}\{d(e'',e')/\weightClustEl{j}{e''}\}) \le 1$ and therefore $ Y_e({e'})\in [0,1]$, hence 
	\[
	\mu_e \doteq
	\mathbb{E}\left[\sum_{e' \in C_j} Y_e({e'})\right] = {\weightEl{j} \over \tau_e} = {t \over c} \enspace.
	\]
	The above discussion implies that
	\begin{eqnarray*}
		\Prob
		\left(\weightElEst{j} \leq (1+\varepsilon) \weightEl{j}\right)
		= 
		\Prob
		\left(\sum_{e' \in C_j} Y_e({e'}) \leq (1+\varepsilon) {t \over c} \right) 
		= 
		\Prob
		\left(\sum_{e' \in C_j} Y_e({e'}) \leq (1+\varepsilon) \mu_e \right) \enspace.
	\end{eqnarray*}
	Since variables $Y_e({e'})$'s are independent random variables,
	and each $Y_e({e'})$ takes either value 
	$(1/\tau_e) d(e,e')/p_{e'} \leq 1$ with probability $p_{e'}$, or 0 otherwise, by the application of~\Cref{th:CHbounds} (see \Cref{appsec:tools})
	we conclude that
	\[
	\Prob\,\left(\sum_{e' \in C_j} Y_e({e'}) \geq (1+\varepsilon) \mu_e \right)
	\leq \mbox{exp}\left(-{\varepsilon^2 \over 3} {t \over c}\right)
	\leq 
	{\delta \over 5nk} \enspace ,
	\]
	where the last inequality follows from the choice of $t$ as from the statement.
	Therefore, we have that 
	\[
	\Prob\!\left(\weightElEst{j} \leq (1+\varepsilon) \weightEl{j}\right)
	\geq 1-{\delta \over 5nk}.
	\]
	A symmetrical argument can be used to also show that
	\[
	\Prob\!\left(\weightElEst{j} \geq (1-\varepsilon) \weightEl{j}\right)
	\geq 1-{\delta \over 5nk}\enspace.
	\]
	
	By the union bound, the probability that there exists a cluster $C_j$ such that the initial sample $\PPSinitSample$ does \emph{not} contain a well positioned element is bounded by $k\delta/(5k)=\delta/5$. Now, by conditioning on the event that for all clusters $C_j$ for $j\in[k]$ the initial sample $\PPSinitSample$
	contains a well positioned element, we proceed with a second application of the union bound to prove our claim. 
	That is, the probability that there exists an element $e\in V$ and a
	cluster $\C_j$ for which $|(\weightElEst{j} - \weightEl{j})/\weightEl{j}|>\varepsilon$ is at most $2nk \delta/(5nk)
	= 2\delta/5$. 
	Therefore, the probability that
	$|(\weightElEst{j} - \weightEl{j})/\weightEl{j}|\leq \varepsilon$ for every $e \in V$ is at least $1-\delta/5-2\delta/5 = 1-3\delta/5$, yielding the statement and concluding the proof.
\end{proof}

Let  \emph{$E$} be the event ``the relative error guarantees stated in~\Cref{lem:mainerror} hold for every element $e \in V$ and every
cluster $C_j$''. 
From now on our analysis conditions on the realization of event $E$, even when not explicitly stated.
Consider an arbitrary element $e\in V$, and let $\sElEst{e}$ be the estimate of the silhouette $s(e)$ computed by
\algLocal. 
The following key technical lemma,
establishes a bound on the absolute
error of the estimates for each element $e\in V$. That is, we show $|\sElEst{e} - \sEl{e}| \le f(\varepsilon)$ as follows.

\begin{lemma}
	\label{lem:singlesilh}
	Under event $E$, \algLocal outputs estimates $\sElEst{e}$ for each element $e\in V$ such that, 
	\[| \sElEst{e}-\sEl{e}| \le \frac{4\varepsilon}{1-\varepsilon}\enspace.
	\]
\end{lemma}

\begin{proof}
	First note that if $\aEl{e}= 0$ (that is, $\{e\}$ is a singleton cluster), then it follows that 
	$\aElEst{e}=0$ and $\sEl{e}=\sElEst{e}=0$, and the claim trivially holds. 
	We thus consider the case $\aEl{e}> 0$ (also, observe that $\bEl{e}>0$ by definition). 
	When event $E$ holds, the relative errors of terms: ($I$)
	$|(\aElEst{e}-\aEl{e})/\aEl{e}|$; and ($II$) $|(\bElEst{e}-\bEl{e})/\bEl{e}|$, 
	are both upper bounded by $\varepsilon>0$.
	That is, the bound on term ($I$) follows immediately from the definition of $\aElEst{e}$ and
	the relative error bound for $\weightElEst{j}$ provided by~\Cref{lem:mainerror}.
	Concerning the bound on the error of term ($II$), consider an arbitrary cluster $C_j$ and let $b_{C_j}(e) =
	\weightEl{j}/|C_j|$ and $\widehat{b}_{C_j}(e) = \weightElEst{j}/|C_j|$. 
	Similar to the argument for term ($I$), it is easy to show that 
	$|(\widehat{b}_{C_j}(e)-b_{C_j}(e))/b_{C_j}(e)| \leq \varepsilon$.
	Recall that $\bEl{e} = \min_{C_j \neq C} b_{C_j}(e)$ and that
	$\bElEst{e} = \min_{C_j \neq C} \widehat{b}_{C_j}(e)$. 
	Suppose that $b(e)=b_{C'}(e)$ and $\bElEst{e}=\widehat{b}_{C''}(e)$,
	for some, possibly different, clusters $C'$ and $C''$. We have:
	\[
	(1-\varepsilon)b(e) = (1-\varepsilon)b_{C'}(e) \leq (1-\varepsilon)b_{C''}(e) 
	\leq \widehat{b}_{C''}(e) = \widehat{b}(e)
	\]
	and
	\[
	\useshortskip
	\bElEst{e} = \widehat{b}_{C''}(e) 
	\leq \widehat{b}_{C'}(e) \leq (1+\varepsilon)b_{C'}(e)= (1+\varepsilon)\bEl{e} \enspace,
	\]
	and the desired bound follows.
	
	Now we establish a bound on the relative error for the the denominator of the estimator $\sElEst{e}$ (see~\Cref{eq:localSilhEst}). Define $M(e)
	\doteq \max\{\aEl{e}, \bEl{e}\}$ and $\widehat{M}(e) \doteq
	\max\{\aElEst{e},\bElEst{e}\}$. We show that
	$|(\widehat{M}(e)-M(e))/M(e)| \leq \varepsilon$.
	Suppose that $M(e) = \aEl{e}$, hence $\aEl{e} \geq \bEl{e}$ (the case $M(e) = \bEl{e}$ is analogous).  If $\widehat{M}(e) = \aElEst{e}$, the bound
	is immediate.
	Instead, if $\widehat{M}(e) = \bElEst{e}$, hence $\bElEst{e} \geq\aElEst{e}$, the bound follows since
	\begin{equation*}
		\widehat{M}(e) = \bElEst{e}  \leq 
		(1+\varepsilon)\bEl{e}
		\leq  (1+\varepsilon)\aEl{e} = (1+\varepsilon)M(e)\enspace,
	\end{equation*}
	and 
	\begin{equation*}
		\useshortskip
		\widehat{M}(e) = \bElEst{e} \geq  \aElEst{e}
		\geq  (1-\varepsilon)\aEl{e} 
		=  (1-\varepsilon)M(e)\enspace.
	\end{equation*}
	We are now ready to obtain a bound on the absolute difference
	$|\sElEst{e} - \sEl{e})|$. Using the relative error bounds established 
	above, we have that
	\[
	\sElEst{e} = \frac{\bElEst{e}-\aElEst{e}}{\widehat{M}(e)} 
	\leq
	\frac{(1+\varepsilon)\bEl{e}-(1-\varepsilon)\aEl{e}}{(1-\varepsilon)M(e)}=
	\frac{(1+\varepsilon)\bEl{e}-(1-\varepsilon)\aEl{e}}{(1-\varepsilon)\aEl{e}}
	\enspace.
	\]
	Simple algebraic manipulations show that
	\[
	\frac{(1+\varepsilon)\bEl{e}-(1-\varepsilon)\aEl{e}}{(1-\varepsilon)\aEl{e}}
	=
	\sEl{e}+\frac{2\varepsilon}{1-\varepsilon}
	\left(\sEl{e}+1\right) 
	\leq \sEl{e}+\frac{4\varepsilon}{1-\varepsilon} \enspace,
	\]
	where the last inequality follows since $\sEl{e} \leq 1$.
	With analogous calculations, we obtain that
	\begin{align*}
		\sElEst{e} &=  
		\frac{\bElEst{e}-\aElEst{e}}{\widehat{M}(e)} 
		\geq 
		\frac{(1-\varepsilon)\bEl{e}-(1+\varepsilon)\aEl{e}}{(1+\varepsilon)M(e)} 
		=
		\frac{(1-\varepsilon)\bEl{e}-(1+\varepsilon)\aEl{e}}{(1+\varepsilon)\aEl{e}} \\
		& =
		\sEl{e}-\frac{2\varepsilon}{1+\varepsilon}
		\left(\sEl{e}+1\right) 
		\geq 
		\sEl{e}-\frac{4\varepsilon}{1+\varepsilon}
		\geq 
		\sEl{e}-\frac{4\varepsilon}{1-\varepsilon}\enspace,
	\end{align*}
	which concludes the proof. 
\end{proof}

We observe that \Cref{lem:singlesilh} provides rigorous probabilistic guarantees (holding with probability at least $1-3\delta/5$) on the approximation error of each element in the set $\sElEst{V}$ returned by \algLocal, holding simultaneously for all the elements in $V$.  Note that, for example, when $\varepsilon\le 1/2$ then $4\varepsilon/(1-\varepsilon) \le 8 \varepsilon = \BO{\varepsilon}$.

To complete the proof of~\Cref{theo:localAlg}, we now bound the sample complexity and the number of distance computations performed by \Cref{code:algorithm}, as follows. 
\begin{lemma}\label{lem:samples}
	With probability at least $1 -2\delta/5$, over the execution of~\algLocal,
	for every cluster $\C_j$ with $|C_j| > t$ it holds that,
	\begin{eqnarray*}
		\useshortskip
		\quad |F^0_{C_j}|  = \BO{\ln\left(\frac{k}{\delta}\right)} \quad  \text{ and } \quad 
		|F_{C_j}| = \BO{\varepsilon^{-2}\ln\left(\frac{nk}{\delta}\right)} 
		\enspace.
	\end{eqnarray*}
\end{lemma}
\begin{proof}
	Consider a cluster $|C_j|$ with $|C_j|>t$. $|F^0_{C_j}|$ corresponds to a binomial random variable $\Binomial(|\C_j|, 2/|\C_j|\ln(5k/\delta))$, 
	implying that $\Exp[|F^0_{C_j}|] = 2\ln(5k/\delta))$.
	Therefore, a simple application of the Chernoff bound suffices to show that $|F^0_{C_j}| = \BO{\ln(k/\delta)} $
	with probability at least $1-\delta/(5k)$, hence the bound
	holds for all $F^0_{C_j}$'s with probability at least $1-\delta/5$. 
	Let us now bound the size of $|F_{C_j}|$ (our result expands on the argument of~\citet{Chechik2015}). $|F_{C_j}|$ is a Poisson binomial random variable 
	with expectation
	\[
	\useshortskip
	\mathbb{E}\left[ |F_{C_j}| \right]
	= \sum_{e \in C_j} p_e 
	\leq 
	t \sum_{e \in C_j} \gamma_e \enspace.
	\]
	Observe that
	\begin{equation*}
		\hspace*{3em} 
		\begin{aligned}
			\useshortskip
			\sum_{e \in C_j} \gamma_e &=
			\sum_{e \in C_j} 
			\max \left\{\frac{1}{|C_j|},\max_{\bar{e} \in \PPSinitSample}\frac{d(e,\bar{e})}{\weightClustEl{j}{\bar{e}}} \right\} &&\qquad \text{(by definition of $\gamma_e, e\in V$)}\\
			& \leq 1+\sum_{e \in C_j} \max_{\bar{e} \in \PPSinitSample}\frac{d(e,\bar{e})}{\weightClustEl{j}{\bar{e}}} &&\qquad \text{(for $A,B\ge 0$ then $\max\{A,B\} \le A+B$)}\\
			& \leq 1+\sum_{e \in C_j} \sum_{\bar{e} \in \PPSinitSample} \frac{d(e,\bar{e})}{\weightClustEl{j}{\bar{e}}} &&\qquad \text{($\max\{A_i\} \le \sum A_i, A_i\ge 0$)} \\
			& = 1+\sum_{\bar{e} \in \PPSinitSample} \sum_{e \in C_j} \frac{d(e,\bar{e})}{\weightClustEl{j}{\bar{e}}} &&\qquad \text{(swapping finite and bounded sums)}\\
			& = 1+|\PPSinitSample| &&\qquad \text{(definition of $\weightClustEl{j}{\bar{e}}$)} \enspace.
		\end{aligned}
	\end{equation*}
	Now consider the following cases: 
	\begin{inparaenum}[(I)]
		\item\label{caseproof:C1} $|\PPSinitSample| = 1$ consists of a \emph{single well-positioned element};
		\item\label{caseproof:C2} $|\PPSinitSample| = |C_j|$, i.e., the initial sample size consists of \emph{all} elements of $C_j$;
		\item\label{caseproof:C3} $1<|\PPSinitSample|<|\C_j|$, accounting for the remaining cases.
	\end{inparaenum}
	Clearly, by the above analysis, case (\ref{caseproof:C1}) yields $\sum_e{\gamma_e} \le 2$. 
	Now let us define $\bar{\gamma}_e = \max\{1/|\C_j|,\max_{e'\in C_j} d(e,e')/\weightEl{j}\}$, that is $\bar{\gamma}_e$ is computed through the actual weights $\weightEl{j}$ over all elements of $\C_j$. 
	For case (\ref{caseproof:C2}) when $|\PPSinitSample| =|C_j|$ the algorithm~\algLocal relies on the values $\bar{\gamma}_e$. Let $\gamma_e, e\in \C_j$ be the coefficients obtained when $\PPSinitSample$ contains a single well-positioned element, then by~\Cref{eq:gamma} there exists a constant $c$ such that
	\[
	\sum_{e\in \C_j} \bar{\gamma}_e \le c \sum_{e\in \C_j} \gamma_e \le 2c \enspace.
	\]
	Therefore even when $|\PPSinitSample| = |C_j|$ we have that $\sum_e\gamma_e \le \bigO(1)$. 
	Note that case (\ref{caseproof:C2}) implies that even in the highly unlikely case of $|\PPSinitSample| = |C_j|$, the final sample $\F_{\C_j}$ will be such that $\mathbb{E}[|\F_{\C_j}|] = \bigO(t)$. 
	Finally, consider case~(\ref{caseproof:C3}), when $1 < |\PPSinitSample| < |\C_j|$, and let $\gamma_e$ be the coefficients obtained by~\algLocal with $\PPSinitSample$. 
	By contradiction assume that $\sum_e\gamma_e > \sum_e\bar{\gamma}_e$, for $\bar{\gamma}_e$ defined as in case~(\ref{caseproof:C2}), which implies that $\sum_e \left[\max\left\{\tfrac{1}{|C_j|}, \max_{\bar{e} \in \PPSinitSample}\frac{d(e,\bar{e})}{\weightClustEl{j}{\bar{e}}}\right\} - \max\left\{\tfrac{1}{|C_j|}, \max_{\bar{e} \in \C_j}\frac{d(e,\bar{e})}{\weightClustEl{j}{\bar{e}}}\right\}\right] > 0$. 
	It is simple to observe that none of the terms in the summation can be larger than $0$ since $\PPSinitSample\subset \C_j$, yielding the contradiction. 
	Hence, we have that for each set $\PPSinitSample, |\PPSinitSample| \in (1,|\C_j|)$ it holds that $\sum_e \gamma_e \le \sum_e\bar{\gamma}_e \le \BO{1}$, and in turn $\mathbb{E}[|\F_{\C_j}|] = \bigO(t)$.
	Combining the above three cases we have that any (random) choice of set $\PPSinitSample$ will yield $\mathbb{E}[|\F_{\C_j}|] = \bigO(t)$.
	Applying the Chernoff bound, we obtain that  $|F_{C_j}| = \BO{t}$ with probability $1-\delta/(5k)$ by the choice of $t$ as from~\Cref{lem:mainerror}, proving our statement. 
\end{proof}

Observe that the number of distance computations performed by~\algLocal is bounded by
\begin{equation}\label{eq:distanceComp}
	\useshortskip
	\sum_{j=1}^k |C_j| (|\PPSinitSample|+|\PPSsample|)\enspace, 
\end{equation}
where for ease of notation $|\PPSinitSample| = 0$ when $|C_j| \leq t$. 
By combining~\Cref{eq:distanceComp}, together with~\Cref{lem:samples}, we obtain the bound on the sample size and the bound on the number of distance computations claimed in~\Cref{theo:localAlg}, which
hold with probability at least $1-2\delta/5$---yielding \Cref{theo:localAlg} by \Cref{lem:mainerror,lem:singlesilh}.

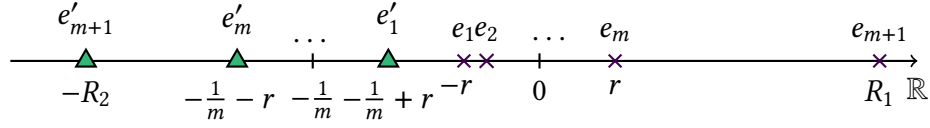
\begin{figure}[t]
	\centering
	\begin{tikzpicture}[scale=1, thick,
		green place/.style={circle,draw=mydarkred!80,fill=white!60,thick, inner sep=0pt,minimum size=6.5mm},
		square/.style={regular polygon,regular polygon sides=4}
		every node/.style={draw,circle},
		fsnode/.style={fill=mygrey},
		pre/.style={->,shorten <=0.pt,>=stealth',semithick},
		post/.style={<-,shorten >=0.pt,>=stealth',semithick},
		ssnode/.style={fill=mygreen},
		every fit/.style={ellipse,draw,inner sep=-2pt,text width=2cm},
		]
  \draw[->] (-6,0) -- (6,0) node[below,yshift=-2pt] {$\mathbb{R}$};

  \node[draw, cross out, thick, inner sep=2pt, color=vir1] (R1) at (5.5, 0) {} node[below=0.5pt of R1] {$R_1$} node[above=0.5pt of R1] {$e_{m+1}$};
  
  
  \draw (1,0.1) -- (1,-0.1) node[below] {$0$};
  \node[draw, cross out, thick, inner sep=2pt, color=vir1] (C1M) at (2, 0) {} node[below=0.5pt of C1M] {$r$} node[above=0.5pt of C1M] {$e_{m}$};

  \node[draw, cross out, thick, inner sep=2pt, color=vir1] (C11) at (0, 0) {} node[below=0.5pt of C11,xshift=-0.1cm] {$-r$} node[above=0.5pt of C11] {$e_{1}$};
  \node[draw, cross out, thick, inner sep=2.2pt, color=vir1] (C12) at (0.3, 0) {} node[below=0.5pt of C12] {} node[above=0.5pt of C12] {$e_{2}$};
  \node[right=0.35 of C12, yshift=0.3cm] {$\dots$};

 \node[draw, regular polygon, regular polygon sides=3, fill=vir3, inner sep=1.5pt] (R2) at (-5, 0) {} node[below=0.5pt of R2] {$-R_2$} node[above=0.5pt of R2] {$e'_{m+1}$};

  \draw (-2,0.1) -- (-2,-0.1) node[below] (C2) {$-\tfrac{1}{m}$};
  \node[draw, regular polygon, regular polygon sides=3, fill=vir3, inner sep=1.5pt] (P21) at (-1, 0) {} node[below=0.5pt of P21] {$-\tfrac{1}{m}+r$} node[above=0.5pt of P21] {$e_1'$};
  
  \node[draw, regular polygon, regular polygon sides=3, fill=vir3, inner sep=1.5pt] (P22) at (-3, 0) {} node[below=0.5pt of P22,xshift=-0.1cm] {$-\tfrac{1}{m}-r$} node[above=0.5pt of P22] {$e_m'$};

  \node[above=0.2 of C2] {$\dots$};
  
  
  
  
  
\end{tikzpicture}
	\caption{Instance used to prove~\Cref{theo:unifails} (see~\Cref{subsubsec:uniFails}). $C_{1}=\{e_1,\dots,e_{m+1}\}$ and $C_{2}=\{e_1',\dots,e_{m+1}'\}$. 
	}
	\label{fig:instance}
\end{figure}

\subsubsection{Proof of \protect\Cref{theo:unifails}}\label{subsubsec:uniFails}
We now show that simple Poisson sampling over each cluster $\C_j, j\in[k]$, with an expected sample size equal to~\algLocal, is \emph{not} sufficient to obtain an accurate estimate of weights~$\weightEl{j}, j\in [k], e\in V$. 
That is, we will provide an instance of~\Cref{probl:LocalEst} where a \emph{uniform} sampling technique with a sample complexity of $\bigT(tk)$, i.e., roughly the same sample size of~\algLocal, outputs estimates of $\sEl{e}, e\in V$ with constant error $\Omega(1)$ for a large fraction of elements in $V$, proving~\Cref{theo:unifails}. Our result shows that non-uniform sampling (e.g., \PPS sampling) is necessary to get the rigorous approximation guarantees provided by our algorithm~\algLocal.

\textbf{Building the instance.} Let $V=\{e_1,\dots,e_{m+1},e_1',\dots,e_{m}',e_{m+1}' \}$, $|V|=2m+2$ for fixed and very large $m$ and $\clust=\{\C_1,\C_2\}$ such that $\C_1= \{e_1,\dots, e_{m+1}\}, \C_2= \{e_1',\dots, e_{m+1}'\}$ and fix $d$ to be the Euclidean distance over $\mathbb{R}\supset V$.\footnote{Recall for $a,b\in \mathbb{R}$ the Euclidean distance corresponds to $d(a,b) = |a-b|$, denoted also as \emph{length}.}
For ease of presentation, and without loss of generality, we assume that the elements of each cluster are ordered according to their indexing, that is $e_{i}< e_j$ (resp. $e_i'> e_j'$) if it holds $1\le i< j\le m+1$. 
The elements of $\C_1$ are placed as follows: ($i$) the $m$ elements $e_i, 1\le i \le m$, span an interval of length $2r$,  $r=o(1/m), r >0$ centered at the origin (i.e., $d(e_i,0) \le r$);
($ii$) $e_{m_1+1} \ge  0$ is such that $d(e_{m_1+1},0)\ge R_1$, for some large value $R_1$ (that  we will specify later). 
The structure of $\C_2$ is similar to the structure of $\C_1$: the $m$ elements $e_i', 1\le i \le m$ are within an interval of length $r$ centered at $-1/m$, and $e_{m+1}'<0$ at distance $R_2>0, R_2=\Theta(m)$ from $0$. 
Furthermore, let the points $e_i$ (resp. $e_i'$) for $i\in [m]$ be equally spaced over the intervals containing them, that is $d(e_i,e_{i+1})$ (resp. $d(e_i',e_{i+1}')$) equals $\beta r$ with $ \beta \doteq 2/(m-1)$, for $i\in [m-1]$. We display the geometry of the resulting dataset in~\Cref{fig:instance}. Finally let $R_1\ge L'(R_2+1)$ for arbitrary large $L'$.

\begin{proposition}\label{prop:exSilhBalls1D}
	Consider dataset $V$ and its clustering $\clustC$ as constructed above. Then, for each element $e_i, i\in[m]$ there exists a large constant $L$ such that $\sEl{e_i} =\tfrac{(1-L)}{L} \approx -1$. 
\end{proposition}
\begin{proof}
	Fix an element from $C_1$ in the closed interval of radius $r$ centered at the origin, i.e., $e_i, i\in[m]$. Then consider the value of $\aEl{e_i}, i\in [m]$ in the definition of $\sEl{e_i}$, we have that 
	\begin{align*}
		\aEl{e_i} &= \frac{\weightClustEl{}{e_i}}{m} =  \frac{1}{m}\sum_{j=1, j\neq i}^m d(e_j, e_i) + \frac{d(e_{m+1},e_i) }{m} \le \frac{m-1}{m} 2r + \frac{(R_1 + r)}{m} \le \frac{R_1}{m}+3r \enspace,
	\end{align*}
	by using the triangle inequality. We also have the following lower bound on $\aEl{e_i}$,
	\begin{align*}
		\aEl{e_i}&= \frac{\weightClustEl{}{e_i}}{m} =  \frac{1}{m}\sum_{j=1, j\neq i}^m d(e_j, e_i) + \frac{d(e_{m+1},e_i) }{m} \ge \frac{m-1}{m} \beta r + \frac{(R_1 - r)}{m}\ge \frac{R_1+r}{m} \enspace ,
	\end{align*}
	where in the first inequality we used the fact that the distance between two elements $e_i,e_j$ with $1\le i \neq j \le m$ of $\C_1$ is $\beta r,\beta=2/(m-1)$ and that  $d(e_{m+1}, e_i) \ge R_1-r$.  Hence $\aEl{e_i} = R_1/m + o(1/m)$ by the fact that $r=o(1/m)$. 
	Next for the term $\bEl{e_i}, i\in[m]$ with a similar argument we have that
	\begin{align*}
		\bEl{e_i}&= \frac{\weightClustEl{2}{e_i}}{m+1} =  \frac{1}{m+1}\sum_{j=1}^m d(e_j', e_i) + \frac{d(e_{m+1}',e_i) }{m+1} \le \\
		&\le \frac{m}{m+1} \left(\frac{1}{m}+2r\right) + \frac{(R_2 + r)}{m+1} \le \frac{r+1}{m+1}+2r+\frac{R_2}{m+1}\enspace.
	\end{align*}
	(e.g., by using the triangle inequality). Finally, for the lower bound,
	\begin{align*}
		\bEl{e_i}&= \frac{\weightClustEl{2}{e_i}}{m+1} =  \frac{1}{m+1}\sum_{j=1}^m d(e_j', e_i) + \frac{d(e_{m+1}',e_i) }{m+1} \ge\\
		&\ge \frac{m}{m+1} \left(\frac{1}{m}-2r\right) + \frac{(R_2 - r)}{m+1} \ge \frac{1-2mr-r}{m+1}+\frac{R_2}{m+1}\enspace.
	\end{align*}
	Hence $\bEl{e_i} = (R_2+1)/(m+1) + o(1/m)$. Focusing on the dominating terms of $\aEl{e_i}$ and $\bEl{e_i}$ and setting $R_1 = L(R_2+1)m/(m+1)$,  we have that
	\[
	\sEl{e_i} = \frac{(R_2+1)/(m+1) - R_1/m}{\max\{(R_2+1)/(m+1),R_1/m\}}= {\frac{(1-L)(R_2+1)/(m+1)}{L(R_2+1)/(m+1)}} = {\frac{(1-L)}{L}} = -1+{\frac{1}{L}} \approx -1\enspace .
	\]
\end{proof}

\textbf{Comparing uniform sampling with our algorithm~\algLocal.} Consider now a \emph{uniform Poisson sampling strategy} (denoted as \emph{uniform sampling} for short), that estimates $\weightEl{j}, j=1,2$ as follows:
\begin{compactenum}
	\item Collect a sample $S_{C_j}, j=1,2$ of points from each cluster, including each point with probability $p = t/|C_j| = t/(m+1), p>0$ for $t=o(m)$ (i.e., the same $t$ used by~\algLocal);
	\item Estimate $\weightEl{j}, j=1,2$ as:
	\[
	\weightElEstUni{j} = \sum_{e'\in S_{C_j}}  \frac{d(e,e')}{p} = \sum_{e'\in S_{C_j}}  \frac{d(e,e')(m+1)}{t}\enspace.
	\]
\end{compactenum}

It is easy to show that the estimators $\weightElEstUni{j}$ are unbiased. 
We can also use the estimators $\weightElEstUni{j}$ to approximate the silhouette $\sEl{e}$ in an analogous way as done in~\Cref{eq:localSilhEst}.
That is for each $e\in C_1$, we consider the estimate
\[ \sElEstUni{e} = \frac{\weightElEstUni{2}/(m+1) - \weightElEstUni{1}/m}{\max\left\{\weightElEstUni{2}/(m+1),\weightElEstUni{1}/m\right\}} \enspace.
\] 
In the next lemma, we show that for a fixed sample size $t=o(m)$ and with arbitrary large probability close to 1, it holds that $|\sElEstUni{e}-\sEl{e}|\approx 2$. That is, the estimation error of the uniform sampling strategy corresponds to the \emph{maximum error attainable}, \emph{for all $e\in C_1\setminus\{e_{m+1}\}$}. 
In contrast, the estimates provided by~\algLocal, using the same value of $t$ achieve, for $e\in C_1$, $|\sElEst{e}-\sEl{e}|\approx 0$ with large probability, that can be made arbitrarily close to 1. 
A similar argument can be used to show a large constant estimation error for $\sElEstUni{e}$ over all elements in $e\in C_2$ .

\begin{proposition}\label{prop:badUniAppx1D}
	Let $V$ and $\clustC=\{\C_1,\C_2\}$ be the instance to~\Cref{probl:LocalEst} constructed above. Then, for each element in $e_i \in C_1$ with $i\in[m]$ the value of the uniform Poisson sampling strategy yields $\sElEstUni{e_i}\approx 1 $ while~\algLocal yields $\sElEst{e_i}\approx -1$ with arbitrarily large probability. Therefore, $|\sEl{e_i} -\sElEstUni{e_i}|\approx 2$, while $|\sEl{e_i} -\sElEst{e_i}|\approx 0$.
\end{proposition}
\begin{proof}
	Let $a$ be a small constant. 
	For ease of presentation and given that $m$ is very large, our analysis considers the following case: ``the uniform Poisson sampling strategy over $C_j$ yields $a t$ total sampled elements''. 
	Note that since the number of samples corresponds to a Binomial distribution Bin$(m+1,t/(m+1))$ for large $m$, this is often the case, up to a very small constant. 
	A similar result holds for the size of $\F_{\C_j}, j\in[k]$ of our algorithm~\algLocal (see the proof of~\Cref{lem:samples}).
	Therefore, for simplicity, in our analysis we assume that both $|S_{C_j}|=|\F_{C_j}|=a t$, \emph{exactly}. 
	Consider the uniform sampling strategy and define the following two events $E_1 \doteq$ ``$e_{m+1} \text{ is not sampled in } S_{C_1}$'' and $E_2 \doteq$ ``$e_{m+1}' \text{ is not sampled in } S_{C_2}$''.
	We have that $\mathbb{P}[E_1] = \mathbb{P}[E_2] = (1-t/(m+1)) = 1 - o(1)$ for arbitrary large $m$, by the choice of $t$. 
	Hence, conditioning on $E_1$ and $E_2$, we obtain that for $e\in C_1\setminus \{e_{m+1}\}$,
	\begin{align*}
		\weightElEstUni{1} \ge \frac{m+1}{t} \beta r a t \ge \frac{(m+1)2ra}{(m-1)} \enspace, \quad \weightElEstUni{1} \le \frac{m+1}{t} 2rat  = (m+1)2ra  
	\end{align*}
	and
	\begin{align*}
		&\weightElEstUni{2} \le \frac{m+1}{t} \left(\frac{1}{m} +2r\right) at \le \frac{a(m+1)}{m} + (m+1)2ra \enspace,\\ 
		&\weightElEstUni{2} \ge \frac{m+1}{t} \left(\frac{1}{m}-2r\right) at \ge \frac{a(m+1)}{m} - a(m+1)2r \enspace .
	\end{align*}
	Hence, for $m$ large and $r=o(1/m)$ we have that $\weightElEstUni{1} $ approaches $ 0$ while $\weightElEstUni{2} = a(m+1)/m$. Therefore, focusing on the dominant terms we have that
	\[
	\sElEstUni{e} = \frac{1/m}{1/m} + o(1) \approx 1\enspace.
	\]
	Next consider the estimates $\weightElEst{j}, j=1,2$ computed by~\algLocal. Where we condition on the events $E_1$ and $E_2$ defined over $F_{C_j}^0$. Then we have that $\gamma_e \ge \max_{\bar{e}\in \F_{\C_1}^0} d(\bar{e}, e)/\widehat{W}_{C_1}(\bar{e})$ for $e=e_{m+1}$ that is $\gamma_e \ge (R_1-r)/(R_1+r+2mr) = 1-o(1)$. 
	Therefore over $t$ points selected in the final sample $\F_{\C_1}$, point $e_{m+1}$ will be selected with probability arbitrary close to 1 by~\algLocal, by the choice of $r$. 
	A similar argument applies to $e_{m+1}'$.
	We therefore have that for $e\in C_1\setminus\{e_{m+1}\}$ it holds
	\begin{align*}
		&\weightElEst{1} = d(e,e_{m+1}) + \sum_{e' \in \F_{\C_1}} \frac{d(e,e')}{p_e'} = R_1 + o(R_1)\\ 
		&\weightElEst{2} = d(e,e_{m+1}') + \sum_{e' \in \F_{\C_2}} \frac{d(e,e')}{p_e'} = R_2 + o(R_2)\enspace 
	\end{align*}
	by the facts that $r=o(1/m)$, $t=o(m), R_2=\Theta(m)$ and that $C_2$ is centered at $-1/m$. Hence, by the choice of $R_1 = L(R_2+1)$, for a suitably large $L$, focusing on the dominant terms
	\begin{align*}
		\sElEst{e} = \frac{R_2/(m+1) - R_1/m}{\max\{R_2/(m+1),R_1/m\}} \approx -1\enspace.
	\end{align*}
	The final claim follows by computing the distance of $\sElEstUni{e}$ and $\sElEst{e}$ using the exact silhouette values of $e_i, i\in[m]$ from~\Cref{prop:exSilhBalls1D}, i.e., $\sEl{e}\approx-1$.
\end{proof}
\textbf{Proving our result.} The final claim of~\Cref{theo:unifails} follows as a corollary of~\Cref{prop:badUniAppx1D}, and noting that $|V|=n=2m+2$ and $|\C_1| = m+1$. We remark that our result provides evidence of the superior accuracy achieved by~\algLocal over a naive uniform Poisson sampling approach. 
That is, our~\algLocal provably reports much more accurate estimates of the values of $\sEl{e}, e\in V$ that cannot be matched by a uniform sampling strategy, used to approximate the values of $\weightEl{j},j\in[k]$, when both approaches retain the same sample complexity.
We further complement our theoretical insights with our extensive experimental evaluation in~\Cref{sec:exp}.

\subsection{Estimating the silhouette of a clustering}\label{subsec:globalSilh}

We now study the problem of obtaining a tight and accurate estimate of the silhouette of a clustering $\clustC$, as captured in our next statement.

\begin{problem}\label{probl:GlobEst}
	Given a set of elements $V$ and a $k$-clustering $\clustC$ of $V$ and two parameters $\varepsilon,\delta \in (0,1)$, obtain an estimator $s'(\clust)$ such that $|s'(\clust) - \sEl{\clust}| \le f(\varepsilon)$ with probability at least $1-\delta$, minimizing the number of distance computations required to compute $s'(\clustC)$.
\end{problem}

Recall that $\sEl{\clust}$ is largely used in practice for evaluating the quality of a clustering $\clust$---where \clust is often obtained optimizing popular clustering metrics such as the $k$-means or $k$-medoids objective functions~\citep{Dudek2020Silh}. Hence, obtaining an accurate approximation of $\sEl{\clust}$ is of high practical utility, e.g., for identifying a good value for the parameter $k$ on large datasets where $\sEl{\clust}$ cannot be computed exactly.

In the next sections we will obtain several estimators for the value $\sEl{\clust}$, solving~\Cref{probl:GlobEst}, where, in analogy with~\Cref{probl:LocalEst}, we want $f(\varepsilon) = \Theta(\varepsilon)$. 
In addition, we want to output our highly accurate estimate $s'(\clustC)$ performing significantly less distance computations required by the exact algorithm, i.e., $\Theta(n^2)$.

\subsubsection{Estimator based on the silhouette definition} 
Our first estimator (\textsc{gl}obal-\textsc{s}ampler, \algGlobOne) is based on the definition of $s(\clustC)$ itself (see Equation~\ref{eq:avgShil}). 
That is, $\sEl{\clustC}$ corresponds to the average of the silhouette of all the elements $e\in V$, where each term $\sEl{e}$ contributing to the average takes a value ranging from $-1$ to $1$.
Therefore, our first estimator is based on \emph{sampling} a sufficiently small number of elements from $V$, and using the average of their \emph{exact} silhouette value (i.e., $\sEl{e}$) as proxy for the overall silhouette $s(\clustC)$.

The estimator is computed as follows: we sample $\left\{ e_1,\dots,e_m \right\}$ \emph{random} elements from $V$ where each element $e_\ell, \ell\in[m]$ is sampled uniformly at random from the set $V$. The resulting estimator can be expressed as
\begin{equation}\label{eq:est1}
	\widehat{s}_1(\mathcal{C}) = \frac{1}{m}\sum_{\ell=1}^{m} s(e_{\ell}) \enspace.
\end{equation} 

The next result quantifies the guarantees offered by estimator $\widehat{s}_1(\mathcal{C})$.

\begin{theorem}
	\label{th:allsilh_alg1}
	Let $V$ be a dataset of $n$ elements, and let
	$\mathcal{C}$ be a $k$-clustering of $V$.  If $m \ge \frac{2}{\varepsilon^2} \ln \frac{2}{\delta}$, then 
	$|\widehat{s}_1(\mathcal{C})-s(\mathcal{C})| \leq \varepsilon$
	with probability at least $1-\delta$.
\end{theorem}

\begin{proof}
	Note that $s(e_{\ell}), \ell=1,\dots,m$ are random variables taking values in the interval $[-1, 1]$. Moreover, for each $\ell=1,\dots,m$, we have that 
	$\mathbb{E}\left[s(e_{\ell})\right] = 
	\sum_{e\in V} s(e)/n =  s(\mathcal{C})$. 
	The result follows directly from the application of Hoeffding's bound (\Cref{th:Hoeffding}).
\end{proof}

The time complexity (i.e., number of distance computations) required to compute $\widehat{s}_1(\clust)$ is $\bigO(mn)$, for fixed $m$.

\subsubsection{Estimator using \algLocal} Recall that in Section~\ref{subsec:localPPS} we developed our algorithm~\algLocal that given a dataset $V$, a clustering $\clustC$ of $V$ and two parameters $\varepsilon, \delta\in (0,1)$ computes a set $\sElEst{V}$ such that for each $e\in V$ it holds $|\sElEst{e}-s(e)|\le (4\varepsilon)/(1-\varepsilon)$ with probability at least $1-\delta$. 
Our second estimator (\textsc{gl}obal-\textsc{pps}-\textsc{f}ull-sampler, \algGlobPPS) approximates the silhouette  $s(\clustC)$ using the estimates $\sElEst{e}$ for each element $e\in V$ as obtained by \algLocal, that is

\begin{equation}\label{eq:est2}
	\widehat{s}_2(\mathcal{C}) = \frac{1}{n}\sum_{e \in V} \sElEst{e}\enspace,
\end{equation}
where the values $\sElEst{e}, e \in V$ are computed by \algLocal. The next theorem establishes an upper bound to the absolute error of $\widehat{s}_2(\mathcal{C})$ with respect to $s(\mathcal{C})$.
\begin{theorem}
	\label{th:allsilh}
	Let $V$ be a dataset of $n$ elements, and let
	$\mathcal{C}$ be a $k$-clustering of $V$. Let $\widehat{s}_2(\mathcal{C})$
	be the estimate of the silhouette of the clustering $s(\mathcal{C})$
	computed as in \Cref{eq:est2} by running \algLocal for given parameters
	$\varepsilon$ and $\delta$, with $0 < \varepsilon, \delta < 1$, and
	for a suitable choice of the
	sample size $t$. Then, 	with probability at least $1-\delta$, it holds that
	\[
	|\widehat{s}_2(\mathcal{C})-s(\mathcal{C})|
	\leq \frac{4\varepsilon}{1-\varepsilon} \enspace .
	\]
\end{theorem}

The proof is an immediate consequence of the definition
of $s(\mathcal{C})$ (\Cref{eq:avgShil}) and $\widehat{s}_2(\mathcal{C})$ (\Cref{eq:est2}),
and of the guarantees provided by \Cref{lem:mainerror,lem:singlesilh}.

For a given (expected) sample size $t$,\footnote{Recall that the number of samples in each cluster in~\algLocal is $\bigO(t)$ with probability at least $1-2\delta/5$.} the time complexity to compute $\widehat{s}_2(\clust)$ is the same required by \algLocal, that is $\bigO(nkt)$ with probability at least $1-\delta$ (see~\Cref{lem:samples}, and the related discussion).

\subsubsection{Estimator using two-phase sampling}
Let us consider, for ease of notation, the set $ \sElEst{V} = \{\sElEst{e_1},\dots,\sElEst{e_n}\}$ of estimates computed with~\algLocal. 
Note that the values in the set $\sElEst{V}$ are \emph{fixed} once the \PPS sampling {phase} of Algorithm~\ref{code:algorithm} (i.e, \emph{Phase 1}) has been performed, even if the values in $\sElEst{V}$ have not explicitly been computed yet.\footnote{\emph{After} Phase 1, the values $\sElEst{e_1},\dots,\sElEst{e_n}$ are determined, i.e., there is no randomness once the samples $\F_{\C_{j}}$ for $j\in[k]$ are fixed.} We therefore approximate the estimator in Equation~\eqref{eq:est2} as follows:
\begin{enumerate}
	\item Execute \emph{Phase 1} of Algorithm~\ref{code:algorithm} and obtain all samples $\F_{\C_{j}}, j\in[k]$, so that the values of $\sElEst{e_i}, i\in[n]$ are determined (i.e., the set $\sElEst{V}$), even if not computed yet;
	\item Select uniformly at random $m\in\mathbb{N}, m\geq2$ indices $i_1,\dots,i_m$ in the set $\{1,\dots,n\}$;
	\item For each $i_j, j \in \left\{1,\dots,m\right\}$ evaluate $\sElEst{e_{i_j}}$;
	\item Return the average $\bar{s}(\mathcal{C})$ over the $m$ estimates $\sElEst{e_{i_j}}$ obtained above.
\end{enumerate}
The resulting estimator (\textsc{gl}obal-\textsc{pps}-\textsc{s}ub-sampler, \algGlobSub) can be expressed as follows,
\begin{equation}\label{eq:est3}
	\bar{s}(\clustC)=\frac{1}{m}\sum_{j=1}^m \sElEst{e_{i_j}}\enspace.
\end{equation}

The following result shows that  $\bar{s}(\clustC)$ is an unbiased estimator of $\widehat{s}_2(\mathcal{C})$.

\begin{lemma}\label{lemma:unbiasedness}
	$\mathbb{E}\left[\bar{s}(\mathcal{C})\right] =\widehat{s}_2(\mathcal{C})$.
\end{lemma}

\begin{proof}
	For each $j \in \left\{1,\dots,m\right\}$, let $X_j$ be a random variable with value $\hat{s}(e_{i_j})$ (observe that $i_j$ is itself a random variable). Note that $\bar{s}(\clustC)=\frac{1}{m}\sum_{j=1}^m X_j$.  By taking expectations, using linearity, and the fact that r.v.'s ${i_j}$'s are taken uniformly at random from the set $\left\{1,\dots,n\right\}$ we have
	\[
	\mathbb{E}[\bar{s}(\clust)] =  \frac{1}{m}\sum_{j=1}^m \mathbb{E}\left[X_j\right] = \frac{1}{m}\sum_{j=1}^m \sum_{i=1}^n \sElEst{e_i}\frac{1}{n} = \frac{1}{n}\sum_{i=1}^n \sElEst{e_i} = \widehat{s}_2(\clust).
	\]
\end{proof}

The following result provides the guarantees on the approximation of $\widehat{s}_2(\mathcal{C})$ provided by its estimator $\bar{s}(\mathcal{C})$.
\begin{lemma}\label{lemma:sshybr}
	Fix $(\varepsilon_2,\delta_2)\in(0,1)^2$ then if $m\geq \frac{2}{\varepsilon^2_2} \ln\left(\frac{2}{\delta_2}\right)$ then $\mathbb{P}[|\widehat{s}_2(\mathcal{C}) - \bar{s}(\mathcal{S})| \ge \varepsilon_2] \le \delta_2$.
\end{lemma}
\begin{proof}
	First recall variables $X_j, j\in[1,m]$ are bounded in the range $[-1,1]$. Then combining the proof of~\Cref{lemma:unbiasedness} showing that for each $j\in[m]$, $\Exp[X_j] = \widehat{s}(\mathcal{C})$ and Hoeffding's bound (\Cref{th:Hoeffding}) we have that 
	\[
	\mathbb{P}\left[\left|\bar{s}(\mathcal{C}) - \widehat{s}_2(\mathcal{C}) \right|\ge \varepsilon_2 \right] \le 2 \exp\left( -\frac{2m\varepsilon_2^2}{4}\right) \le \delta_2 \enspace ,
	\]
	where the last step is by the choice of $m$ as in statement.
\end{proof}

We now prove the resulting error obtained from approximating the silhouette  $s(\mathcal{C})$ with $\bar{s}(\mathcal{C})$. 
\begin{lemma}
	Fix $(\delta_1,\delta_2)\in(0,1)^2$ then $|s(\mathcal{C}) - \bar{s}(\mathcal{C})|\leq \varepsilon_1+\varepsilon_2$ with probability at least $(1-\delta_1)(1-\delta_2)$, where $\varepsilon_1$ is the bound on the error obtained in \Cref{th:allsilh} on $\widehat{s}_2(\mathcal{C})$ and $s(\mathcal{C})$ w.p.\ $1-\delta_1$, and $\varepsilon_2$ is a bound on the error between $\bar{s}(\mathcal{C})$ and $\widehat{s}(\mathcal{C})$ obtained with Hoeffding's inequality as in~\Cref{lemma:unbiasedness} w.p.\ at least $1-\delta_2$.
\end{lemma}
\begin{proof}
	Note that $|s(\mathcal{C}) - \widehat{s}_2(\mathcal{C})| \le \varepsilon_1$ and $|\widehat{s}_2(\mathcal{C}) -\bar{s}(\mathcal{C})| \le \varepsilon_2$ hold respectively w.p.\ $(1-\delta_1)$  and $(1-\delta_2)$, and that such events are independent. Therefore by the triangle inequality,
	\[
	|s(\mathcal{C}) - \bar{s}(\mathcal{C})| = |s(\mathcal{C}) \pm \widehat{s}_2(\mathcal{C})- \bar{s}(\mathcal{C})| \leq |s(\mathcal{C}) - \widehat{s}_2(\mathcal{C})| + |\widehat{s}_2(\mathcal{C}) -\bar{s}(\mathcal{C})| \leq \varepsilon_1+\varepsilon_2 \enspace,
	\] 
	and such event holds with probability at least $(1-\delta_1)(1-\delta_2)$.
\end{proof}

\begin{table}[t]
	\caption{Summary of the various estimators to approximate $\sEl{\clust}$, the silhouette of a clustering \clustC, as captured by~\Cref{probl:GlobEst}. We label with \emph{algorithm} the method used to evaluate each estimator. The sample size $m(\varepsilon_2,\delta_2)$ is detailed in~\Cref{lemma:sshybr}. 
	}
	\label{tab:estAlgs}
	\scalebox{0.87}{
		\begin{tabular}{lcccc}
			\toprule
			Algorithm & Estimator & Guarantees & Distance-Computations  & Probability\\
			\midrule
			\algGlobOne & $\globOne{\clust}$ & $|\globOne{\clust}-s(\clust)| \le \varepsilon $ & $\bigO\left(\frac{n}{\varepsilon^2}\log \left(\frac{1}{\delta}\right)\right)$ & $>1-\delta$\\
			\algGlobPPS & $\globTwo{\clust}$ & $|\globTwo{\clust}-s(\clust)| \le  \frac{4\varepsilon}{1-\varepsilon} $ & $\bigO\left(\frac{nk}{\varepsilon^2} \log\left(\frac{nk}{\delta}\right)
			\right)$ & $>1-\delta$\\
			\algGlobSub & $\globHyb{\clust}$ & $|\globHyb{\clust}-s(\clust)| \le  \frac{4\varepsilon_1}{1-\varepsilon_1} + \varepsilon_2 $ & $\bigO\left(
			\left(n+\frac{m(\varepsilon_2,\delta_2)k}{\varepsilon_1^2}\right)
			\left( 
			\log\left(\frac{nk}{\delta_1}\right) 
			\right) 
			\right)$ & $> (1-\delta_1)(1-\delta_2)$\\
			\bottomrule
		\end{tabular}
	}
\end{table}

It now remains to control $|s(\mathcal{C}) - \bar{s}(\mathcal{C})|\leq \varepsilon_1+\varepsilon_2$ to hold with probability at least $1-\delta$ within user defined accuracy and confidence $\varepsilon, \delta\in(0,1)$. Since the final error is bounded by $\varepsilon_1+\varepsilon_2$ we can set $ \varepsilon_2 = \varepsilon - \varepsilon_1 $. Noting that $\varepsilon_1 \le 4\varepsilon'/(1-\varepsilon')$ it is sufficient to set $\varepsilon'=\varepsilon/q, q> 4$. Finally to control the confidence we can set $\delta_1=\delta/2$ and $\delta_2 = \delta/(2-\delta)$, yielding to the desired confidence to be at least $1-\delta$---solving~\Cref{probl:GlobEst} as desired.

When performing Phase 1 of \algLocal for a fixed $t$, the time complexity to compute $\bar{s}(\mathcal{C})$ is at most $\bigO(n\log(nk/\delta)+mkt)$. That is the sum of the time complexity of Phase 1 (i.e., $\bigO(n\log(nk/\delta)$) and the time complexity of running Phase 2 of \algLocal only for the $m$ elements $e_{i_j}, j\in\left\{ 1,\dots,m\right\}$ (i.e., $\bigO(mkt)$).

We provide a summary of our sampling techniques designed to solve~\Cref{probl:GlobEst} in~\Cref{tab:estAlgs}.

\subsection{Distributed algorithms}
\label{sec:MR}

In this section, we present a distributed design of our silhouette approximation algorithms (i.e., both for local and global estimation) using the MapReduce framework \citep{DeanG08,LeskovecRU14,PietracaprinaPRSU12}. 
At the end of the section, we will briefly discuss how to simply adapt our design to the popular MPC model
\citep{KarloffSV10,BeameKS17,Im2023MPC}.  

\subsubsection{Computational model}\label{subsubsec:MRmodel} MapReduce adopts a functional programming model where an algorithm is specified as a sequence of \emph{rounds}. 
In each round:
\begin{compactenum}
	\item a multiset $X$ of key-value pairs is first transformed into a new multiset $X'$ of key-value pairs by applying a \emph{map function} independently to each pair (possibly in parallel);
	\item\label{item:MRstep2} the multiset $X'$ is then transformed into a multiset $Y$ of pairs by applying a \emph{reduce function} independently to each subset of key-value pairs of $X'$ sharing the same key (possibly in parallel);
	\item the multiset of key-value pairs $Y$ is then used as an input to the subsequent round, if any.
\end{compactenum}
Throughout the description of our algorithms, a key-value pair will be denoted as
$(\mbox{\it key} \; | \; \mbox{\it value}_1,\allowbreak \ldots, \mbox{\it value}_r)$, 
where $\langle \mbox{\it value}_1,\allowbreak \dots,\allowbreak \mbox{\it value}_r \rangle$ corresponds to an $r$-dimensional \emph{value}, for (a small) constant $r$. 
As customary in the MapReduce literature \cite{LeskovecRU14}, we use the term \emph{reducer} to refer to the application of a reduce function to a set of pairs with the same key, i.e., see step~(\ref{item:MRstep2}) in the above discussion.

When a MapReduce algorithm is executed in a distributed environment, the instances of the map (resp., reduce) functions of a round are assigned by the underlying system to the available computing nodes. 
That is, a MapReduce algorithm is agnostic to the actual number of computing nodes. 
Therefore, important performance indicators to be minimized for a MapReduce algorithm are: 
\begin{inparaenum}[$i$)]
	\item the number of rounds;
	\item the maximum amount $M_L$ of \emph{local memory} required by any instance of the map or reduce functions;
	\item the maximum amount $M_A$ of \emph{aggregate
		memory} required at \emph{any} round to store the input, output and intermediate data.
\end{inparaenum}

In our setting, 
consider a
$k$-clustering $\mathcal{C} = \{C_1, \dots, C_k\}$ of a dataset $V =
\{e_0, \dots e_{n-1}\}$ of $n$ elements, and assume that each element $e\in V$ has additional information $C(e)$ denoting the cluster to which $e\in V$ is assigned to over $\clustC$.  
Initially, we represent the input
clustering by the following set of key-value pairs:
\[
\{ (i | e_i) : 0 \leq i < n \}\enspace.
\]
Reasonably, we assume that the values $n$, $k$, and $t$ used in the MapReduce algorithm, and the size $|C_j|$ for each $j\in[k]$, are known globally.  
A key ingredient of our design is that in each
round, the elements of $V$, and the newly created intermediate data, 
will be partitioned into $w$ subsets of size $n/w$ each,
where $w \in [0,n-1]$ is a suitable design parameter. 
To simplify the description and the analysis of our methods, we assume that $w \in \BO{n^\alpha}$, for some (small) $\alpha <1/2$.
While our methods can be generalized to larger values of $w$ (i.e., those corresponding to $\alpha\geq 1/2$) with minimal modifications, we highlight that $w$ constitutes an upper bound to the maximum \emph{parallelism} exploitable by the algorithm. Thus, assuming that $w\in \BO{\sqrt{n}}$ 
does not constitute a limitation in practice, especially for modern datasets of massive size $n$.

\subsubsection{Approximating the silhouette of each element}\label{subsubsec:MRlocal} 

We now present how to adapt our algorithm \algLocal to the MapReduce framework.
Our method requires 3 rounds: the first two rounds obtain the sample $\PPSsample$, for every cluster $C_j$ and $j\in[k]$, while the last round computes the estimates $\sElEst{e}$ of $\sEl{e}$ of each element $e\in V$. 

\noindent
{\bf Round 1.}
\begin{compactitem}
	\item{\textbf{Map.}}
	Map each pair $(i | e_i)$ into the pair $(i \bmod w |
	e_i,0)$. In addition, if 
	$|C(e_i)| > t$ then, with probability $(2/|C(e_i)|)\ln (2k/\delta)$,
	select  $e_i$ to be part of the initial Poisson sample
	$F^{(0)}_{C(e_i)}$, producing $w$ additional pairs $(\ell | e_i,1)$,	with $0 \leq \ell < w$. 
	\item{\textbf{Reduce.}}
	For $0 \leq \ell < w$, the reducer associated with key $\ell$
	has a copy of all initial samples $F^{(0)}_{C_j}$ (represented by the pairs $(\ell | e_i,1)$) and a subset $V_{\ell}$ of $V$ (represented by the pairs $(\ell | e_i,0)$). 
	Observe that the sets $V_{\ell}, \ell\in[0,w)$ form a
	balanced partition of $V$ into $w$ subsets.
	Then, for each element $e_i \in F^{(0)}_{C(e_i)}$, the reducer computes the term $W_{i,\ell}$, defined as the sum
	of distances from $e_i$ to all elements of $V_{\ell} \cap C(e_i))$
	producing the pair $(\ell | e_i,W_{i,\ell},1)$.
	In addition, for each pair $(\ell | e_i,0)$
	produce the pair $(\ell | e_i,0,0)$.
\end{compactitem}

\noindent
{\bf Round 2.}
\begin{compactitem}
	\item{\textbf{Map.}} 
	\sloppy
	Map 
	each pair  $(\ell | e_i,W_{i,\ell},1)$ into
	the $w$ pairs $(0 | e_i,W_{i,\ell},1), (1 | e_i,W_{i,\ell}), \ldots (w-1 | e_i,W_{i,\ell},1)$, and each pair $(\ell | e_i,0,0)$ into itself.
	\item{\textbf{Reduce}.} 
	For $0 \leq \ell < w$, the reducer associated with key $\ell$
	has the following pairs: for every element $e_i$ belonging to the initial sample $F^{(0)}_{C(e_i)}$ of its cluster, $w$ pairs
	$(\ell | e_i,W_{i,\ell'},1)$, with $0 \leq \ell' < w$; and for every $e_i \in V_\ell$, a pair $(\ell | e_i,0,0)$.
	The reducer first computes $W_{C(e_i)}(e_i)$, for each 
	element $e_i \in F^{(0)}_{C(e_i)}$, by summing all $W_{i,\ell'}$'s.
	Next, for each $e_i \in V_{\ell}$ it produces the pair $(\ell | e_i,p(e_i))$, where:
	\begin{inparaenum}[1)]
		\item if $|C(e_i)| \leq t$, then $p(e_i)=1$; otherwise
		\item $p(e_i)$ is computed in terms of the $W_{C(e_i)}(e)$ with  
		$e \in \ensuremath{\F^0_{C(e_i)}}$, as specified in~\Cref{line:num8,line:num10} of~\Cref{code:algorithm}.
	\end{inparaenum}
\end{compactitem}

\noindent
{\bf Round 3.}
\begin{compactitem}
	\item{\textbf{Map.}} 
	Map each pair $(\ell | e_i,p(e_i))$ into the pair 
	$(\ell | e_i,p(e_i),0)$. 
	In addition, with probability $p(e_i)$
	select $e_i$ to be part of the Poisson sample
	$\PPSsample$ producing pairs $(\ell' | e_i,p(e_i),1)$,
	with $0 \leq \ell' < w$.
	\item{\textbf{Reduce.}} 
	For $0 \leq \ell < w$, the reducer associated with key $\ell$
	has a copy of all samples $\PPSsample$ (represented by the pairs $(\ell | e_i,p(e_i),1)$) and the subset $V_{\ell}$ (represented by the pairs $(\ell | e_i,p(e_i),0)$). Then, for each $e_i \in V_{\ell}$,
	the reducer computes the estimate 
	$\widehat{s}(e_i)$, and generates the pair $(0 | e_i,\widehat{s}(e_i))$ to be reported in output.
\end{compactitem}

\sloppy
Recall that \Cref{lem:samples} states that with probability at least   $1 -2\delta/5$, for every cluster $\C_j$ with $|C_j| > t$,  $|F^0_{C_j}|  =  \BO{\ln(k/\delta)}$
and $|F_{C_j}| = \BO{\varepsilon^{-2}\ln(nk/\delta)}$. 
Then, since $t= \BO{\varepsilon^{-2}\ln(nk/\delta)}$, it  immediately follows that, with probability at least $1-2\delta/5$ the maximum local memory required by any map or reduce phase in any of the above rounds is bounded by
$M_L = \BO{n/w + w + k\varepsilon^{-2}\ln(nk/\delta)} = \BO{n/w + k\varepsilon^{-2}\ln(nk/\delta)}$ 
and the aggregate memory is bounded by $M_A = \BO{wM_L}$. 

\subsubsection{Approximating the silhouette of a clustering}\label{subsubsec:MRglobal}

In Section~\ref{subsec:globalSilh}, we introduced three
methods for computing an approximation
$\widehat{s}(\clustC)$ to the silhouette of the entire
clustering $\clustC$, see~\Cref{tab:estAlgs}.
The two estimators $\widehat{s}_2(\clustC)$ 
(see \Cref{eq:est2})
and $\bar{s}(\clustC)$
(see \Cref{eq:est3}), can be computed through simple aggregations of all, or an $m$-sample, of the individual silhouettes estimates $\sElEst{e}$ for $e\in V$. 
These two estimators can be computed in MapReduce through a  straightforward adaptation of the 3-round algorithm described above,
combining, in the reduce phase of the third round, the individual silhouette estimates (or the sample ones) within each of the $w$ partitions, and then combining the resulting values in an extra round. Overall, this 4-round approach would use the same local and aggregate space bounds obtained in~\Cref{subsubsec:MRlocal}.
Finally, concerning the sampling method that computes $\widehat{s}_1(\clustC)$ based on~\Cref{eq:est1}, it can be simply implemented through the MapReduce framework, we omit the details for brevity, since the resulting approach is analogous to the techniques already introduced.
Such final approach uses $\BO{1+(\log w / \log(n/w))} = \BO{1}$ rounds, local memory $M_L = \BO{n/w + k\varepsilon^{-2}\ln(1/\delta)}$ and aggregate memory $M_A=\BO{wM_L}$, 

\subsubsection{Porting MapReduce algorithms to the MPC framework}\label{subsubsec:MPCPorting}

The MPC framework models a parallel system with $p$  machines, each equipped with a local memory of size $s$.
Similar to MapReduce, the computation is structured as a sequence of synchronous rounds. 
In a round, each machine
\begin{inparaenum}[(1)]
	\item performs an arbitrary amount of computation on its local data in memory;
	\item exchanges point-to-point messages with other machines. 
\end{inparaenum}
The MPC framework allows \emph{arbitrary} communication patterns, subject to the only constraint that each machine may send or receive \emph{at most $s$} point-to-point messages in each round, with messages sent in a round processed by their recipients in the next round. 
Similarly to MapReduce, the ultimate goal is to design MPC algorithms that run in as few rounds as possible, using local memory (heavily) sublinear in the aggregate memory size. 
However, MPC and MapReduce differ in the way they express parallelism: 
\begin{inparaenum}[$i$)]
	\item a MapReduce computation is purely functional and exhibits, at each round, a degree of (virtual) parallelism equal to the number of instances of the map and reduce functions active in the round; in contrast
	\item an MPC computation has fixed parallelism $p$ throughout all rounds. 
\end{inparaenum}

It is simple to observe that the MapReduce algorithms described above (for both local and global silhouette estimation) do not rely on the dynamic parallelism offered by the MapReduce model. 
That is, in every round, the algorithms always partition the data into $w$ disjoint subsets, and the computation essentially involves: 
\begin{inparaenum}[$i$)]
	\item local operations within each partition;
	\item broadcast of individual elements to all partitions;
	\item summing $w$ contributions, one per partition.
\end{inparaenum}
As a consequence, our algorithms can be immediately adapted to the MPC framework by associating one machine per partition, setting $p=w$ and $s=M_L$. 
Finally, since in our setting  $M_L=\BOM{n^{1/2}}$, all the required broadcast and aggregation operations can be performed in constant rounds using standard MPC techniques~\citep{CzumajGGJ25}. 
Thus, the overall number of rounds remains constant. 

The following theorem, whose proof is a direct consequence of all the arguments developed in this section, provides a quantitative characterization of the performance of our distributed strategies.
\begin{theorem}\label{th:ParallelGuarantees}
	There exist MapReduce and MPC and algorithms providing distributed implementations of~\algLocal, for local silhouette estimation, and of all our algorithms for global silhouette estimation from~\Cref{tab:estAlgs}. For fixed  $\varepsilon,\delta\in(0,1)$, $w=\BT{n^{\alpha}}$, for $\alpha \leq 1/2$, and $k = \BO{n^{1-\alpha}/\log^2(n)}$, our distributed algorithms require a constant number of rounds, sublinear local memory $M_L=\BT{n^{1-\alpha}}$, and  linear aggregate memory, with probability at least $1-2\delta/5$.
\end{theorem}

\section{Experimental evaluation}
\label{sec:exp}
In this section, we present the results of our extensive experimental evaluation. We start by describing the setup of our experiments (\Cref{subsec:setupExp}), and then we investigate the following issues.
\begin{compactitem}
	\item \textbf{\Qone.} Compare our new estimators (from~\Cref{subsec:globalSilh}) for the global silhouette $\sEl{\clust}$ with state-of-the-art baselines, studying the trade-offs between accuracy and efficiency exhibited by our estimators (\Cref{subsec:globalExps}).
	\item \textbf{\Qtwo.} Evaluate the accuracy of our algorithm~\algLocal for the estimation of all local silhouette values $\sEl{e}$, for each element $e\in V$ (\Cref{subsec:localEst}).
	\item \textbf{\Qthr.} Study the scalability of a distributed implementation of our algorithm~\algLocal (\Cref{subsec:parExps}).
	\item \textbf{\Qfou.} Illustrate applications of our new algorithms: \begin{inparaenum}
		\item our~\algLocal algorithm provides highly accurate \emph{local} estimates that can be used to visualize accurate silhouette plots;
		\item our \emph{global} estimates can be used for the accurate selection of a good value of the parameter $k$ of a $k$-clustering method
	\end{inparaenum} (\Cref{subsec:applicationsExps}).
\end{compactitem}

\subsection{Setup}\label{subsec:setupExp}

\subsubsection{Methods and baselines}\label{subsubsec:baselines}

We now introduce all methods considered in our experiments. We distinguish between methods for \emph{global} (\Cref{probl:GlobEst}) and \emph{local} estimation (\Cref{probl:LocalEst}) of silhouette values.

\textbf{Global estimates.} We compared our algorithms with two state-of-the-art baseline algorithms. In addition, we also considered a variant of our methods coupled with uniform sampling (see~\Cref{subsubsec:uniFails}).
The list of methods is as follows. 
\begin{itemize}
	\item 
	Algorithms \algGlobOne, \algGlobPPS, and \algGlobSub (from~\Cref{tab:estAlgs}), which compute the estimators in~\Cref{eq:est1,eq:est2,eq:est3}, respectively, using \PPS sampling to estimate individual silhouette values.
	\item 
	Two algorithms, dubbed \algGlobUNI and \algGlobUNISub, which compute the estimators in \Cref{eq:est2,eq:est3}, respectively, using uniform sampling (rather than \PPS sampling) to estimate all individual silhouettes (see~\Cref{subsubsec:uniFails}). 
	\item ``\fsbl'' by~\citet{FrahlingS08}, which, for each element $e\in V$ computes the term $\aEl{e}$ according to its definition (\Cref{eq:AandBterms}), while estimates the term $\bEl{e}$ as $1/|C_{\mu_e}| \sum_{e'\in C_{\mu_e}} d(e',e)$, where $C_{\mu_e}$ corresponds to the cluster with closest center (or centroid) to $e\in V$, different from the cluster $C$ of $e$. 
	Observe that the estimate of $\bEl{e}$ is simplified with respect to the one proposed in the original paper. 
	That is, we do not verify if the computed estimate is smaller than all the distances between $e$ and all centers (or centroids) of the clusters in $ \clustC\setminus\{ C_{\mu_e} \}$---to avoid performing an exact computation of $\bEl{e}$ when the check fails. 
	\item ``\simplbl'' by~\citet{Hruschka2004}, the simplified silhouette approach, which evaluates the silhouette $\sEl{e}$ of each element $e\in V$ using the distances of $e$ to the closest centers (or centroids) as proxies of the average distances. 
	Formally, let $\mu_1,\dots,\mu_k$ the centers of clusters $C_1,\dots,C_k \in \clust$ then for an element $e$ in some cluster $C$ the method \simplbl evaluates $\aEl{e} = d(e, \mu)$ where $\mu$ is the center of $C$, and $\bEl{e} = \min_{C_j\neq C} d(e,\mu_j)$. 
\end{itemize}
Note that both the methods \fsbl and \simplbl are deterministic, and do not offer quantitative
guarantees on their estimation error,
in contrast with our approaches from~\Cref{tab:estAlgs}, 
where the estimation error can be controlled as a function of the input parameters.

\textbf{Local estimates.} When solving \Cref{probl:LocalEst} we will compare our method \algLocal with the uniform-sampling based approach detailed in~\Cref{subsubsec:uniFails}, that we denote by \UNIbuck. That is,  we do not consider \fsbl that scales poorly on large data (our main focus for~\Cref{probl:LocalEst}), and \simplbl, which has very high estimation errors.

\begin{table}[t]
	\centering
	\caption{Datasets used in our experiments. For each dataset, the first three columns report its name and label, its size $n$, and the dimensionality $z$ of its elements (i.e., $e\in \mathbb{R}^z$). 
		Column ``exact'' indicates whether the $\Theta(n^2)$ exact algorithm could be run within the allotted time limit; 
		Column ``$k$-med'' indicates whether the dataset could be clustered with \kmed, under the distances listed in~\Cref{sec:data}); 
		and Column ``size'' labels the dataset as small (\texttt{S}), medium (\texttt{M}) or large (\texttt{L}), used in our discussion.
	}
	\label{tab:data}
	\scalebox{0.87}{
	\begin{tabular}{l r r c c c | l r r c c c}
		\toprule
		Name (label) & $n$ & $z$ & exact & $k$-med & size & Name (label) & $n$ & $z$ & exact & $k$-med & size\\
		\midrule
		Breast (\texttt{BR}) & 568 & 30 & \vmark & \vmark & \texttt{S} & BioKDD (\texttt{BioKDD}) & 146\,$K$ & 74 & \vmark & \xmark & \texttt{L} \\
		Wine (\texttt{WI}) & 6.5\,$K$ & 11 & \vmark & \vmark & \texttt{S} & RNA-seq (\texttt{RNA}) & 489\,$K$ & 8 & \xmark & \xmark & \texttt{L} \\
		Credit (\texttt{CR}) & 30\,$K$ & 23 & \vmark & \vmark & \texttt{M} & Metro (\texttt{MT}) & 1.5\,$M$ & 7 & \xmark & \xmark & \texttt{L} \\
		Shuttle (\texttt{SHU}) & 58\,$K$ & 9 & \vmark & \vmark & \texttt{M} & PowerHouse (\texttt{PH}) & 2.1\,$M$ & 7 & \xmark & \xmark & \texttt{L}\\
		IoT (\texttt{IOT}) & 123\,$K$ & 78 & \vmark & \vmark & \texttt{M} & Gowalla (\texttt{GOW}) & 6.4\,$M$ & 2 & \xmark & \xmark & \texttt{L} \\
		\bottomrule
	\end{tabular}
	}
\end{table}
\subsubsection{Datasets}\label{sec:data}
A summary of the datasets used in our experiments, and their statistics is reported in~\Cref{tab:data}. 
We provide more details of the processing of each dataset in~\Cref{appsec:data}.
We use small-size datasets (those marked with \texttt{S} in~\Cref{tab:data}) and \texttt{CR} to showcase applications of \Cref{probl:GlobEst,probl:LocalEst} (see~\Cref{subsec:applicationsExps}). 
As for medium-size datasets, they are mainly used in~\Cref{subsec:globalExps} and~\Cref{subsec:localEst}, and are processed as follows. 
We cluster each dataset using both the \kmeans and the \kmed objectives. 
For the \kmeans objective
we consider the Euclidean distance 
$\deucl(x,y) = \sqrt{\sum_{i} (x_i-y_i)^2}$. 
For the \kmed objective, in addition to $\deucl$, we also consider several widely used distance functions:
\begin{itemize}
	\item \dcos, the cosine distance $\dcos(x,y) = 1- \sum_{i} (x_i-y_i)/(\sqrt{\sum_i x_i^2} \sqrt{\sum_i y_i^2})$;
	\item \dmanh, the Manhattan distance $\dmanh(x,y) = \sum_{i} |x_i-y_i|$;
	\item \dcanb, the Canberra distance $\dcanb(x,y) = \sum_i |x_i-y_i|/(|x_i| + |y_i|)\enspace .$
\end{itemize}

For each clustering objective (either \kmeans or \kmed) and each distance, we use different values of the parameter $k\in \{2,5,10,15,20\}$ to obtain different clusterings $\clust$. 
Hence, each distinct configuration (dataset, value of $k$,  distance, objective) leads to a different clustering $\clust$, which we use as an input to either \Cref{probl:GlobEst} or \Cref{probl:LocalEst}.

Finally, for the large datasets from~\Cref{tab:data} we proceeded similarly to the medium-sized datasets, but we only clustered such data using the \kmeans objective to obtain different clusterings $\clust$ varying $k$, with $k\in\{2,5,10,15,20\}$. We do not use \kmed given its well-known poor scalability on large data. 

\subsubsection{Environment and parameters}
\label{subsubsec:impEnv}
We provide extensive details on the setup of the experiments in~\Cref{appsec:missingsetup}. The code to reproduce our results is available online.\footnote{\codeRepo.} 

For our experiments, we set a time limit of one-hour for the execution of each method on a given configuration. 
Unless otherwise specified, each algorithm was executed for 10 independent runs and the results show the average over the runs. 
For our approximate algorithms, we directly set the parameter $t$ controlling the expected sample size, 
rather than obtaining $t$ as a function of the other parameters (see~\Cref{code:algorithm}).  
When comparing the algorithms from~\Cref{tab:estAlgs} for~\Cref{probl:GlobEst}, and their \UNIbuck variant, we will set their parameters to have a comparable runtime: we fix $t$ for \algGlobPPS and \algGlobUNI and we set the sample size $m$ of \algGlobOne to $kt$. Unless otherwise stated we set $\delta=0.01$.

\subsubsection{Exact silhouette scores}
\label{subsubsec:exscores}

For the small and medium sized datasets from~\Cref{tab:data} we were able to compute the exact silhouette scores on all the various configurations described in~\Cref{sec:data} within one hour of computation. 
In contrast on large data, we were not able to obtain the exact silhouette scores of any of the multiple configurations within one-hour time limit.
To obtain a good proxy for the exact values of $\sEl{e}$ for each element $e\in V$, we then proceed as follows. 
We perform five independent executions of \algLocal and its uniform variant (\UNIbuck) from~\Cref{subsubsec:uniFails}, obtaining respectively $\sElEst{e}_i$ and $\sElEstUni{e}_{i}$ with $i\in\{1,\dots, 5\}$ with a high sample size of $t=800$. We then estimate $\sEl{e} = (1/10) \sum_i {(\sElEst{e}_i + \sElEstUni{e}_i)} $ to reduce the bias in the final result. 

\begin{figure}[t]
	\centering
	\captionsetup[subfigure]{labelformat=empty}
	\includegraphics[width=0.9\textwidth]{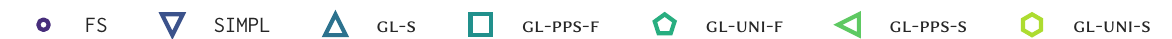}\\
	\begin{tabular}{lccc}
		\includegraphics[width=0.23\columnwidth]{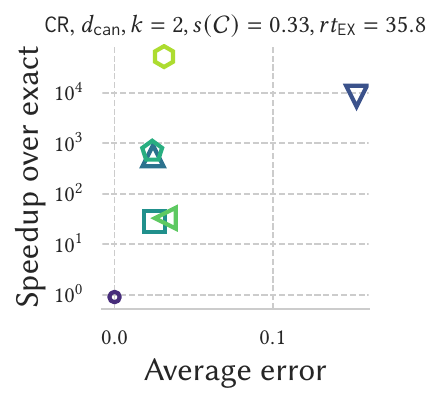} & 	\includegraphics[width=0.23\columnwidth]{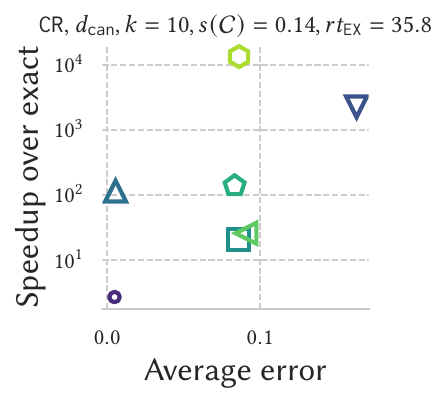}&
		\includegraphics[width=0.23\columnwidth]{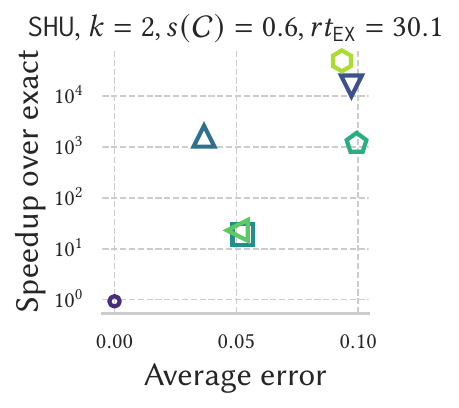} & 	\includegraphics[width=0.235\columnwidth]{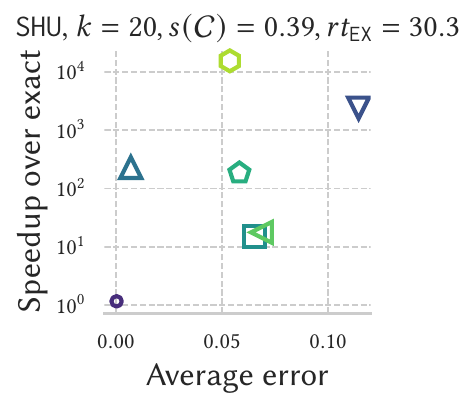}\\	\includegraphics[width=0.23\columnwidth]{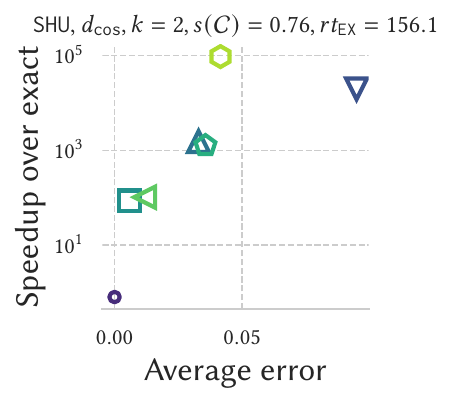}&
		\includegraphics[width=0.23\columnwidth]{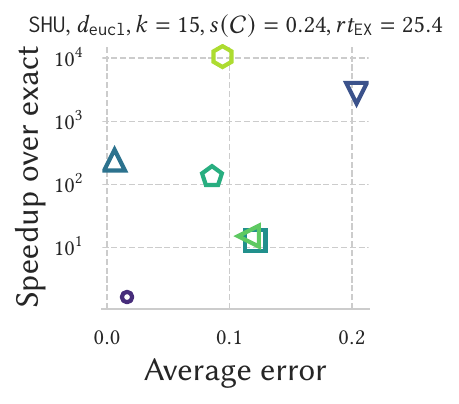} & 	\includegraphics[width=0.235\columnwidth]{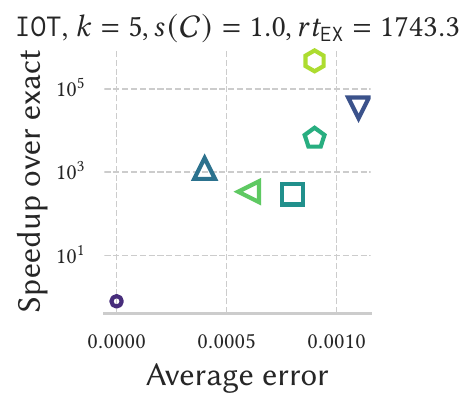} & \includegraphics[width=0.235\columnwidth]{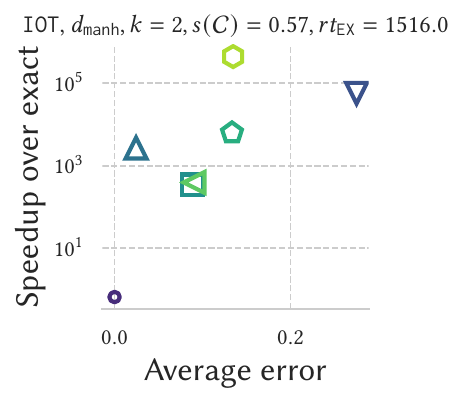} \\
		\includegraphics[width=0.235\columnwidth]{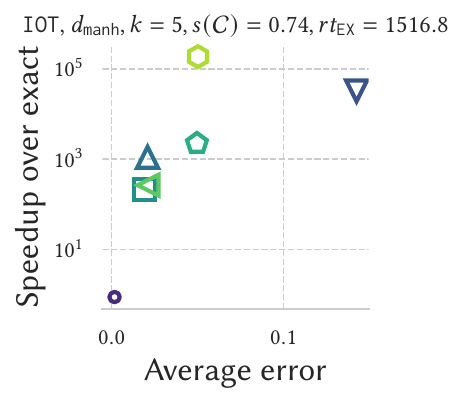}&
		\includegraphics[width=0.235\columnwidth]{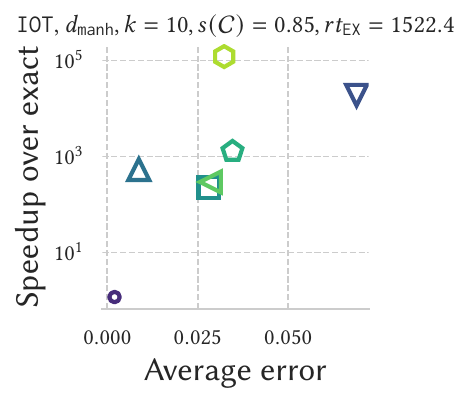} & 	\includegraphics[width=0.235\columnwidth]{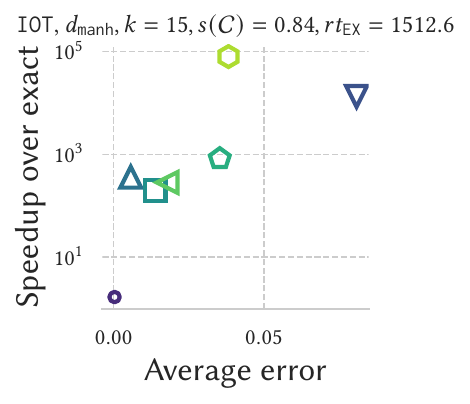} & \includegraphics[width=0.235\columnwidth]{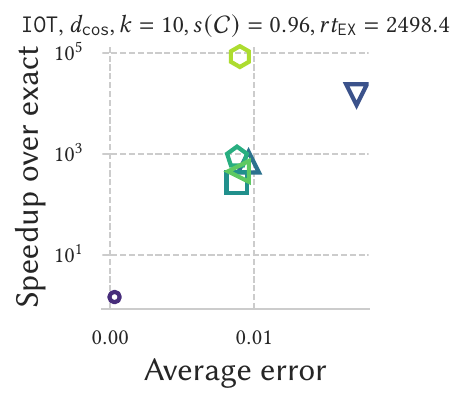} \\
	\end{tabular}
	\caption{Trade-off between accuracy and efficiency over medium sized datasets (see~\Cref{tab:data}). 
		The $x$ axis is associated with the average error over 10 independent runs, while the $y$ axis is associated with the average speed-up over the exact computation. For each plot that corresponding to a different configuration, we report: the dataset, the distance function (when \kmed was used to obtain the input clustering $\clustC$), the value $k$, the exact silhouette $\sEl{\clust}$, and the runtime of the exact algorithm $rt_{\text{EX}}$. 
	}
	\label{fig:tradeoffMedium}
\end{figure}

\subsection{Global silhouette approximation}\label{subsec:globalExps}

\begin{figure}[!tbp]
	\centering
	\captionsetup[subfigure]{labelformat=empty}
	\includegraphics[width=0.9\textwidth]{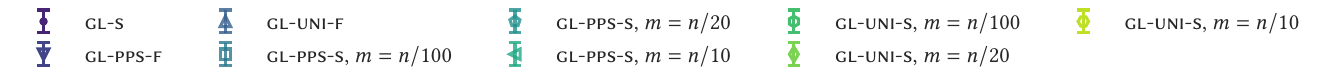}\\
	\begin{tabular}{ll}
		\includegraphics[width=0.49\textwidth]{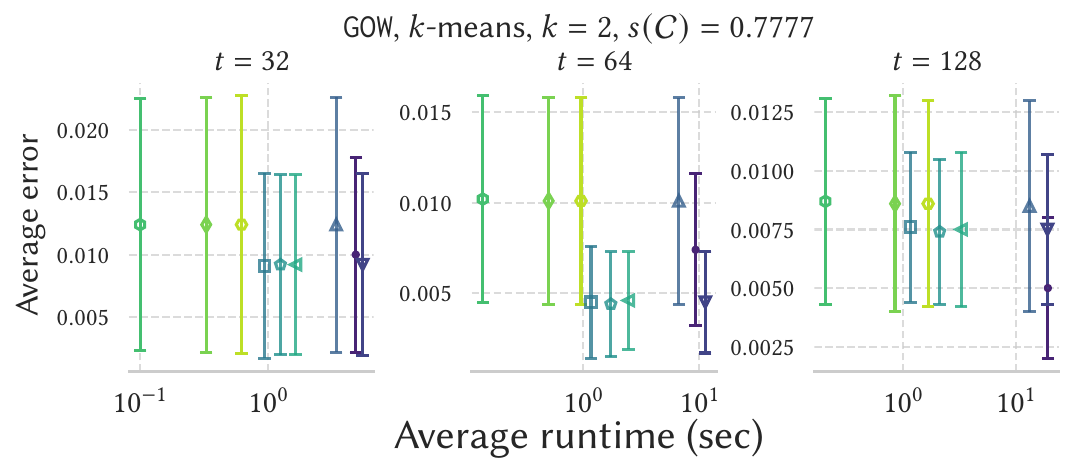} &
		\includegraphics[width=0.49\textwidth]{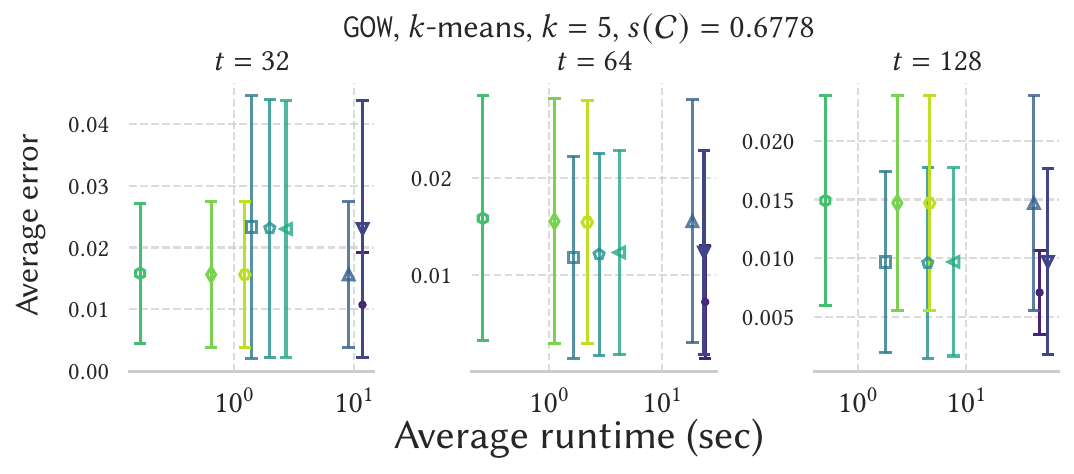}\\
		\includegraphics[width=0.49\textwidth]{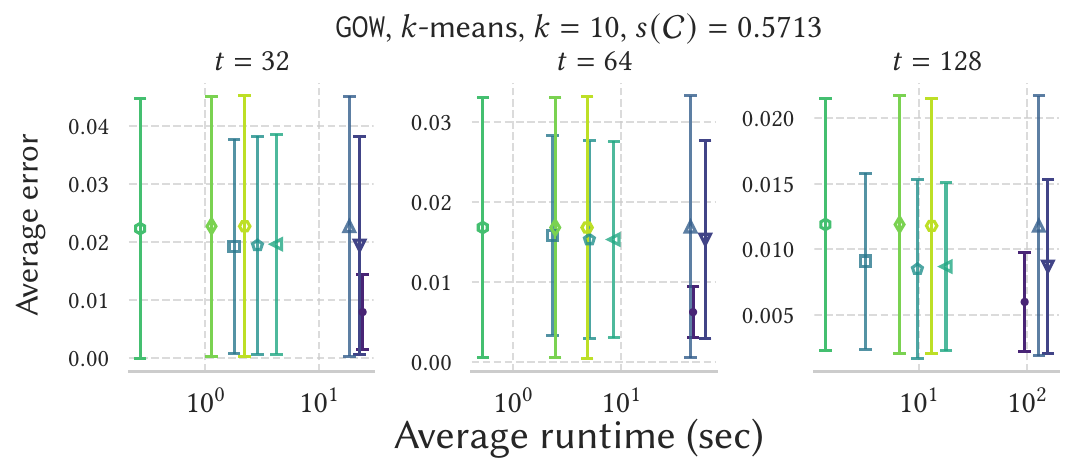} &
		\includegraphics[width=0.49\textwidth]{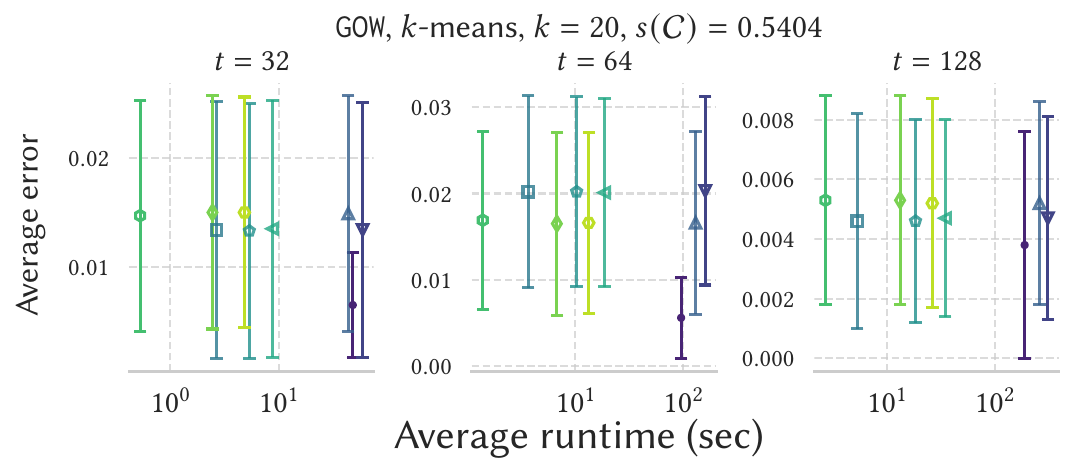}\\
		\includegraphics[width=0.49\textwidth]{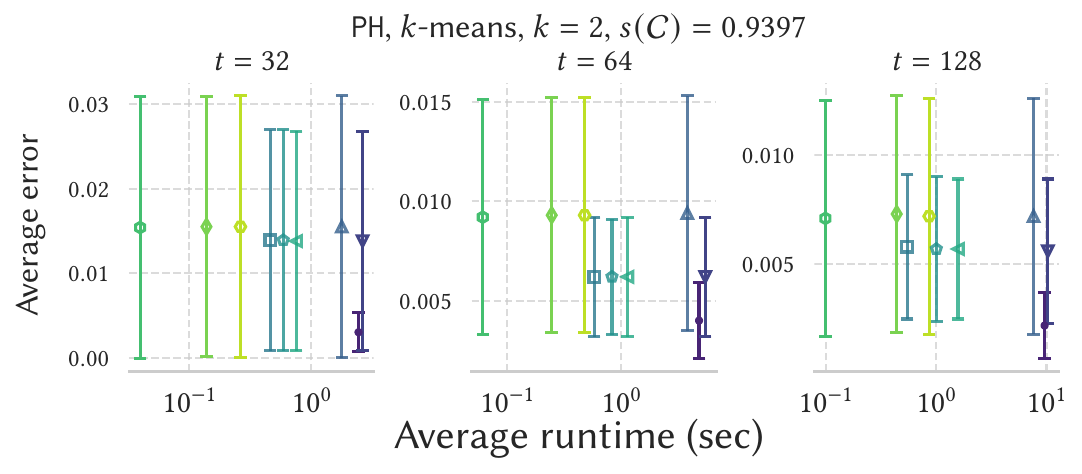} &
		\includegraphics[width=0.49\textwidth]{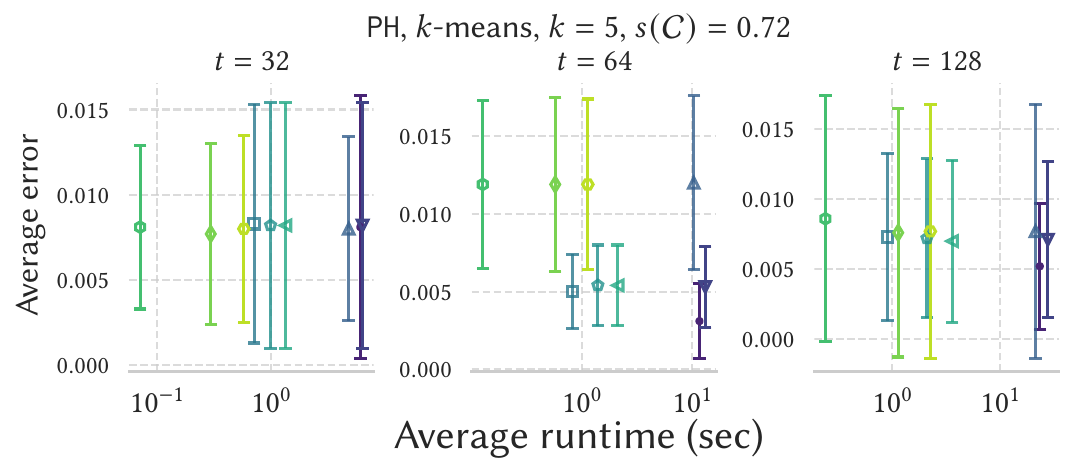}\\
		\includegraphics[width=0.49\textwidth]{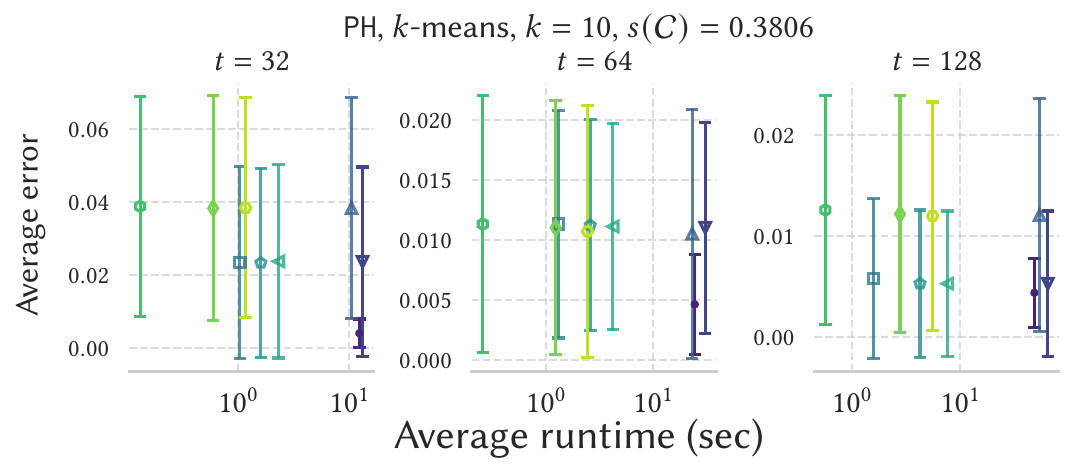} &
		\includegraphics[width=0.49\textwidth]{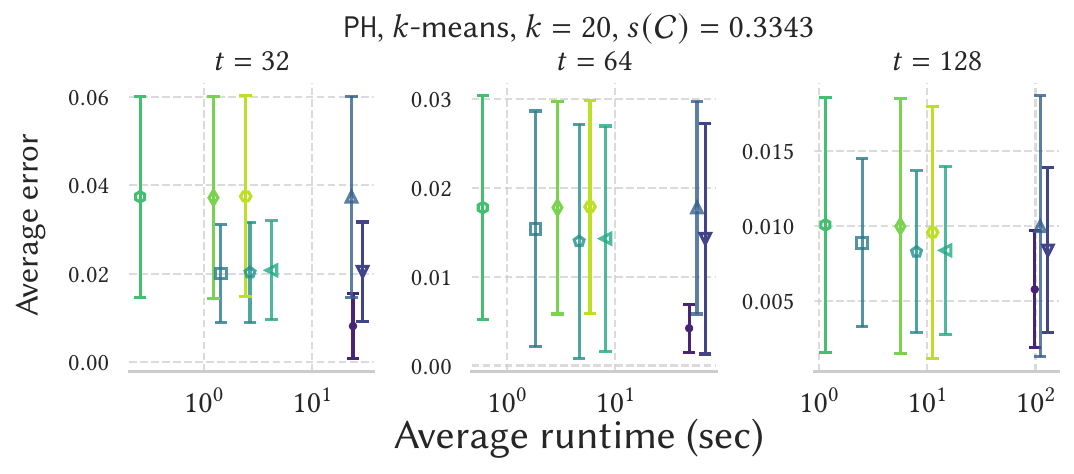}\\
	\end{tabular}
	\caption{Methods comparison. The $x$ axis (in $\log$ scale) is associated with the average runtime, and the $y$ axis with the average error (and its standard deviation) over 10 independent runs. For each clustered dataset we considered different values of the expected sample size $t\in \{32,64,128\}$.}
	\label{fig:globalEsts}
\end{figure}

In this section we evaluate the accuracy of the various methods and baselines described in~\Cref{subsubsec:baselines}. We compare all methods considering the average error $|\sElEst{\clust} - \sEl{\clust}| \in [0,2]$ over ten independent runs (see~\Cref{subsubsec:impEnv} for details), where $\sElEst{\clust}$ corresponds to the approximation provided by the method. 
In our discussion we also mention the \emph{relative} approximation error that corresponds to $|\sElEst{\clust} - \sEl{\clust}|/\sEl{\clust}$. 
For medium-sized datasets we set the parameters of the various methods (see~\Cref{tab:paramsRecap}) as: $t=16$ for both \algGlobPPS and \algGlobUNI, $m=\ceil{n/20}, t=16$ for \algGlobSub and \algGlobUNISub, and the sample size of \algGlobOne to $m'=tk$ depending on $\clust$ (see~\Cref{subsubsec:impEnv}). 
Our choice of a small value of $t$ is motivated by the size of our datasets, while the choice of $m$ enforces the computation of the estimator in~\Cref{eq:est3} on a strictly sublinear number of elements in $n$.

\subsubsection{Medium datasets}\label{subsubsec:mediumData}

We show some representative results in~\Cref{fig:tradeoffMedium}, where we report the trade-off between the estimation error, and the speed-up achieved over the exact computation of the silhouette coefficient by the various methods over ten runs.

We start by observing that the~\fsbl method has a very high runtime almost matching the exact algorithm. 
The only configuration where the speedup is more than $3\times$, with respect to the exact approach, is the one relative to the~\crdata dataset with \dcanb and $k=10$. 
Such behavior is not surprising since, for small $k$, it is likely that some clusters have roughly linear size, hence \fsbl will need a quadratic number of distance computations to compute the exact $a(e)$'s. 
This makes the \fsbl method  
not practical on large or medium sized data, since it scales poorly. 
Interestingly, we note that even if the \fsbl method is a heuristic approach (not offering guarantees), 
it often produces estimates with the smallest error among all methods evaluated. 
This is due to the fact that, for our datasets and configurations, the cluster minimizing the term $\bEl{e}, e\in V$ often corresponds to the cluster with the closest center in $\clust$ to $e$.

Next, we note that the~\simplbl method is often the fastest approach among all methods we compared (except on the \shudata data clustered with \kmeans and $k=2$). 
Unfortunately, \simplbl is also the method with the largest approximation errors, which cannot be controlled by any parameter. 
For example, on the \shudata data clustered with \kmed and distance $\deucl$ the exact silhouette is $\sEl{\clust}=0.24$, while \simplbl outputs an estimate with an absolute error of 0.2, implying a relative error of more than 80\%. 
Similarly, for the \crdata data clustered with \kmed, \dcanb and $k=10$, where the relative approximation error of \simplbl is more than 105\%. 
Such large errors hinder the use of the \simplbl method in practice.

Concerning the randomized sampling algorithms we observe that~\algGlobOne often achieves the smallest estimation error compared to~\algGlobPPS and \algGlobUNI--under our setting of parameters all three algorithms perform the same number of distance computations in expectation.
Interestingly, for some inputs such as \crdata (with \dcanb and $k=2$), \shudata (with \dcos and $k=2$), or \iotdata (with \dcos, $k=10$), our new~\algGlobPPS and \algGlobSub methods achieve small approximation errors, better than~\algGlobOne, while being slightly less efficient.
Comparing~\algGlobPPS and~\algGlobUNI, the former often achieves smaller approximation errors with respect to the latter (e.g., on the \shudata data with \kmeans and  $k=2$, or on the \iotdata data on most configurations).
Such a higher accuracy, comes at the expense of a higher runtime. 
Interestingly, we observe significant differences in runtime only on medium-sized data. 
This is probably caused by the fact that Phase 1 of \algGlobPPS has a non-negligible runtime compared to Phase 2 on medium-sized data such as \iotdata. 
Finally, we observe that both methods \algGlobSub and~\algGlobUNISub  retain the approximation guarantees offered by \algGlobPPS and~\algGlobUNI, respectively, while often improving slightly (for~\algGlobSub) or significantly (for~\algGlobUNISub) the runtime required to evaluate the estimates.

\subsubsection{Large datasets}\label{subsubsec:largeData}
For the experiments on large datasets, we clustered each dataset with \kmeans and distance $\deucl$. 
We did not consider baseline~\fsbl for its poor scalability and the~\simplbl method for its poor quality estimates (see~\Cref{subsubsec:mediumData}). 
In addition, we used highly accurate estimates of $\sEl{\clust}$ as ground truth, given the poor scalability of the exact computation (see~\Cref{subsubsec:exscores}). 

The results are reported in~\Cref{fig:globalEsts}. 
First, as expected, all algorithms exhibit increasing accuracy as the sample size $t$ grows. For example, on dataset~\gowdata with $k=5$, the error with $t=128$ for~\algGlobPPS is less than half the error obtained with $t=32$. 
Next, regarding the trade-off between runtime and accuracy of~\algGlobOne, \algGlobPPS and \algGlobUNI, which perform the same expected number of distance computations, we observe several interesting trends. 
In almost all cases \algGlobOne exhibited the highest accuracy and smallest variance, while \algGlobPPS provided more accurate estimates with (often significantly) smaller variance compared to \algGlobUNI.

For what concerns \algGlobSub and \algGlobUNISub, 
which compute the estimators of~\algGlobPPS and~\algGlobUNI respectively, using only a sample of size $m$ of the total elements $n$, we observe that a small sample ($m=n/20$) is often sufficient to retain the accuracy of \algGlobPPS and~\algGlobUNI. 
In addition, both \algGlobSub and~\algGlobUNISub achieve remarkable speedups of up to $10\times$ over \algGlobPPS, \algGlobUNI and \algGlobOne.
As in the case of medium datasets, the runtime of~\algGlobSub is slightly higher than the one of ~\algGlobUNISub given the overhead introduced by Phase 1 of \Cref{code:algorithm}.

\begin{figure}[t]
	\centering
	\captionsetup[subfigure]{labelformat=empty}
	\includegraphics[width=0.25\textwidth]{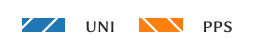}\\
	\begin{tabular}{lc}
		\includegraphics[width=0.45\textwidth]{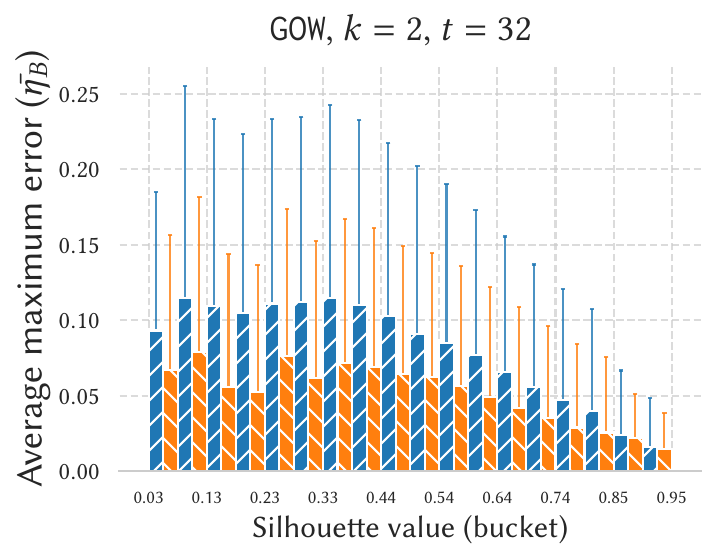} &
		\includegraphics[width=0.45\textwidth]{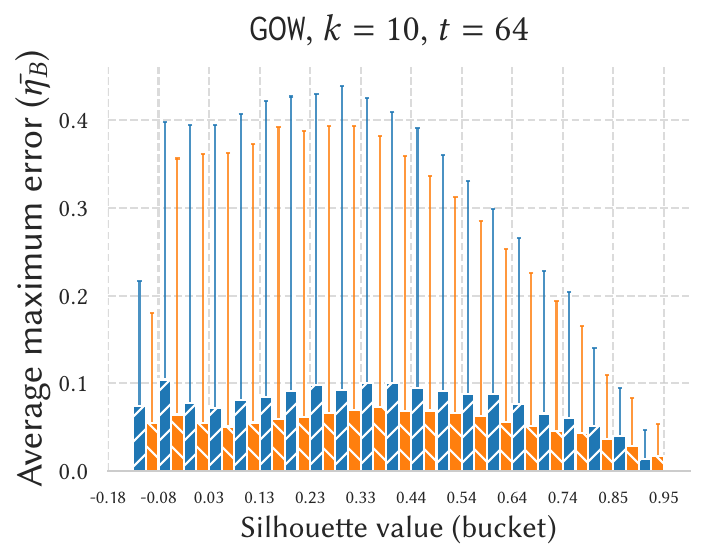} \\
		\includegraphics[width=0.45\textwidth]{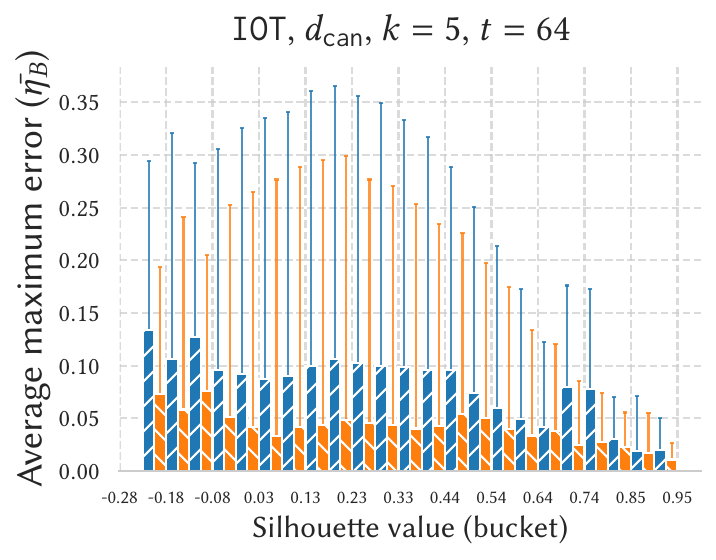} &
		\includegraphics[width=0.45\textwidth]{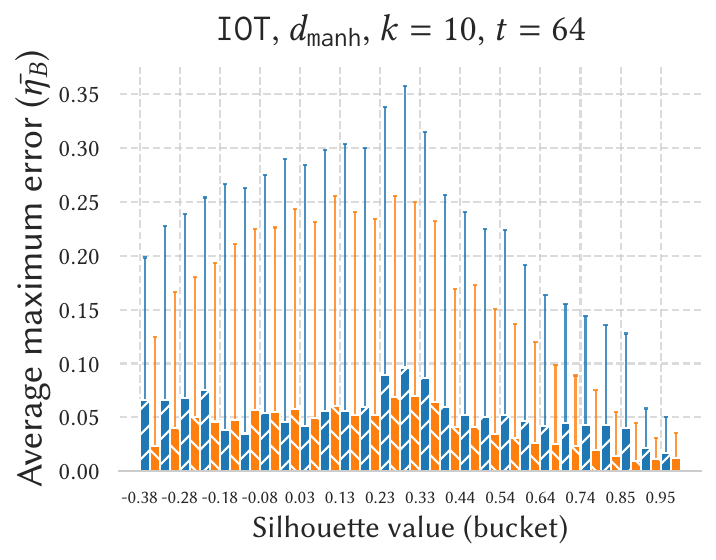} \\
	\end{tabular}
	\caption{Plots of the average maximum error and its standard deviation obtained on all elements $e\in V$, grouped by bucket, for \PPSbuck and \UNIbuck, over four configurations. When we report the distance, it indicates that the dataset has been clustered with \kmed. For ease of visualization we only display the average maximum error over non-empty buckets.
	}
	\label{fig:bucketsLarge}
\end{figure}

\textbf{Summary for issue \Qone.}  
Our experiments show that neither \fsbl nor \simplbl are suitable for estimating the global silhouette of a clustering, especially for large datasets. 
While \fsbl features good accuracy, it does not yield significant speedups over the exact algorithm. 
On the other hand, \simplbl is very fast but it incurs high approximation errors, which cannot be controlled by any parameter. 

As for randomized methods, \algGlobOne often achieves the best trade-off between accuracy and runtime compared with~\algGlobPPS and~\algGlobUNI. 
\algGlobPPS features better accuracy compared to \algGlobUNI, at the expense of a higher runtime on most datasets. 
\algGlobSub (resp., \algGlobUNISub) attains comparable accuracy to 
\algGlobPPS (resp., \algGlobUNI) even when processing a small number of elements $m$ compared to $n$---yielding a significant speedup on large datasets. 

Therefore, for large datasets, where the exact computation of $\sEl{\clust}$ is impractical even for a single and fixed value of $k$, 
our methods as well as \algGlobOne provide accurate estimates efficiently. 
For example, 
on the \gowdata data, which has more than 6 million elements, \algGlobPPS or \algGlobSub with $m=n/10$ and $t=128$ output very accurate silhouette estimates for \emph{all} values of $k\in \{2,5,10,20\}$ in less than five minutes and one minute, respectively.

\subsection{Local silhouette estimation}\label{subsec:localEst}
We now validate our algorithm \algLocal that estimates all values of $\sEl{e}, e\in V$, so to deal with issue \Qtwo. For ease of notation we denote \algLocal with \PPSbuck.
As baseline we consider the uniform Poisson-sampling method, that we denote with \UNIbuck, which computes $\sElEstUni{e}, e\in V$ using uniform Poisson sampling over each cluster $C\in \clust$ (see~\Cref{subsubsec:uniFails}) . 
For each dataset and value of $k$ considered in~\Cref{subsec:globalSilh}, we evaluate the accuracy of the local estimates in two settings. 
First, we consider a bucketing of the values of silhouette $\sEl{e}, e\in V$. 
That is, we assign each element $e\in V$ to a bucket based on its silhouette value $\sEl{e}$ where each bucket $B_i$, with $i=1,\dots,40$, corresponds to elements $e$ with  $s(e) \in [L_i, U_i)$, where $L_i = -1 + (i-1)/20$ and $U_i = -1 + i/20$.\footnote{The bucket corresponding to index $i=40$ also contains the right endpoint $L_{40}=1$.} 
We use the buckets $B_i$ to evaluate the accuracy of the local estimates over each individual bucket. 
Second, we also analyze the accuracy of the~\PPSbuck and \UNIbuck approaches over each individual cluster $C_1,\dots,C_k$ in $\clust$.

\subsubsection{Accuracy over buckets}\label{subsubsec:exactBuckets}
For each bucket $B_i$ let $E_i\subseteq V$ be the elements of $V$ with $\sEl{e}\in B_i$. 
We computed, for each independent run, the maximum error $\max_{e\in E_i}\{|\sElEst{e}-\sEl{e}|\}$ (where $\sElEst{e}$ is obtained either with \PPSbuck or \UNIbuck), and 
averaged the maximum error over ten independent runs, to obtain the average maximum error ($\bar{\eta}_{B_i}$) and its standard deviation, shown in~\Cref{fig:bucketsLarge}.
Note that $\bar{\eta}_{B_i}$ is a very precise metric to capture the local estimation performances, given by the fact that it considers the \emph{maximum error} over each bucket $B_i$, which can be very large, i.e., at most 2.

From~\Cref{fig:bucketsLarge} we first observe that both \PPSbuck and \UNIbuck approaches achieve smaller maximum error ($\bar{\eta}_{B_i}$) for buckets corresponding to elements with higher $\sEl{e}$.
Such behavior is likely caused by the fact that for elements with high value of $\sEl{e}$, the difference between $\bEl{e}-\aEl{e}$ (see~\Cref{eq:AandBterms}) is high, hence even rough estimates of terms $\weightEl{j}$ yield a good approximation of $\sEl{e}$. 
The previous observation also explains the larger error over elements with silhouette $\sEl{e}$ close to 0.
Comparing the average error $\bar{\eta}_{B_i}$ we observe that, \PPSbuck significantly reduces the error (and its standard deviation) over \emph{all} buckets with respect to \UNIbuck on most inputs. 
For example, on the \iotdata data with \dcanb and $k=5$ the estimates reported by \PPSbuck have half of the maximum error achieved by \UNIbuck.  
Similarly, on the \gowdata data for $k=2, t=32$ the standard deviation of the error of \PPSbuck is significantly smaller on each bucket compared to \UNIbuck. 
We observe that on some configurations~\PPSbuck achieves similar performances to the \UNIbuck approach (e.g., \gowdata for $k=10$). Nevertheless, we never observed \PPSbuck reporting estimates with significantly smaller accuracy compared to \UNIbuck, highlighting the high quality of the estimates of \PPSbuck as captured by~\Cref{theo:localAlg}. We provide some additional results in~\Cref{appsec:additRes}. 

\begin{figure}[t]
	\centering
	\captionsetup[subfigure]{labelformat=empty}
	\includegraphics[width=0.25\textwidth]{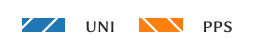}\\
	\begin{tabular}{lr}
		\includegraphics[width=0.45\textwidth]{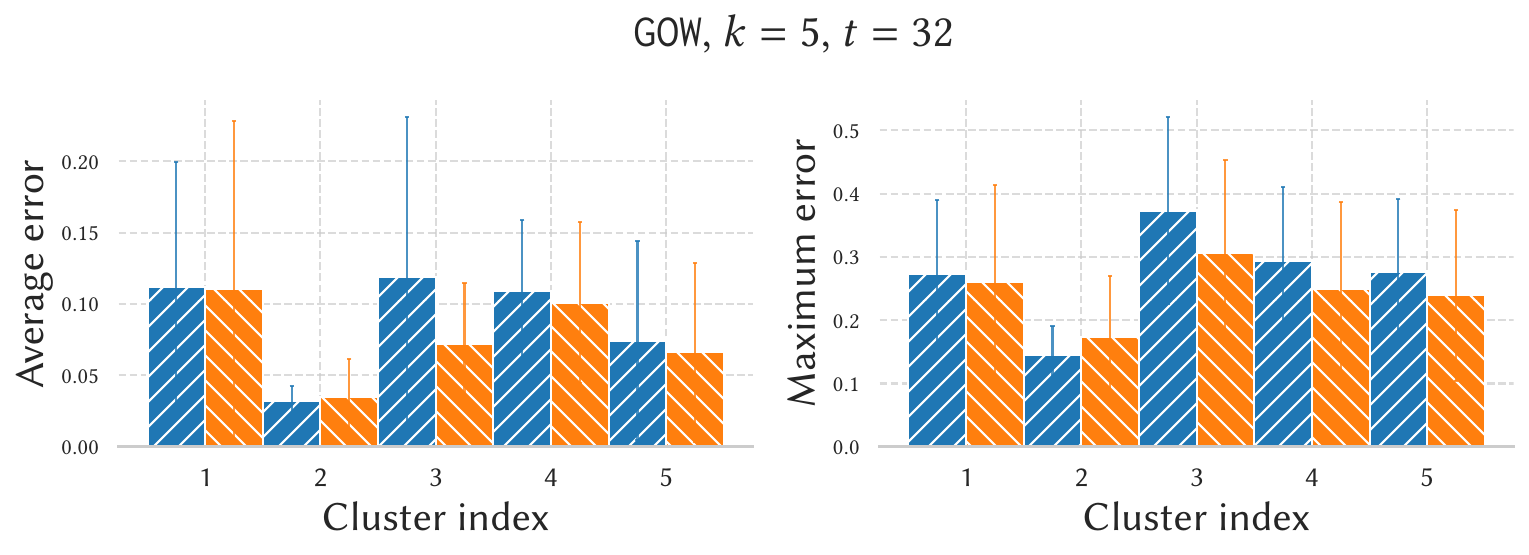} &
		\includegraphics[width=0.45\textwidth]{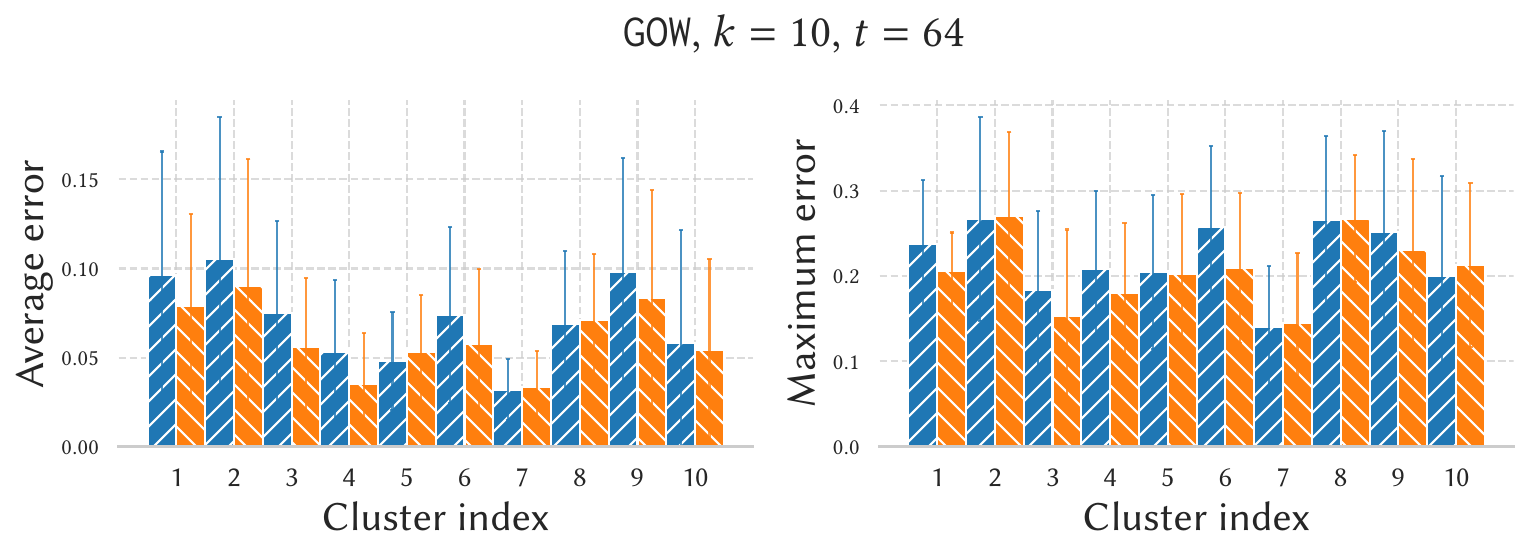}\\
		\includegraphics[width=0.45\textwidth]{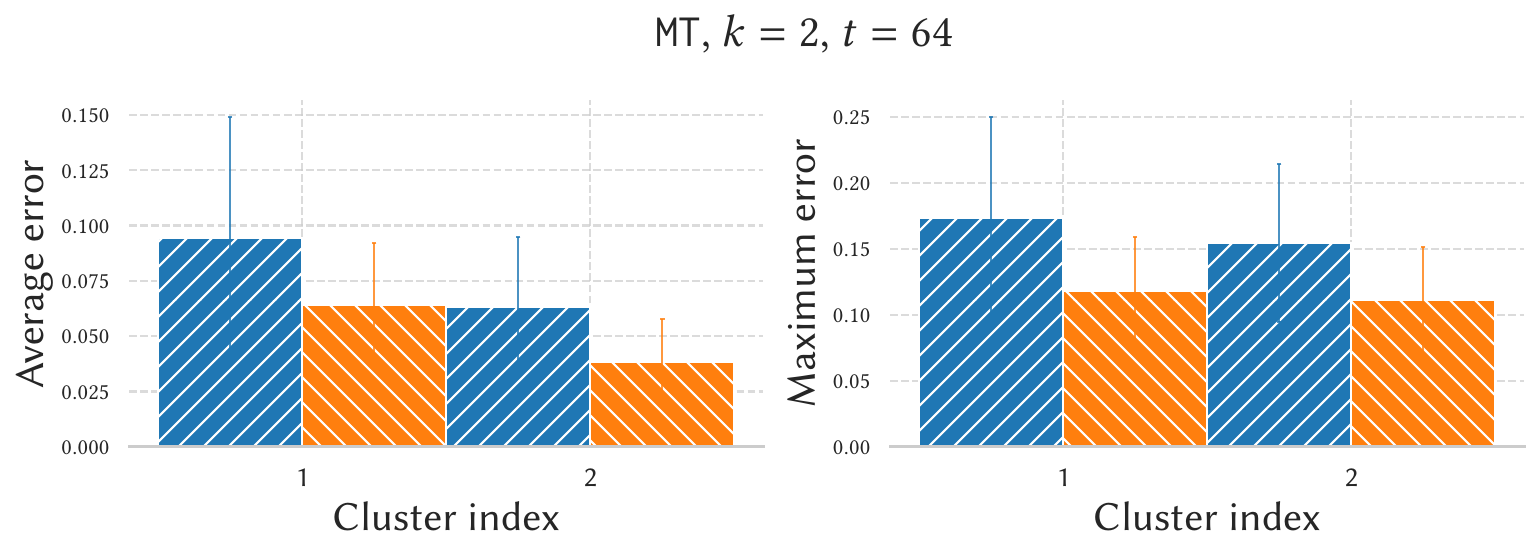} & 
		\includegraphics[width=0.45\textwidth]{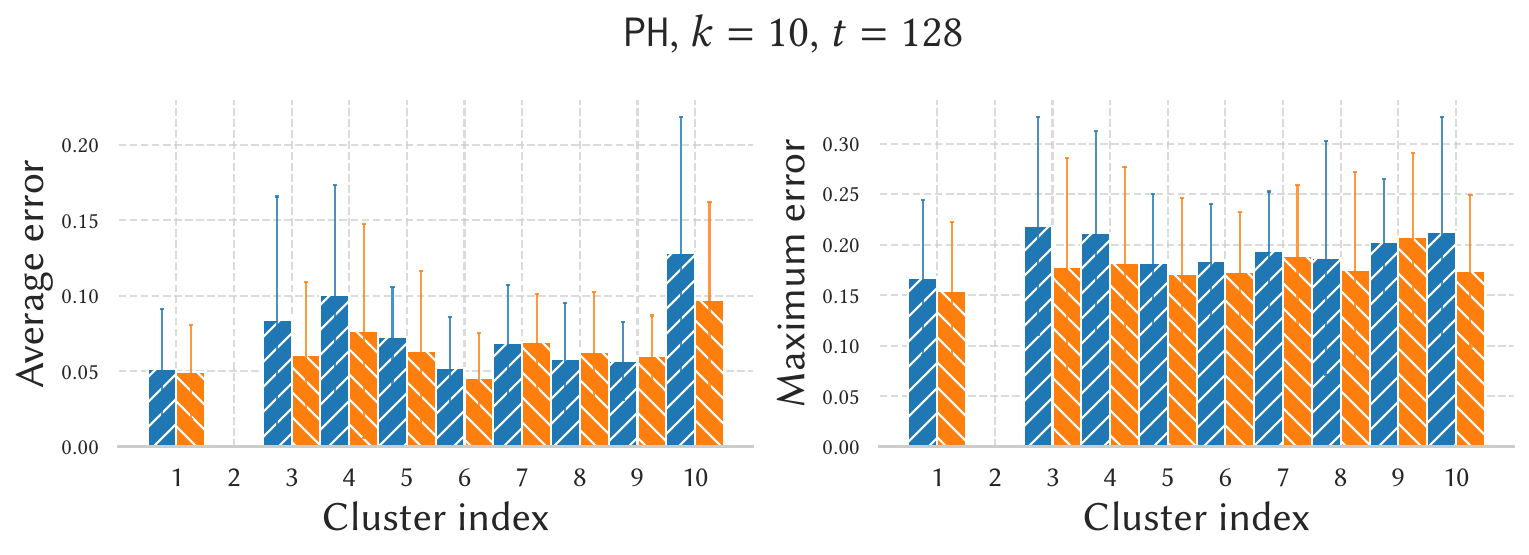}
		\\
	\end{tabular}
	\caption{Plots of the average and maximum errors within each cluster achieved by \PPSbuck and \UNIbuck, over ten independent runs, for various input configurations. Each dot represents a distinct run.
	}
	\label{fig:clustErrMaxandAvg}
\end{figure}

\subsubsection{Accuracy over clusters}\label{subsubsec:clustsAcc}

Next we investigate the local approximation error of \PPSbuck and \UNIbuck across each cluster of $\clust$. 
That is, for each cluster we computed both the maximum error $\max_{e\in C} \{|\sElEst{e}-\sEl{e}|\}$  and the average error $ (1/|C|) \sum_{e\in C} |\sElEst{e}-\sEl{e}|$ (where $\sElEst{e}$ is computed either with \UNIbuck or \PPSbuck).

In~\Cref{fig:clustErrMaxandAvg}, we visualize the maximum and average errors for different input configurations, over ten independent runs. We observe the \PPSbuck method reports estimates with comparable or significantly smaller average or maximum error compared to \UNIbuck. 
As an example, on the \mtdata dataset with $k=2$, the average and maximum errors of \PPSbuck for Cluster 1 are respectively bounded by 0.12 and 0.21, while, in contrast, those of \UNIbuck are bounded by 0.17 and 0.28, respectively. 
Similarly, on the \gowdata dataset with $k=10$, on all but cluster-index 7, our \algLocal  always attains smaller average and maximum errors with respect to~\UNIbuck. We report additional results in~\Cref{appsec:additRes}. 

\textbf{Summary for issue \Qtwo.}
We investigated the accuracy for the extremely challenging task of computing the local silhouette value of \emph{all} elements $e\in V$ (i.e., \Cref{probl:LocalEst}).
Our experiments show that our~\algLocal algorithm provides highly-accurate estimates for all elements $e\in V$, with significantly reduced variance compared to the \UNIbuck baseline. 
That is, the estimates of \PPSbuck are more accurate independently of the actual silhouette values (\Cref{subsubsec:exactBuckets}), and independently of the cluster to which the elements belong (\Cref{subsubsec:clustsAcc}), which are desirable properties for practical applications.

\subsection{Distributed implementation}\label{subsec:parExps}
As discussed in \Cref{sec:MR}, our~\algLocal can be efficiently implemented in a distributed setting, such as the one modeled by the MapReduce/MPC frameworks, using a constant number of rounds. In this section, we assess the scalability of this distributed implementation of~\algLocal, thus dealing with  Issue \Qthr. 
Since our experiments are performed on a single machine we rely on a simple parallel implementation retaining the steps in~\Cref{sec:MR} (see~\Cref{appsec:missingsetup} for details of the platform used, and~\Cref{appsec:parallelImp} for more details on our parallelization strategy). 

Our goal is to study the parallel speedup of~\algLocal varying the number of threads in $\{2,4,8,16\}$, varying the number of clusters $k\in\{2,5,10,15,20\}$, and varying $t\in \{512,1024\}$. 
We considered our largest datasets from~\Cref{tab:data}. 
For each configuration we perform five runs, and recorded the speedup over the sequential implementation. 
That is, let $T_i$ be the average time (over the five runs) to execute~\algLocal with $i\in \{1,2,4,8,16\}$ threads, i.e., $T_1$ is associated with the sequential implementation. 
We will study the average speedup $T_1/T_i$.

The results are reported in~\Cref{fig:parallelSU}. We observe a remarkable speed-up over the sequential implementation of the parallel algorithm up to 8 threads on all $k$ values (except for $k=2$, see discussion below). 
More specifically, \algLocal achieves an almost linear speedup using up to 8 threads. 
The speedup becomes slightly smaller than linear for a larger number of threads, which is expected due to the increased overhead introduced by synchronization. 
Next, we observe that the parallel speedup is significantly greater for larger values of $k$ and $t$, since our implementation parallelizes the work over different clusters in $\clust$ (see~\Cref{appsec:parallelImp}). 
That is, for smaller values of $k$ and $t$ the degree of parallelism in the computation is not sufficient to compensate the overheads incurred by the multi-threaded implementation, i.e., synchronization. 
Finally, we note that a larger value of $t$ leads to a higher speedup (e.g., on \mtdata with $t=1024$ the speedup with 16 threads is about 14 for $k=20$). 
Remarkably, when using $t=1024$ on \powdata with $k\ge 5$ and 16 threads, we obtain a speedup of at least $10\times$ compared to the sequential execution of~\algLocal. 
Notably, over such configuration, our parallel implementation of~\algLocal computes extremely accurate estimates in less than 30 seconds (see~\Cref{tab:runtimesParallel} in \Cref{appsec:additRes}, and recall that the \powdata dataset cannot be processed by the exact algorithm even within hours).

\begin{figure}[t]
	\centering
	\captionsetup[subfigure]{labelformat=empty}
	\includegraphics[width=0.95\textwidth]{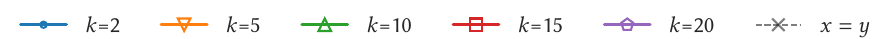}\\
	\begin{tabular}{lccr}
		\includegraphics[width=0.232\textwidth]{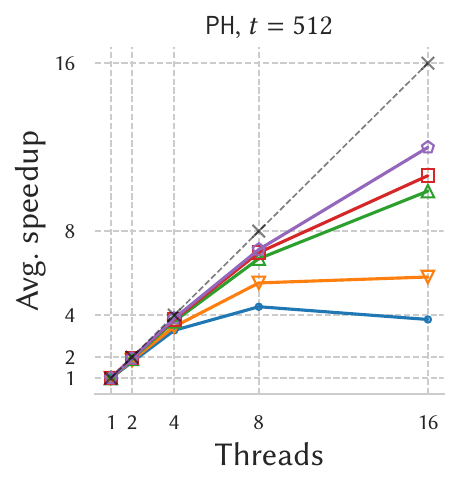} &
		\includegraphics[width=0.22\textwidth]{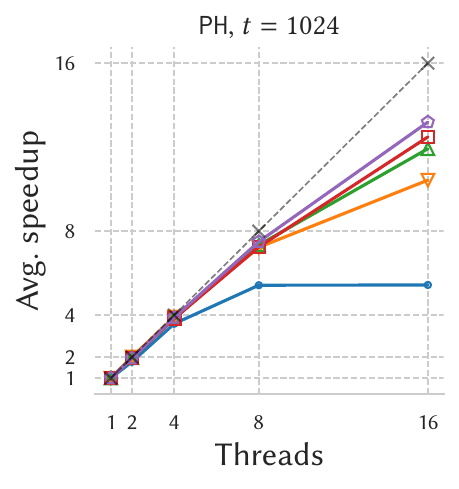}&
		\includegraphics[width=0.22\textwidth]{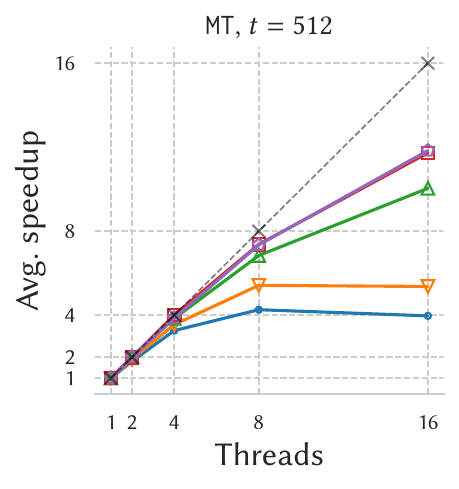} &
		\includegraphics[width=0.22\textwidth]{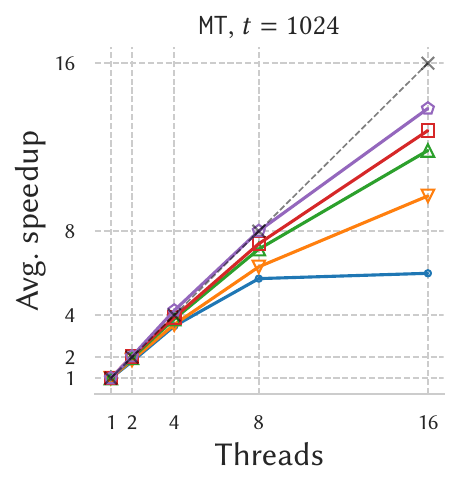}\\
	\end{tabular}
	\caption{Average parallel speedup over five independent runs by running~\algLocal with a parallel implementation. Each configuration shows different values for the sample size $t$. 
	}
	\label{fig:parallelSU}
\end{figure}

\textbf{Summary for issue \Qthr.} A simple distributed implementation of~\algLocal achieves remarkable speedups over its sequential version. 
The parallel speedup and the high accuracy achieved by \algLocal enable the estimation of the silhouette of all elements of massive datasets. 
That is, our method enables an extremely efficient processing of datasets that cannot be analyzed through current (sequential) methods.

\begin{figure}[t]
	\begin{tabular}{l}
		\includegraphics[width=0.95\textwidth]{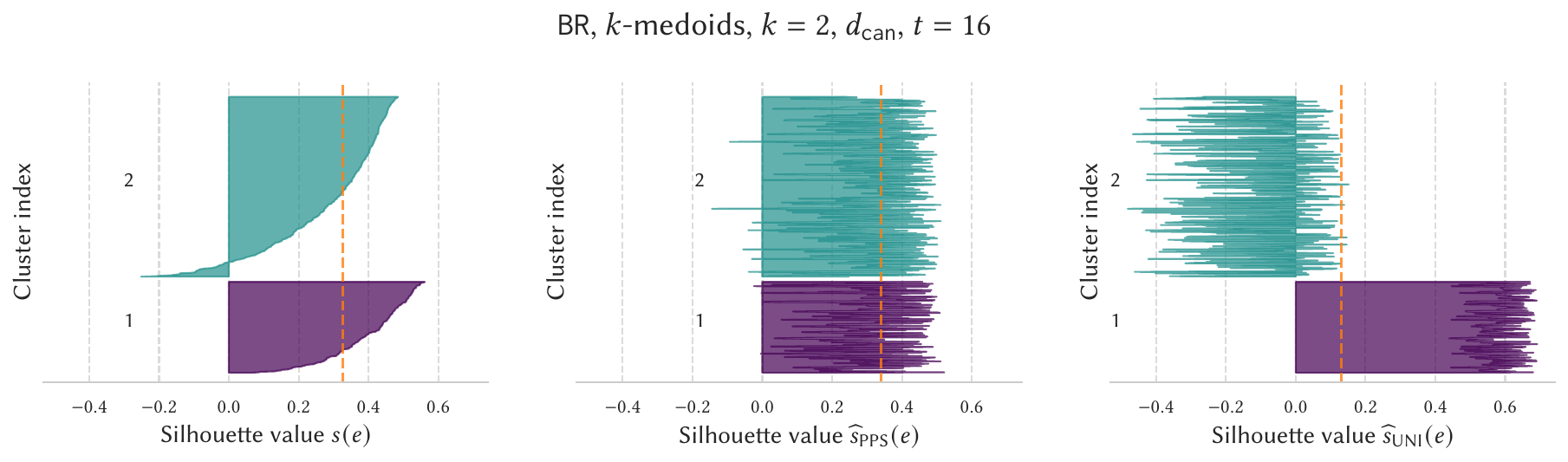} \\
		\includegraphics[width=0.95\textwidth]{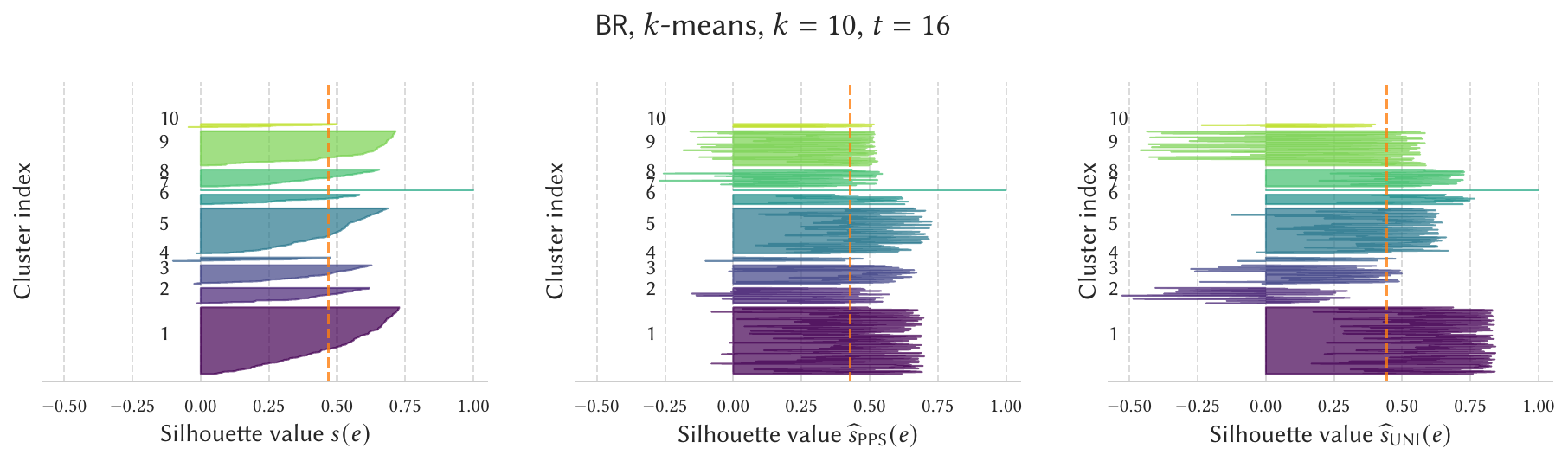} \\
	\end{tabular}
	\caption{Silhouette plot construction. (Left): Silhouette plot obtained from the exact algorithm. (Center): Silhouette plot obtained by~\algLocal, where elements are ordered according to the Left plot. (Right): Silhouette plot obtained by~\UNIbuck, where elements are ordered according to the Left plot. For all plots, the dashed line represents the average value, computed over all elements.}\label{fig:silhReconst}
\end{figure}
\begin{figure}[t]
	\centering
	\captionsetup[subfigure]{labelformat=empty}
	\includegraphics[width=0.5\textwidth]{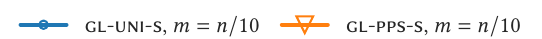}\\
	\begin{tabular}{lc}
		\includegraphics[width=0.4\textwidth]{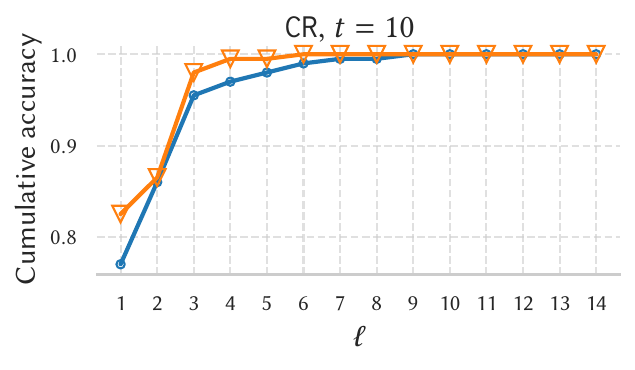} & \includegraphics[width=0.4\textwidth]{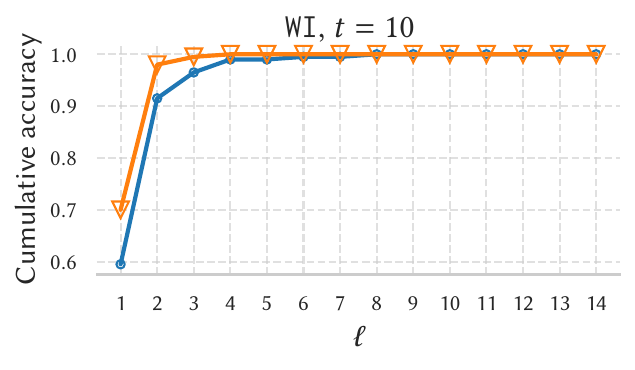}
		\\
	\end{tabular}
	\caption{Plots of $\text{Cumulative accuracy}(\ell, \widehat{\bm{k}})$, where $\ell \in [1,14]$ and $\widehat{\bm{k}}$ is either $\kvecPPS$ or $\kvecUNI$. For each value of $\ell$, the plots report the cumulative accuracy of each estimation method over 200 runs: values closer to 1, denote a higher accuracy for the identification of the best value of $k$. 
	}\label{fig:cumulativeAcc}
\end{figure}
\subsection{Applications}\label{subsec:applicationsExps}

\subsubsection{Silhouette plot construction}\label{subsubsec:expReconst}

We now show that the local estimates in output to~\algLocal can be used to build a silhouette plot of a clustering \clust that closely matches the exact one (e.g., see an example of silhouette plot in \Cref{subfig:silhPlot}).  
In this setting, we focus on \brdata, our smallest dataset, for ease of visualization,\footnote{To avoid cluttered plots, it is good practice to visualize only a small subset of the elements of $V$ over silhouette plots; even on large datasets.} and for the fact that we can compute the exact values of  $\sEl{e}$, for all $e\in V$. 

\Cref{fig:silhReconst} reports an example of two clusterings \clust, obtained respectively with \kmed (top) and with \kmeans (bottom). 
In the top panel, we observe that the estimates provided by \algLocal are generally more accurate than those provided by \UNIbuck, whose estimates are rather far from their actual values, especially for Cluster 2. 
In the bottom panel, we first note that both \algLocal and \UNIbuck provide similar values for the estimate of the global silhouette coefficient $\sEl{\clust}$ (dashed line). Interestingly, the two methods yield rather different silhouette plots.
Namely, the plot obtained with the estimates of \algLocal is much closer to the one obtained with the exact algorithm. 
For example, the estimates of \algLocal on Cluster 1 never exceed 0.75, in contrast with \UNIbuck, whose estimates can be larger than 0.8
(note that the exact values for Cluster 1 are all less than $0.75$). Similarly, for Cluster 9, \algLocal assigns negative values to very few elements as opposed to \UNIbuck.

\subsubsection{Selection of the best value of $k$}\label{subsubsec:expClustIdentification}

Let $V$  be a dataset and let $K=\{k_1,\dots, k_{|K|}\}$ be a set containing different values of $k$ used to cluster $V$. 
The global silhouette is often used to select the value of $k\in K$ that yields the best clustering for $V$.
Specifically, for each $k\in K$ we obtain a clustering $\clust_k$ (e.g., using \kmeans) and select $\clust_{k^*}$ as the best clustering for $V$, where $k^* =\arg\max_{k\in K}\sEl{\clust_k}$ \citep{Shahapure2020SilhQuality}.
In this subsection we compare the effectiveness of \algGlobSub versus \algGlobUNISub in detecting the best value of $k$ to be used for clustering. 
We focus on these two estimators, since, as argued before, they are accurate and their computation is scalable (see summary for Issue \Qone in~\Cref{subsec:globalExps}). 
Thus, they can be used to compute high-quality approximations of $\sEl{\clust_k}$ for many different candidate values of $k$.

We considered two datasets for which we can compute $\sEl{\clust_k}$ exactly and efficiently, namely \texttt{CR} and \texttt{WI}. 
We clustered each dataset with \kmeans, with different values of $k \in K$, with $K=\{2,\dots, 15\}$, and for each resulting clustering $\clust_k$ we computed the exact global silhouette $\sEl{\clust_k}$. 
Then, for each of the two algorithms \algGlobSub and \algGlobUNISub, both run with $m= n/10$ and $t=10$, we performed $N=200$ independent runs for each distinct value of $k$. 
Let $\sElEst{\clust_k}_{j}$ be the estimate of either \algGlobSub or \algGlobUNISub on dataset $V$ clustered with value $k\in [2,15]$ at iteration $j\in [1,N]$. 
We let $\kvecEx$ be the vector of the 14 values of $k \in K$, sorted by non-increasing values of the corresponding global silhouette, that is,
$\sEl{\clust_{\kvecEx_1}} \ge \sEl{\clust_{\kvecEx_{2}}} 
\ge \cdots \ge \sEl{\clust_{\kvecEx_{14}}}$. 
Similarly, we let $\kvecPPS(j)$ (resp. $\kvecUNI(j)$) denote the vector
of the 14 values of $k\in K$, sorted by non-increasing values of the estimates of the corresponding global silhouette  provided by \algGlobSub (resp. \algGlobUNISub) in the $j$-th run, for $j \in [1,N]$.
Note that, given a run $j\in[1,N]$ the \emph{best value} of $k\in K$ returned by an approximate method corresponds to $\kvecPPS_1(j)$, or $\kvecUNI_1(j)$. 
To assess the accuracy of \algGlobSub and \algGlobUNISub for selecting a reasonable value of $k\in K$ with respect to the exact approach, we now define a \emph{cumulative accuracy} function of each approximate method. 
Given $\ell \in [1,|K|]$, the cumulative accuracy evaluated at $\ell$ corresponds to the fraction of runs in which the approximate method selects as best value $k\in K$ a value in the top-$\ell$ positions of $\kvecEx$ (i.e., $k\in \{\kvecEx_1,\dots,\kvecEx_\ell\}$). More formally given the set of $N$ sorted vectors $\widehat{\bm{k}}(j)$, where $\widehat{\bm{k}}(j)$ is either $\kvecPPS(j)$ or $\kvecUNI(j)$, $j\in[1,N]$, the cumulative accuracy corresponds to, 
\[\text{Cumulative accuracy}(\ell, \widehat{\bm{k}}) \doteq  \frac{1}{N} \sum_{j=1}^{N} \sum_{i=1}^{\ell} \mathbf{1}\left[\widehat{\bm{k}}_1(j) = \kvecEx_i\right] \enspace .
\]
Where $\bm{1}[\cdot]$ corresponds to the indicator function. 
Observe that, as a function of $\ell$, the cumulative accuracy eventually hits the value of $1$ for $\ell = |K|$. 
Clearly, a function approaching $1$ very quickly indicates that the associated estimation method features high precision and a small variance in selecting a good value of $k$, according to $\kvecEx$. 

Results are reported in~\Cref{fig:cumulativeAcc}, where we observe that the~\algGlobSub method has significantly higher accuracy than~\algGlobUNISub. 
For example, more than 97\% of the 200 runs of \algGlobSub yield the selection as best value of $k$ one of the top-3 values in $\kvecEx$ on dataset \crdata, and one of the top-2 values in $\kvecEx$ on the \winedata dataset. 
In contrast, \algGlobUNISub achieves smaller accuracy: less than 95\% of the 200 runs correctly yield the selection of one the top-3 values of $k$ in $\kvecEx$ on the \crdata and less than 92\% of the runs detect one of the top-2 values in $\kvecEx$ on the \winedata dataset.
We observe that the global accuracy is affected by how the exact silhouette values vary across the top-$\ell$ elements of $\kvecEx$ (we show the ordering of $\kvecEx$ in~\Cref{tab:appclustid}). 
For example, in the above experiments the silhouette values associated with the top-2 values of $\kvecEx$ 
do not differ much, making the detection of the \emph{best} $k$ extremely challenging, especially for~\algGlobUNISub. 

\textbf{Summary for issue I4.} We studied two different and widely adopted applications of the silhouette coefficient.
First, we showed that the local estimates of~\algLocal can be used to approximate the silhouette plots of a clustering \clust, matching exact plots. 
Second, we showed that our algorithm~\algGlobSub for global silhouette estimation can be safely used in lieu of the expensive exact algorithm to detect the best clustering granularity (i.e., the value $k$ yielding higher $\sEl{\clust_k}$) for a given dataset $V$, improving over existing methods.  
Our experiments showcase the large applicability of the methods designed in this work for both local and global estimation scenarios.

\section{Conclusions}
\label{sec:conclusions}
In this work, we introduced an efficient, sampling-based
algorithm to estimate the silhouette of \emph{all} elements of a clustered dataset.
Our new method (\algLocal),
dramatically reduces the quadratic complexity required by the exact computation, while providing quantitative guarantees on the approximation error. 
In addition, we introduced several efficient estimators for the global silhouette providing a tight characterization of their required sample size.

Finally , we designed a distributed variant of our techniques, which requires small (sublinear) local memory, linear aggregate memory, and a constant number of computation rounds.

We validated our techniques through extensive experiments that show the accuracy, efficiency, and parallel scalability of our new methods for local and global silhouette estimation.
For the local silhouette values, we showed that \algLocal provides highly accurate estimates compared to existing (heuristic) methods.
For the global silhouette value, we showed that \algGlobSub (based on~\algLocal's sampling schema) often achieves the best trade-off between number of samples processed and accuracy. 
That is, \algGlobSub reports highly accurate estimates with small variance while being extremely efficient. 
We performed further experiments that showcase the use of our methods for two widespread applications:
\begin{inparaenum}[$i$)]
	\item the (approximate) construction of silhouette plots; and
	\item the selection of the best clustering granularity $k$ for a dataset.
\end{inparaenum}

These are some open issues deserving further investigation: 
\begin{inparaenum}[$i$)]
	\item devising  a more flexible, dynamic version of the algorithm where the desired accuracy can be incrementally refined by reusing (part of) the previously sampled elements;
	\item designing data-dependent sampling approaches, that adapt their sample complexity to the input datasets;
	\item extending  our methods to other popular clustering validation measures, which we discuss  in more detail in~\Cref{sec:gen}.
\end{inparaenum}

\section{Acknowledgments}\label{sec:acks}
This research is funded by the Ministry of University and Research of Italy within the Complementary National Plan PNC-I.1 ``Research initiatives for innovative technologies and pathways in the health and welfare sector, D.D. 931 of 06/06/2022, PNC0000002 DARE - Digital Lifelong Prevention CUP: B53C22006440001''. 
This research is also funded by the
ERC Advanced Grant REBOUND (834862),  and 
the Wallenberg AI, Autonomous Systems and Software Program (WASP) funded by the Knut and Alice Wallenberg Foundation.

\bibliographystyle{abbrvnat}
\bibliography{biblio}

\cleardoublepage
\appendix
\section{Technical tools}\label{appsec:tools}
The next results is a standard tool for data mining and machine learning algorithms, to provide concentration for sums of independent and bounded random variables~\cite[Theorem 4.12]{Mitzenmacher2017Prob}, and is crucial for the analysis of our estimators.
\begin{theorem}[Hoeffding's bound]\label{th:Hoeffding}
	Let $X_1,\dots,X_m$ be independent random variables taking values in $[a,b]\subseteq\mathbb{R}$ such that for all $j\in[1,m], \mathbb{E}[X_j]=\mu$. Then
	\[
	\mathbb{P}\left[\left|\frac{1}{m}\sum_j X_j - \mu \right|\ge \varepsilon \right] \le 2 \exp\left( -\frac{2m\varepsilon^2}{(b-a)^2}\right) \enspace.
	\]
\end{theorem}

The next results (belonging to the large farmily of bounds known as Chernoff-Hoeffding bounds) are used in the proof of~\Cref{theo:localAlg}. The form of the bounds that we use is from the textbook by~\citet[Theorem 1.1]{DubashiP09}.

\begin{theorem}[Chernoff-Hoeffding bounds]\label{th:CHbounds}
	Let $X_1,\dots,X_m$ be independent random variables taking values in $[0,1]$ and let $X=\sum_{i=1}^m X_i$. Then for any $\varepsilon\in(0,1)$ it holds that
	\[
	\mathbb{P}\left[X \ge (1+\varepsilon) \mathbb{E}[X] \right] \le \exp\left( -\frac{\varepsilon^2}{3} \mathbb{E}[X]\right) \quad \text{ and } \quad \mathbb{P}\left[X \le (1-\varepsilon) \mathbb{E}[X] \right] \le \exp\left( -\frac{\varepsilon^2}{2} \mathbb{E}[X]\right)\enspace.
	\]
\end{theorem}

\section{Estimating other clustering validation measures}\label{sec:gen}

Our strategy that estimates the terms $\weightEl{j}, j\in[k], e\in V$ through the \PPS approach can be used to approximate various other measures for
internal clustering evaluation. 
In particular, all those measures that are based on sums of distances between elements of the dataset $V$. 
This is the case of
measures that consider the \emph{cohesion} (i.e., the average
intracluster distance) and the \emph{separation} (i.e., the average
intercluster distance) of a clustering $\clust$~\cite{TanSK06}.

More precisely, consider a $k$-clustering $\mathcal{C} = \{C_1, \dots, C_k\}$ and define the cohesion and separation of $\mathcal{C}$ as
\begin{eqnarray*} 
	\mbox{Coh}(\mathcal{C}) & = &
	\frac{1}{2} 
	\frac{\sum_{j=1}^k \sum_{e',e'' \in C_j} d(e',e'')}
	{\sum_{j=1}^k {|C_j| \choose 2}} \enspace,\\ 
	\mbox{Sep}(\mathcal{C}) & = &
	\frac{\sum_{1 \leq j_1 < j_2 \leq k}\sum_{e' \in C_{j_1}}\sum_{e'' \in C_{j_2}} 
		d(e',e'')}
	{\sum_{1 \leq j_1 < j_2 \leq k} (|C_{j-1}||C_{j_2}|)} \enspace,  
\end{eqnarray*}
respectively. 
These measures have also been used to evaluate the average cluster reliability on networks, where distances correspond to connection probabilities
\citep{LiuJAS12,CeccarelloFPPV17}.  
We can rewrite the above measures
in terms of the sums $\weightEl{}$ defined in~\Cref{eq:weights}, as
follows
\begin{align*}
	\mbox{Coh}(\mathcal{C}) = 
	\frac{1}{2}
	\frac{\sum_{j=1}^k \sum_{e \in C_j} W_{C_j}(e)}
	{\sum_{j=1}^k {|C_j| \choose 2}} \quad \text{ and } \quad     
	\mbox{Sep}(\mathcal{C}) =
	\frac{\sum_{1 \leq j_1 < j_2 \leq k}\sum_{e \in C_{j_1}} W_{C_{j_2}}(e)}
	{\sum_{1 \leq j_1 < j_2 \leq k} (|C_{j_1}||C_{j_2}|)}  \enspace . 
\end{align*}

Clearly, highly accurate approximations of terms $\weightEl{j}, j\in[k], e\in V$ with low relative errors
yield approximations with low \emph{relative error} for
$\mbox{Coh}(\mathcal{C})$ and $\mbox{Sep}(\mathcal{C})$. 
Specifically,
let 
$\widehat{\rm Coh}(\mathcal{C})$ and $\widehat{\rm Sep}(\mathcal{C})$ as the approximation of
$\mbox{Coh}(\mathcal{C})$ and $\mbox{Sep}(\mathcal{C})$ respectively. Where each approximation is obtained by
substituting the term $W_{C}(e)$
with the value $\weightElEst{}$ computed within~\algLocal. 
The following result is an immediate
consequence of~\Cref{lem:mainerror}.
\begin{theorem}
	\label{th:cohsep}
	Let $V$ be a dataset of $n$ elements, and let
	$\clust$ be a $k$-clustering of $V$. Let $\widehat{\rm
		Coh}(\clust)$ and $\widehat{\rm Sep}(\clust)$ be the
	approximations to $\mbox{Coh}(\clust)$ and
	$\mbox{Sep}(\clust)$ respectively based on the values $\weightElEst{}$
	computed within~\algLocal for parameters $0 < \varepsilon, \delta < 1$, and
	for a suitable choice of constant $c>0$ in the definition of the
	sample size $t$. Then with probability at least $1-\delta$
	\begin{align*}
		\left|
		\frac{\widehat{\rm Coh}(\mathcal{C}) - {\rm Coh}(\mathcal{C})}
		{{\rm Coh}(\mathcal{C})}
		\right| \leq \varepsilon \quad \text{ and } \quad 
		\left|
		\frac{\widehat{\rm Sep}(\mathcal{C}) - {\rm Sep}(\mathcal{C})}
		{{\rm Sep}(\mathcal{C})}
		\right| \leq \varepsilon \enspace .
	\end{align*}
\end{theorem}

\section{Reproducibility}\label{appsec:reprod}

\subsection{Datasets}\label{appsec:data}
We consider various datasets from the UCI and SNAP repositories, that are publicly available online.\footnote{\url{https://archive.ics.uci.edu/datasets/}, \url{https://snap.stanford.edu/data/}}
We provide a script to download and process each dataset with our code.\footnote{\codeRepo}
For the \kmeans objective we used the implementation from \texttt{scikit-learn},\footnote{\url{https://scikit-learn.org/}} which uses Lloyd's algorithm coupled with an initialization obtained through $k$-means++.
For the \kmed objective dataset we used the state-of-the-art approach by~\citet{schubert2022fast}.\footnote{\url{https://pypi.org/project/kmedoids/}}

\subsection{Setup}\label{appsec:missingsetup}
We implemented our algorithms and the baseline methods in C++20, and compiled it under gcc 9.5.0, 
with all optimization flags set. 
Except otherwise stated, all the experiments were performed single-threaded on a 72-core machine Intel Xeon Gold, running Ubuntu 20.04. 
To cluster the various datasets, we used Python 3.12.3. 
For the \kmeans objective, we used the default implementation from the Scikit-learn library,\footnote{\url{https://scikit-learn.org/stable/}} while for the \kmed objective we used the \kmed library~\citep{Schubert2021PAM}.\footnote{\url{https://pypi.org/project/kmedoids/}} Our distributed implementation of \algLocal is done using the OpenMP library for C++, given that our execution is on a single many-core machine.

\subsection{Parallel implementation}\label{appsec:parallelImp}
To parallelize our algorithm~\algLocal we rely on the OpenMP library. Our implementation adopts the following parallelization strategy:
\begin{itemize}
	\item We parallelize the \textbf{for} loops in~\Cref{line:initSampleW,line:forLoopPPS} of~\Cref{code:algorithm}, inspired by \textbf{Round 1} and \textbf{Round 2} of our Mapreduce algorithm;
	\item We parallelize the \textbf{for} loop in~\Cref{line:forLoopAllElements} of~\Cref{code:algorithm}, inspired by \textbf{Round 3} of our Mapreduce algorithm.
\end{itemize}
For each of the above parallel loop executions, we perform synchronization on the necessary data-structures, yielding some of the bottlenecks discussed in~\Cref{subsec:parExps}.

\section{Additional results}\label{appsec:additRes}
\begin{table}[t]
	\caption{Algorithms compared in our experimental evaluation in~\Cref{subsec:globalExps}, and their required parameters: 
		$t$ is used in Phase 1 of our~\Cref{code:algorithm};
		$m$ is used to evaluate the estimators in~\Cref{eq:est2,eq:est3}; while we denote with $m'$ the value of the sample size used to evaluate~\Cref{eq:est1}. Both parameters $t$ and $m$ apply to \PPS and its uniform Poisson-sampling variant \UNIbuck (see~\Cref{subsubsec:uniFails}).}\label{tab:paramsRecap}
	\centering
	\begin{tabular}{lccc}
		\toprule
		Algorithm name & $t$ & $m$ & $m'$\\
		\midrule
		\algGlobOne & \xmark & \xmark & \vmark\\
		\algGlobPPS, \algGlobUNI & \vmark & \xmark & \xmark\\
		\algGlobSub, \algGlobUNISub & \vmark & \vmark & \xmark\\
		\bottomrule
	\end{tabular}
\end{table}

\begin{table}[t]
	\caption{Average runtime (in seconds) required by the sequential execution of~\algLocal (i.e., $T_1$) under the setting of~\Cref{subsec:parExps}.}\label{tab:runtimesParallel}
	\centering
	\begin{tabular}{lccccc}
		\toprule
		Configuration & $k=2$ & $k=5$ & $k=10$ & $k=15$ & $k=20$\\
		\midrule
		\mtdata, $t=512$ & 35.8 &  72.5 &  254.3 &  392.6 &  506.8\\
		\mtdata, $t=1024$ & 72.8 &  223.9 &  525.7 &  702.9 &  984.8\\
		\powdata, $t=512$ & 44.4 &  109.6 &  356.8 &  561.0 &  736.9\\
		\powdata, $t=1024$ & 87.7 &  354.6 &  756.3 &  1075.2 &  1427.7\\
		\bottomrule
	\end{tabular}
\end{table}

\begin{table}[t]
	\caption{Values of $k$ ordered by $\sEl{\clust_k}$ (ordering over $\kvecEx$). The setting is from~\Cref{subsubsec:expClustIdentification}.}\label{tab:appclustid}
	\centering
	\begin{tabular}{lcccc}
		\toprule
		& \multicolumn{2}{c}{\crdata} & \multicolumn{2}{c}{\winedata} \\
		\cmidrule(lr){2-3} \cmidrule(lr){4-5}
		& $k$ & $\sEl{\clust_k}$ & $k$ & $\sEl{\clust_k}$ \\
		\midrule
		$\kvecEx_1$ & $k=2$ & 0.552   & $k=2$ & 0.51\\ 
		$\kvecEx_2$ & $k=4$ & 0.463   & $k=3$ & 0.504\\ 
		$\kvecEx_3$ & $k=3$ & 0.451   & $k=4$ & 0.45\\ 
		$\kvecEx_4$ & $k=5$ & 0.419   & $k=5$ & 0.418\\ 
		$\kvecEx_5$ & $k=7$ & 0.391   & $k=6$ & 0.392\\ 
		$\kvecEx_6$ & $k=6$ & 0.387   & $k=7$ & 0.367\\ 
		$\kvecEx_7$ & $k=8$ & 0.38    & $k=8$ & 0.357\\ 
		$\kvecEx_8$ & $k=10$ & 0.364  & $k=9$ & 0.342\\ 
		$\kvecEx_9$ & $k=9$ & 0.343   & $k=10$ & 0.334\\ 
		$\kvecEx_{10}$ & $k=11$ & 0.336  & $k=11$ & 0.329\\ 
		$\kvecEx_{11}$ & $k=15$ & 0.333  & $k=12$ & 0.326\\ 
		$\kvecEx_{12}$ & $k=14$ & 0.314  & $k=14$ & 0.325\\ 
		$\kvecEx_{13}$ & $k=12$ & 0.312  & $k=15$ & 0.325\\ 
		$\kvecEx_{14}$ & $k=13$ & 0.304  & $k=13$ & 0.323\\ 
		\bottomrule
	\end{tabular}
\end{table}

In this section we complement our extensive experimental evaluation, by presenting additional results to those of~\Cref{sec:exp}. 

\begin{itemize}
	\item In~\Cref{fig:globalEst2} we show, for all the missing large dataset of~\Cref{tab:data}. The results of the experiments in~\Cref{subsubsec:largeData} that compare the global estimates of the various randomized algorithms.
	\item In~\Cref{fig:bucketsLargeApp} we show additional results for the accuracy of \PPSbuck and \UNIbuck approaches over various buckets---extending results from~\Cref{subsubsec:exactBuckets}.
	\item In~\Cref{fig:clustErrMaxandAvgApp} we show additional results on the average and maximum error on each cluster \clust of \PPSbuck and \UNIbuck for the setting of~\Cref{subsubsec:clustsAcc}. 
	\item In~\Cref{tab:runtimesParallel} we show the average sequential runtime required by~\algLocal considering the setting of~\Cref{subsec:parExps}.
	\item In~\Cref{tab:appclustid} we show the values of $\sEl{\clust_k}$ for each $k\in \bm{k}$ under the setting of~\Cref{subsec:applicationsExps}.
\end{itemize}

\begin{figure*}[!tbp]
	\centering
	\captionsetup[subfigure]{labelformat=empty}
	\includegraphics[width=0.8\textwidth]{plots/media/largedata/loc-gowalla_totalCheckins_preprocessed/globalEsts/k2/label.pdf}\\
	\begin{tabular}{lr}
		\includegraphics[width=0.49\textwidth]{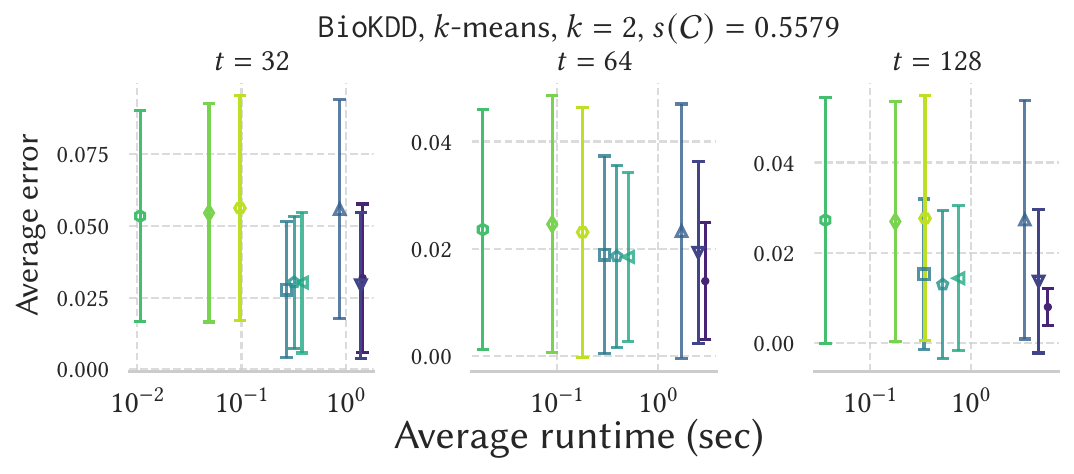} &
		\includegraphics[width=0.49\textwidth]{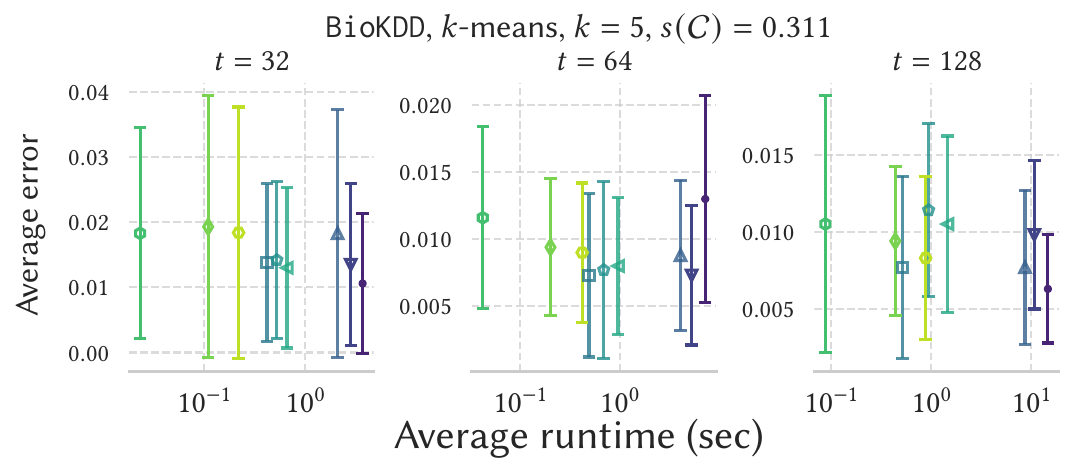}\\
		\includegraphics[width=0.49\textwidth]{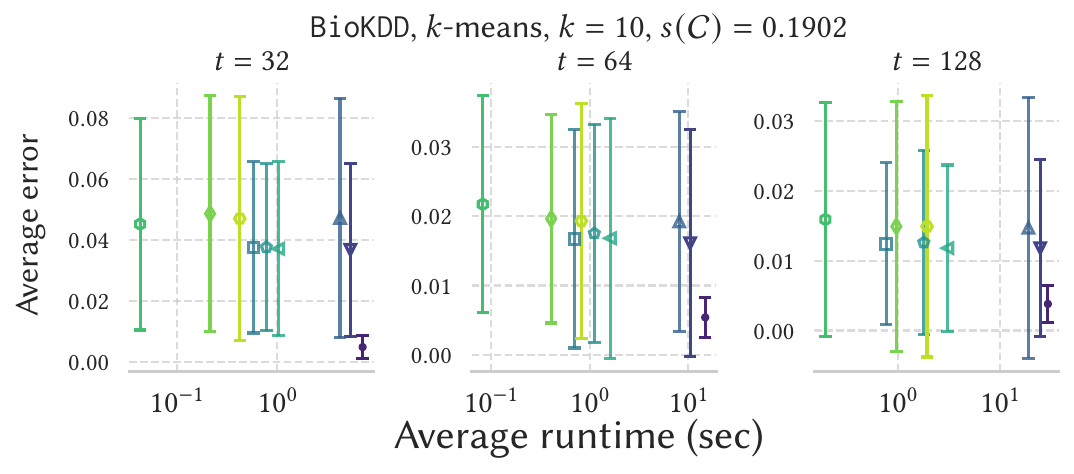} &
		\includegraphics[width=0.49\textwidth]{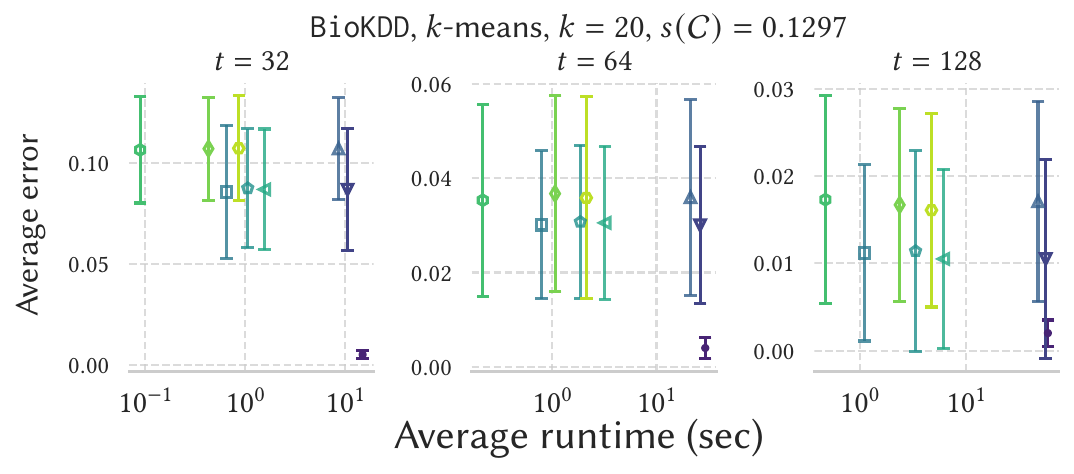}\\
		\includegraphics[width=0.49\textwidth]{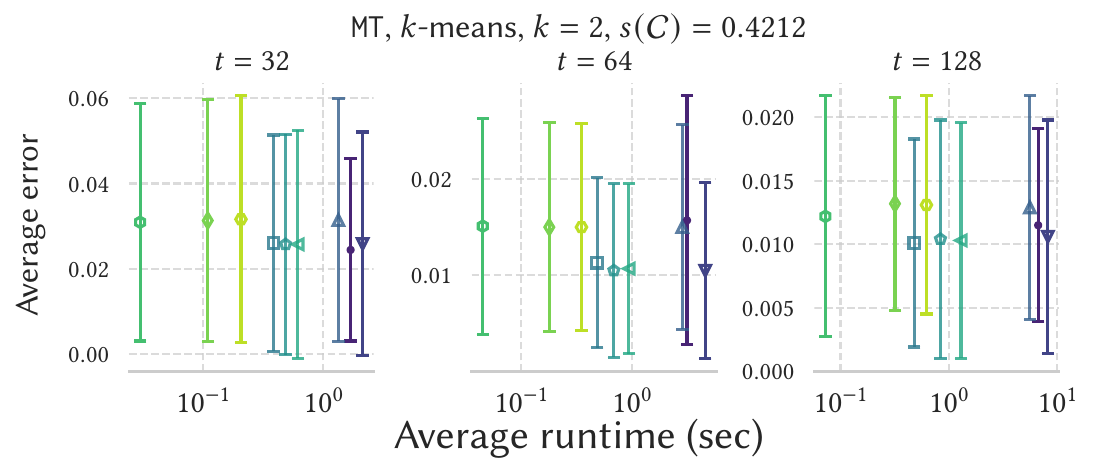} &
		\includegraphics[width=0.49\textwidth]{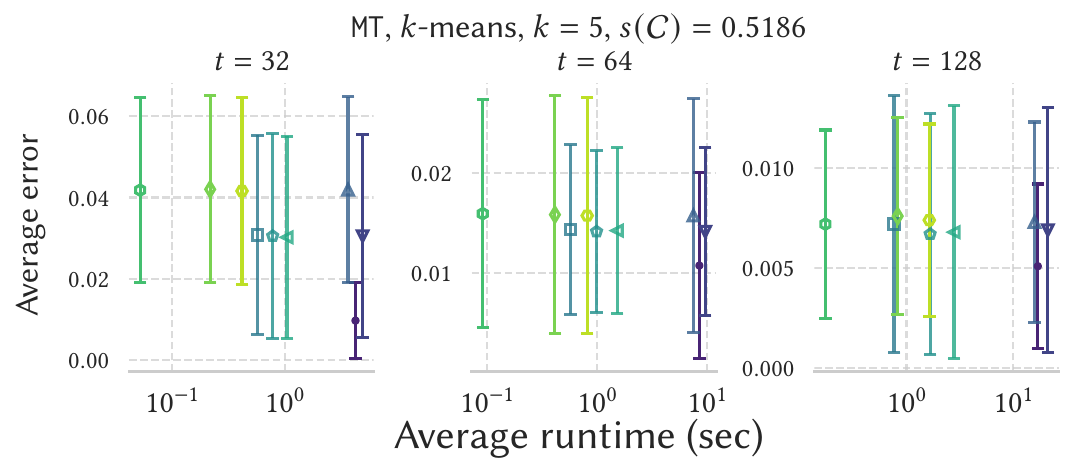}\\
		\includegraphics[width=0.49\textwidth]{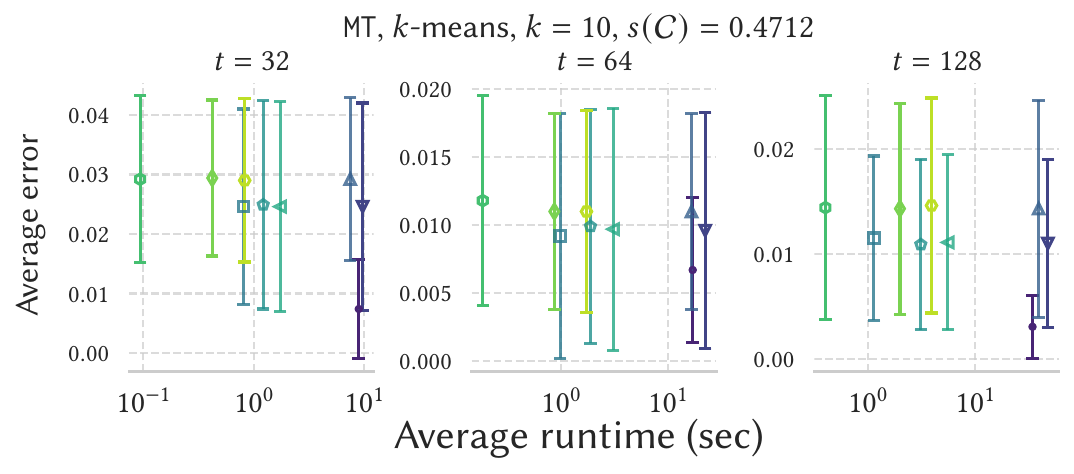} &
		\includegraphics[width=0.49\textwidth]{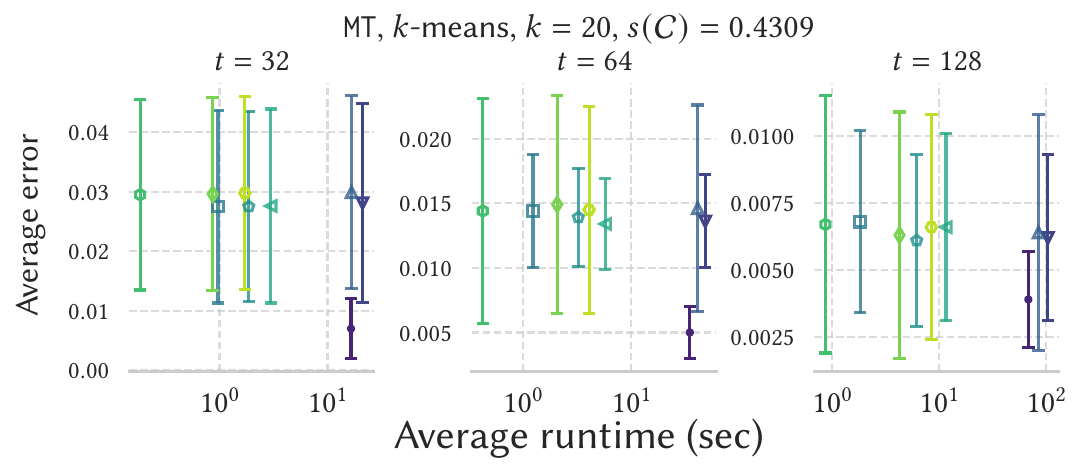}\\
		\includegraphics[width=0.49\textwidth]{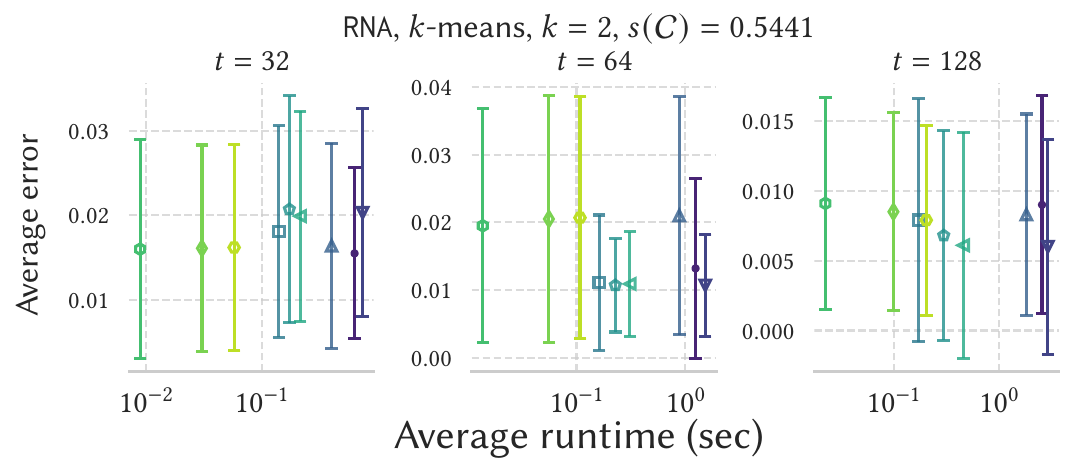} &
		\includegraphics[width=0.49\textwidth]{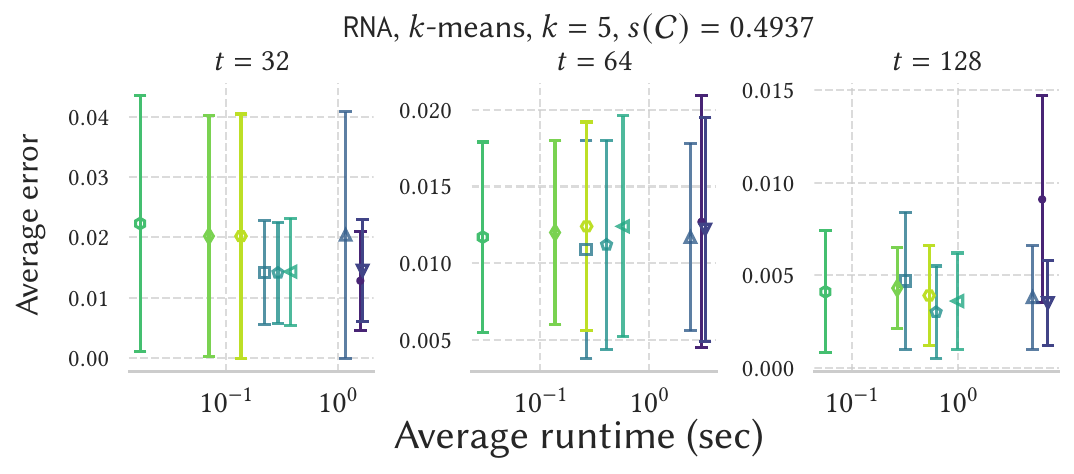}\\
		\includegraphics[width=0.49\textwidth]{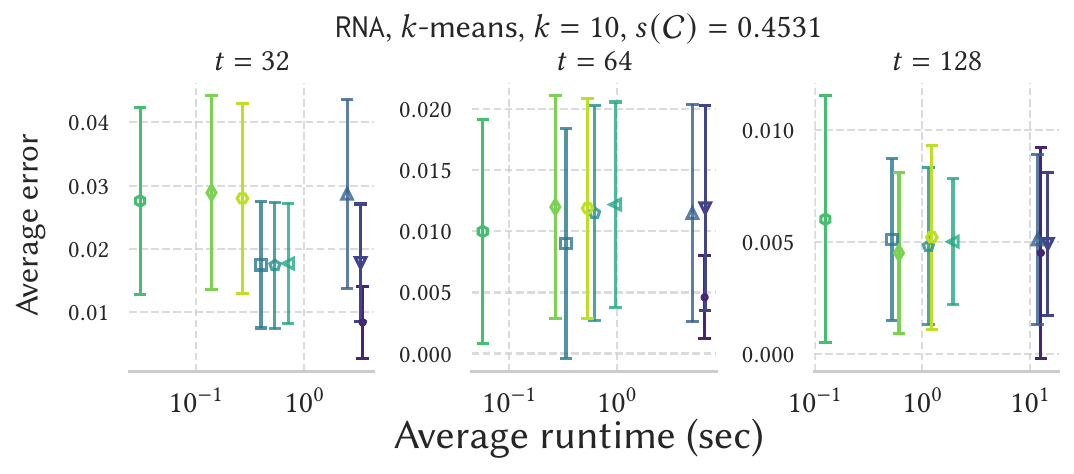} &
		\includegraphics[width=0.49\textwidth]{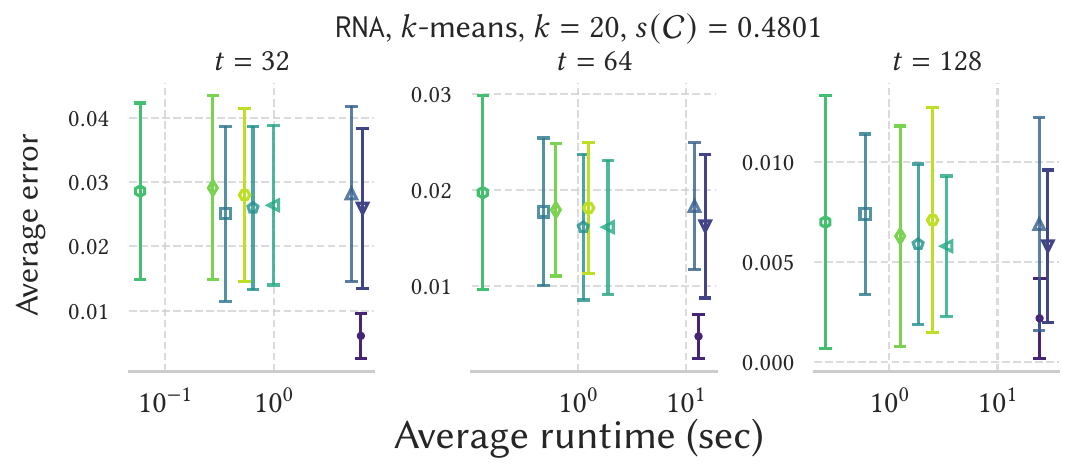}\\
	\end{tabular}
	\caption{Methods comparison. The $x$ axis is associated with the average runtime, and the $y$ axis with the average error (and its standard deviation) over 10 independent runs. For each clustered dataset we considered different values of the expected sample size $t\in \{32,64,128\}$, as illustrated.}
	\label{fig:globalEst2}
\end{figure*}

\begin{figure*}[!tbp]
	\centering
	\captionsetup[subfigure]{labelformat=empty}
	\includegraphics[width=0.25\textwidth]{plots/media/largedata/loc-gowalla_totalCheckins_preprocessed/buckets/k2/legend.pdf}\\
	\begin{tabular}{lc}
		\includegraphics[width=0.45\textwidth]{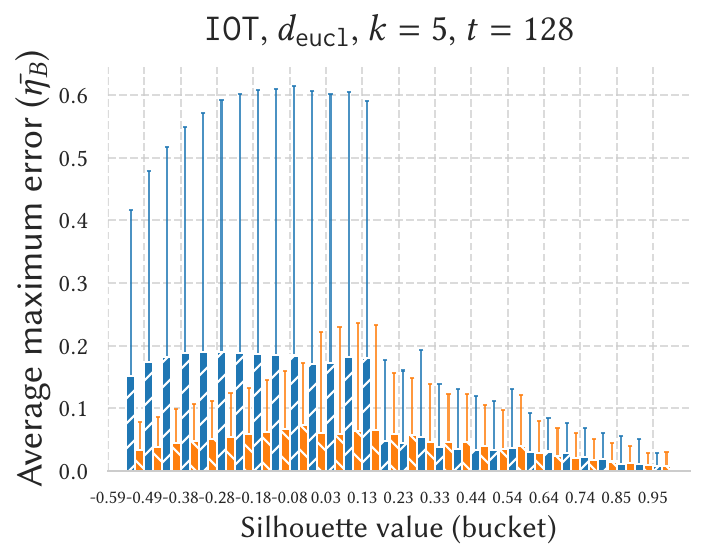} &
		\includegraphics[width=0.45\textwidth]{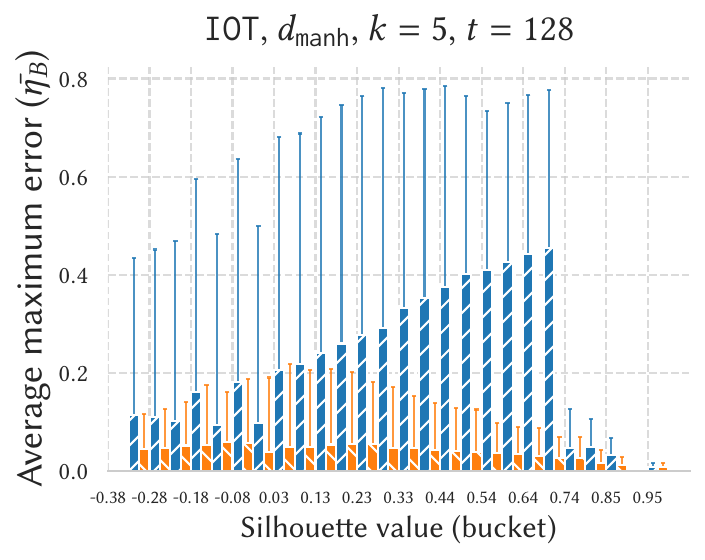} \\
		\includegraphics[width=0.45\textwidth]{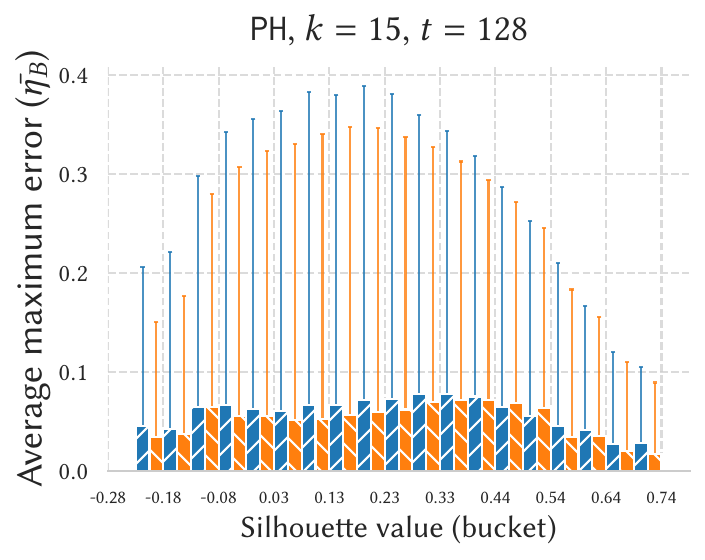} & \includegraphics[width=0.45\textwidth]{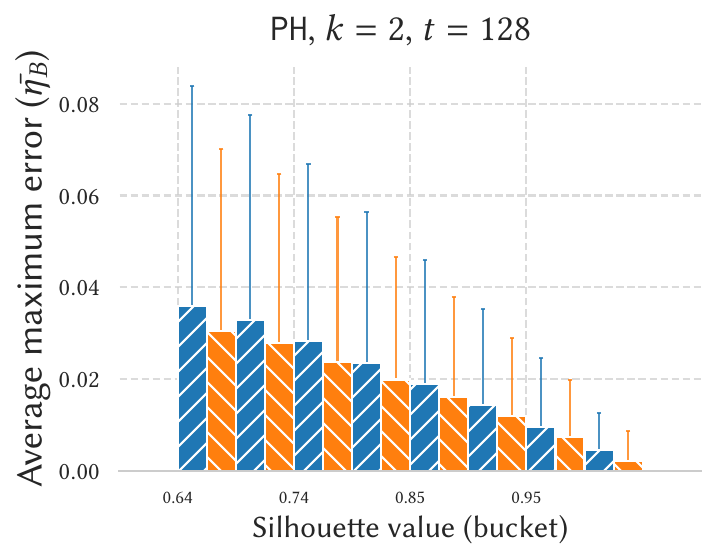} \\
		\includegraphics[width=0.45\textwidth]{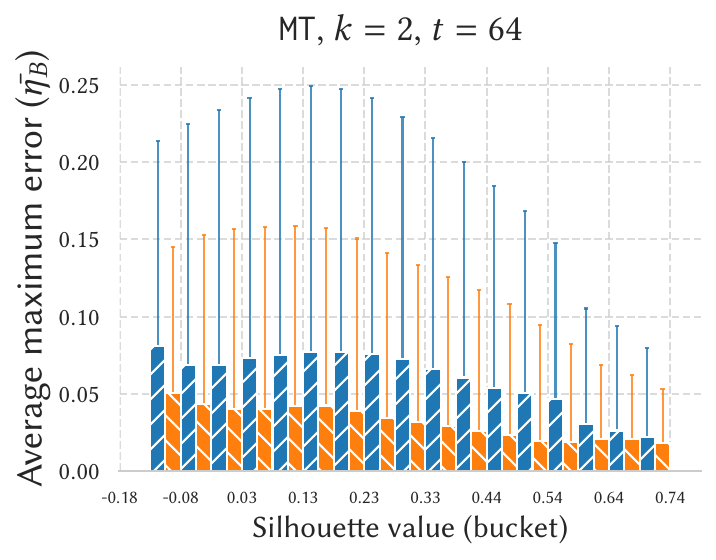} & \includegraphics[width=0.45\textwidth]{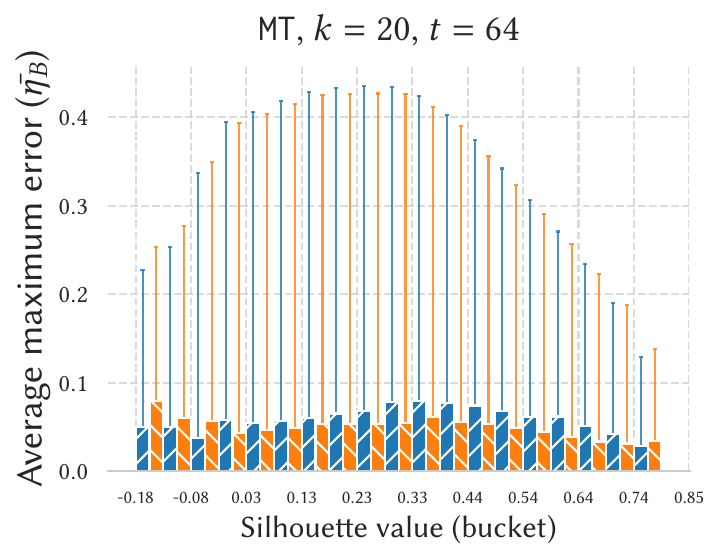} \\
	\end{tabular}
	\caption{
		Plots of the average maximum error and its standard deviation obtained on all elements $e\in V$, grouped by bucket, for \PPSbuck and \UNIbuck, over four configurations. When we report the distance, it indicates that the dataset has been clustered with \kmed. For ease of visualization we only display the average maximum error over non-empty buckets.}
	\label{fig:bucketsLargeApp}
\end{figure*}

\begin{figure*}[!tbp]
	\centering
	\captionsetup[subfigure]{labelformat=empty}
	\includegraphics[width=0.3\textwidth]{plots/media/largedata/loc-gowalla_totalCheckins_preprocessed/clustMaxAvg/k5/legend.pdf}\\
	\begin{tabular}{lr}
		\includegraphics[width=0.5\textwidth]{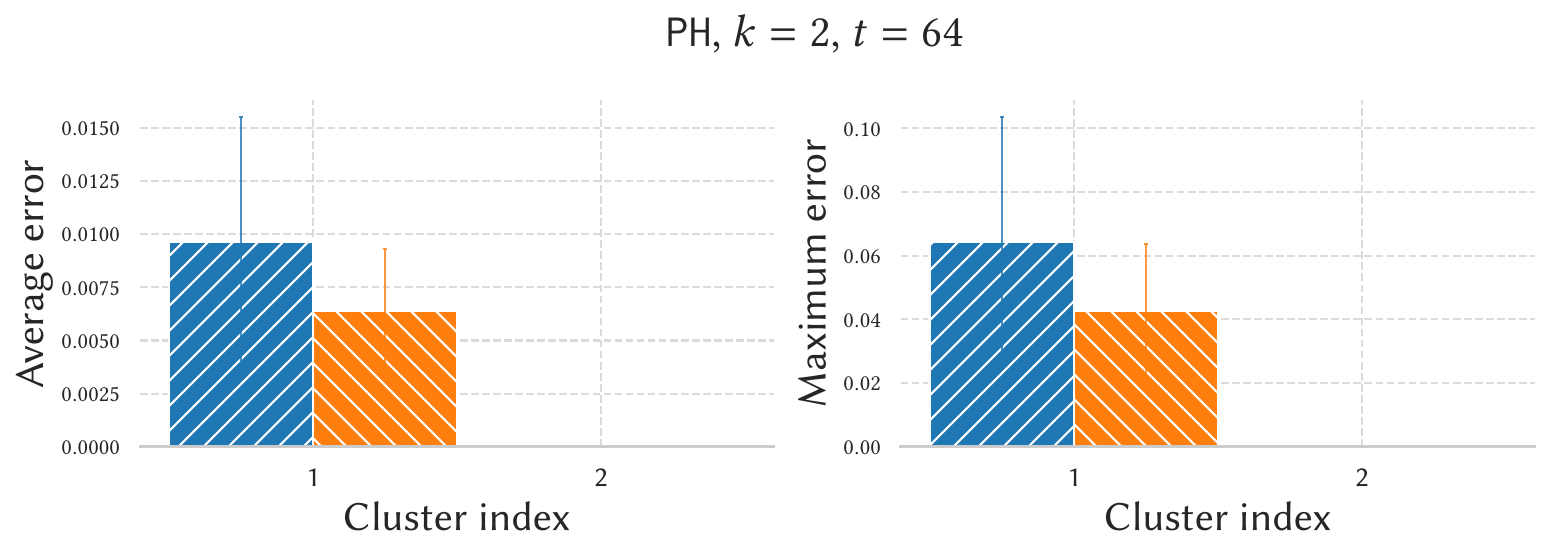} & 
		\includegraphics[width=0.5\textwidth]{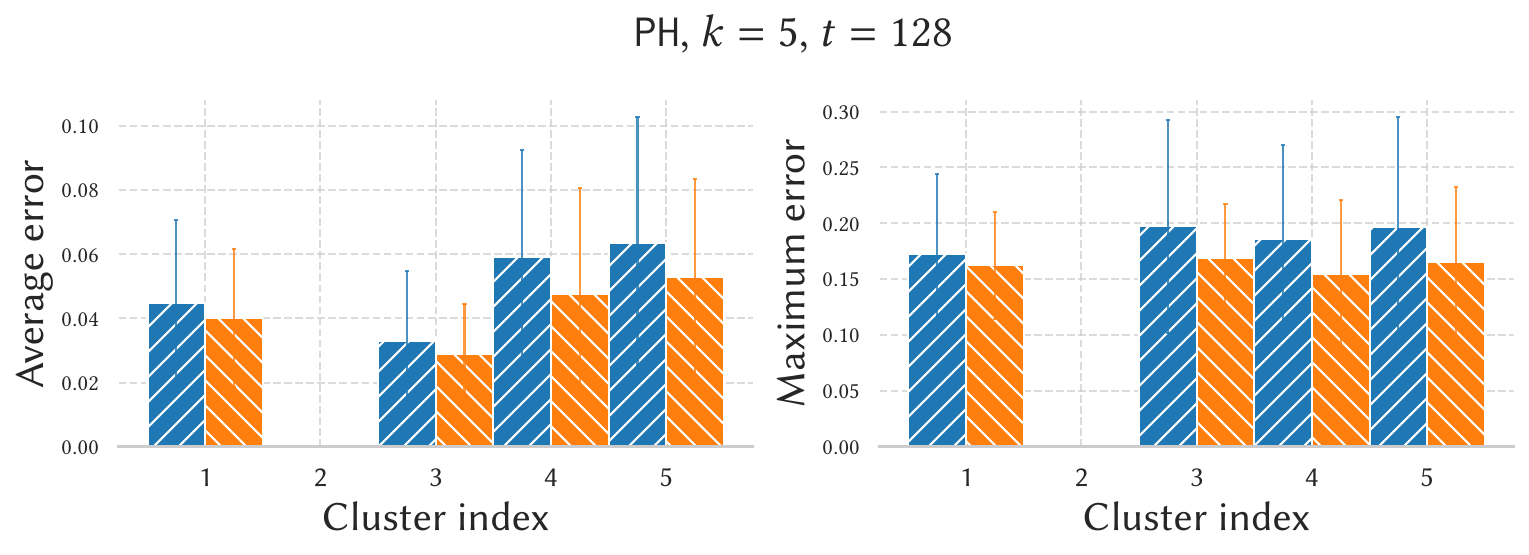}
		\\
		\includegraphics[width=0.5\textwidth]{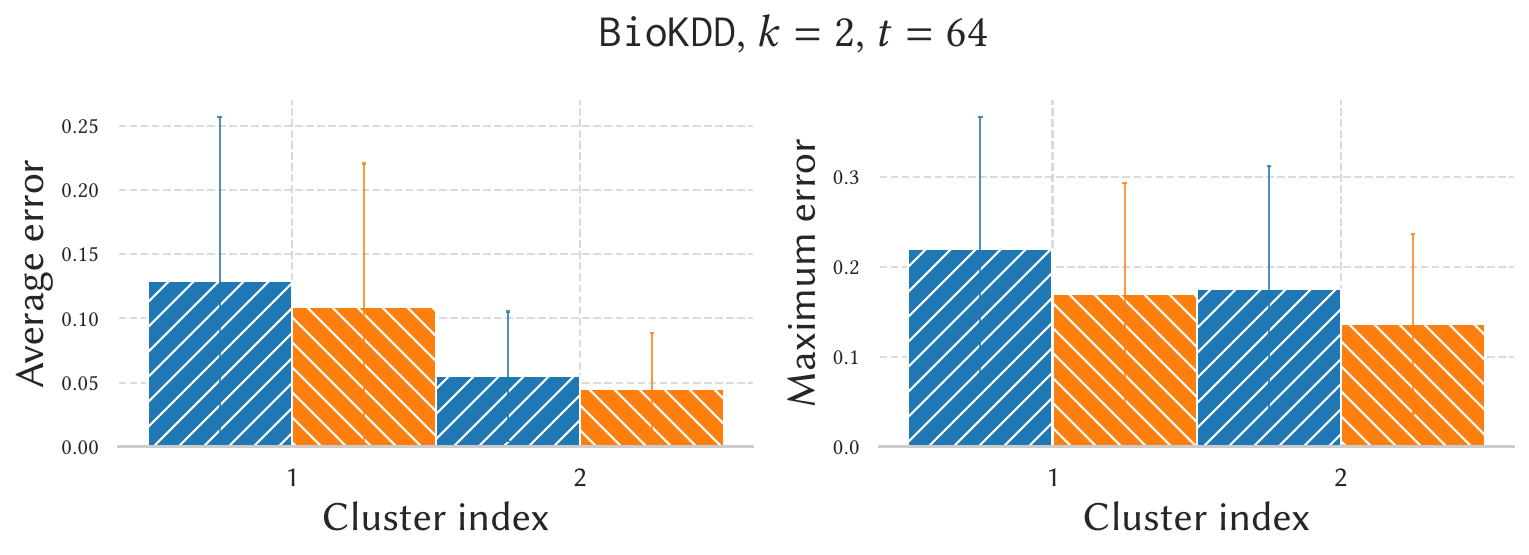} & 
		\includegraphics[width=0.5\textwidth]{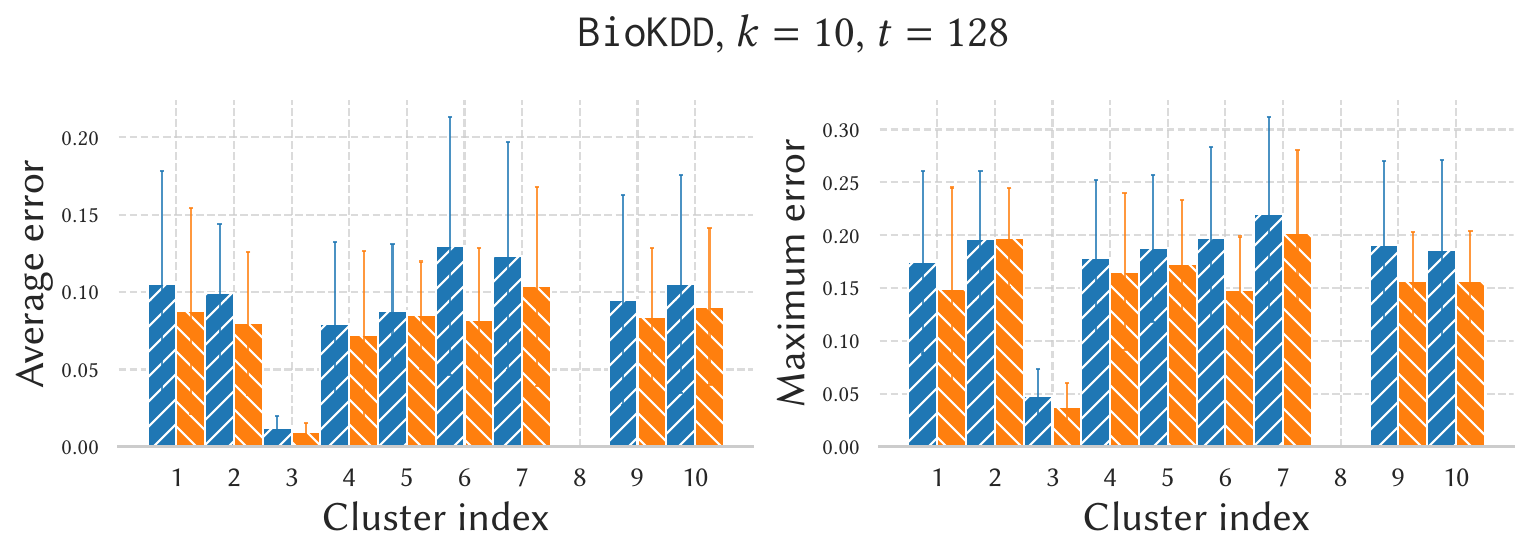}\\
		\includegraphics[width=0.5\textwidth]{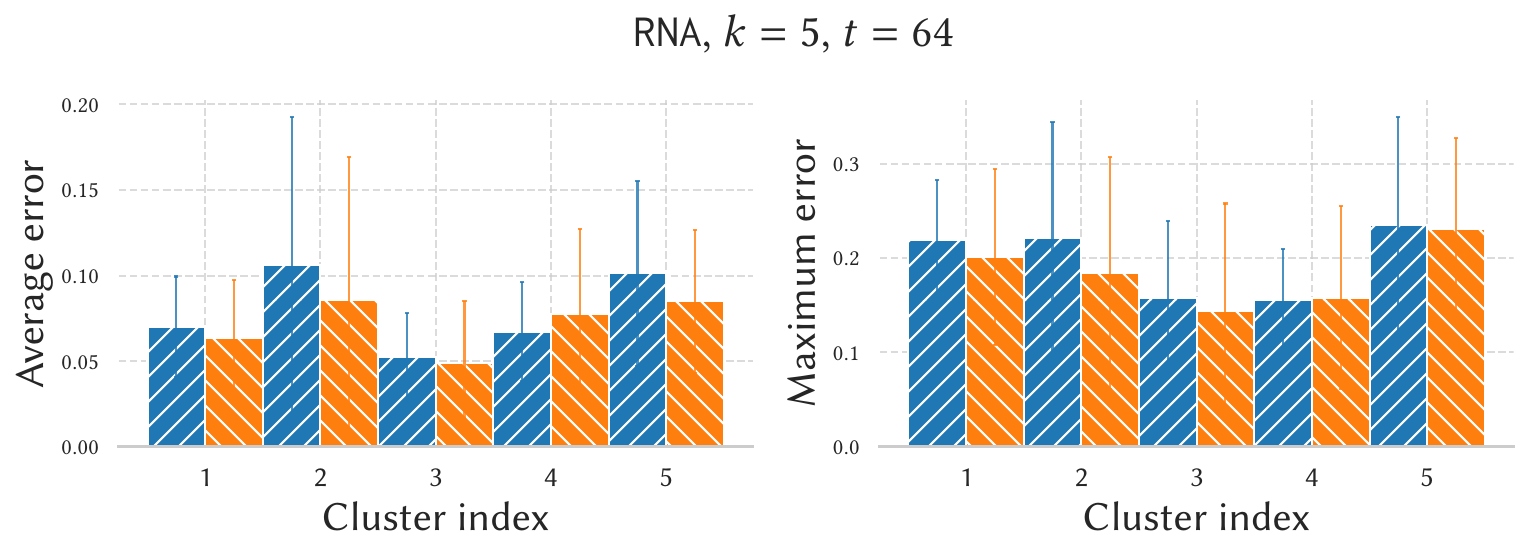} & 
		\includegraphics[width=0.5\textwidth]{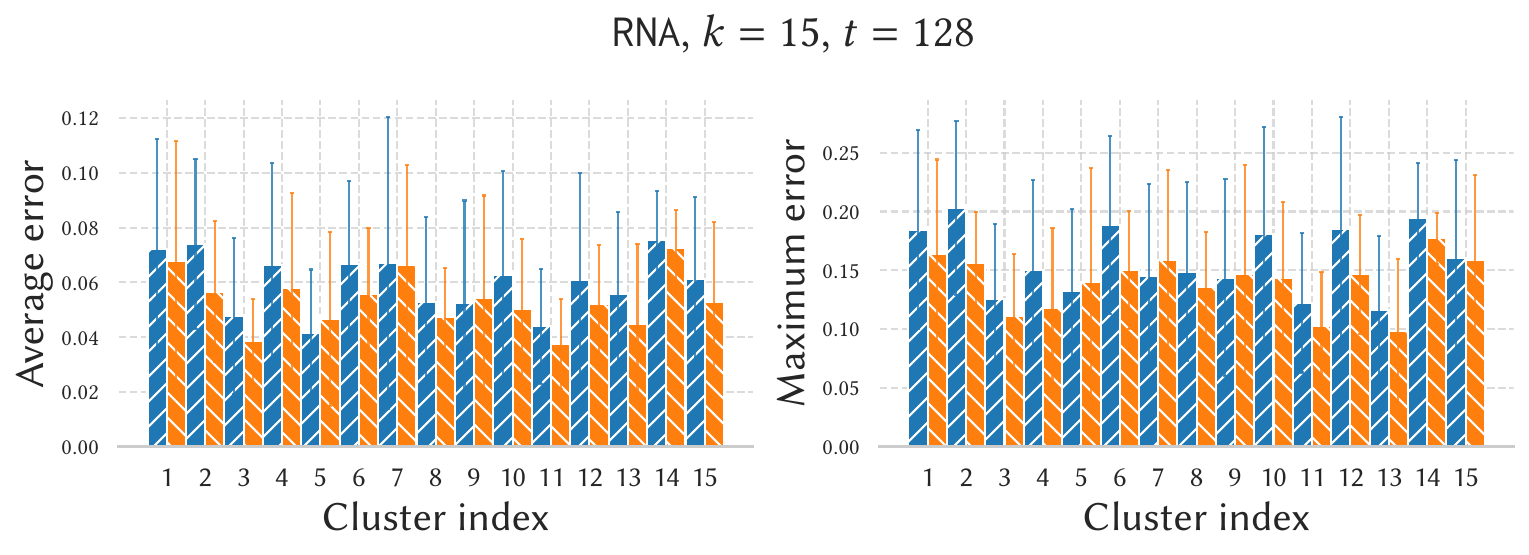}\\
	\end{tabular}
	\caption{For each cluster in \clust we show the average and maximum errors achieved by \algLocal (\PPSbuck) and the uniform sampling based approach from~\Cref{subsubsec:uniFails} (\UNIbuck). Each dot represent an independent run.}
	\label{fig:clustErrMaxandAvgApp}
\end{figure*}

\end{document}